\documentclass{IEEEtran}

\usepackage{graphicx}
\usepackage{algorithm}
\usepackage{amsmath}
\usepackage{amssymb}
\usepackage{cancel}
\usepackage{mathdots}
\usepackage{mathtools}
\usepackage{esint}
\usepackage[version=4]{mhchem}
\usepackage{stackrel}
\usepackage{hyperref}
\hypersetup{
  colorlinks=true,
  linkcolor=blue,
  citecolor=blue,
  urlcolor=blue,
  filecolor=blue
}
\usepackage[caption=false, font=footnotesize]{subfig}
\usepackage{algorithmic}
\usepackage{cite}
\usepackage{xcolor}

\newtheorem{assumption}{Assumption}
\newtheorem{lemma}{Lemma}
\newtheorem{theorem}{Theorem}

\newtheorem{remark}{Remark}

\begin{document}

\title{Fused Constrained Policy Reuse Optimization for Wireless Resource Allocation}
\author{An Liu, Senior Member, IEEE , Zheyuan Zhou and Kexuan Wang \thanks{An Liu, Zheyuan Zhou and Kexuan Wang are with the College of Information Science and Electronic Engineering, Zhejiang University, Hangzhou 310027, China (email: \{anliu, 12331077, kexuanWang\}@zju.edu.cn). (Corresponding Author: An Liu, Zheyuan Zhou)}}
\maketitle
\begin{abstract}
Deep reinforcement learning (DRL) has been widely adopted for wireless resource allocation due to its model-free adaptability. However, online exploration is costly, as randomly initialized policies may violate long-term constraints before sufficient data are collected. Future wireless systems must cope with increasingly dynamic traffic, fluctuating channel conditions, and stringent energy efficiency requirements, demanding algorithms that can learn quickly with minimal environment interactions to reduce both energy consumption and signaling overhead. We develop Fused-CPRO, a knowledge-fused constrained policy reuse optimization method addressing these challenges. Fused-CPRO constructs the allocation policy as a mixture of a learnable target policy, source policies from related scenarios, and domain-knowledge (DK) policies from expert rules, jointly optimizing the target policy and reuse probabilities under a constrained Markov decision process (CMDP). This fusion of heterogeneous priors accelerates convergence and enhances robustness. Constrained stochastic successive convex approximation (CSSCA) handles non-convex objectives and constraints, while a critic trained from mixed offline-online data improves sample efficiency by reusing pre-collected experience. We prove almost-sure convergence to a Karush-Kuhn-Tucker (KKT) point. Simulations on delay-constrained multi-user multiple-input multiple-output (MU-MIMO) power control and Cramer-Rao bound (CRB)-constrained multiple-input multiple-output integrated sensing and communication (MIMO-ISAC) beamforming demonstrate that Fused-CPRO improves empirical performance and converges substantially faster than representative baselines.
\end{abstract}
\begin{IEEEkeywords}
Constrained reinforcement learning, policy reuse, transfer learning, domain knowledge, wireless resource allocation.
\end{IEEEkeywords}
\section{Introduction}
Wireless resource allocation is a fundamental task in communication systems. In modern wireless networks, the network controller needs to allocate radio resources according to environment information such as the current channel state information (CSI), queue state information (QSI), and random traffic arrivals. The allocated resources may include transmit power, precoding-related parameters, and user-priority weights. These decisions influence long-term performance and should satisfy certain constraints or budgets. Future wireless systems face increasingly dynamic traffic patterns, rapidly fluctuating channel conditions, and stringent energy efficiency requirements, making resource allocation more challenging than ever. In such environments, a learning-based allocator must adapt quickly with minimal environment interactions, since each exploratory action incurs energy consumption and signaling overhead that cannot be ignored in practical deployments. Therefore, wireless resource allocation is naturally a long-term constrained decision-making problem that demands fast-converging and sample-efficient learning algorithms.

Model-based and rule-based allocation methods have long been used for such problems. MaxWeight and other queue-aware allocation rules can stabilize traffic-sensitive systems by prioritizing users with large or urgent queues \cite{tassiulas1992stability,stolyar2001largestLWDF}. Lyapunov drift-plus-penalty methods balance queue stability and resource cost through per-slot optimization \cite{neely2010stochastic}. Weighted minimum mean-square error (WMMSE) and regularized zero-forcing (RZF) based designs provide structured and interpretable resource-allocation rules for multi-user MIMO transmission \cite{shi2011wmmse,RZF}. These methods are attractive because their decisions are tied to communication-domain quantities such as queues, channels, interference, and transmit power. However, their performance can be limited by fixed rule structures, modeling assumptions, or parameter choices that do not adapt well across deployment scenarios. Instead of discarding these valuable domain knowledge, a practical learning-based allocator should reuse them as domain-knowledge (DK) policies to accelerate learning and provide reliable fallback decisions.

Deep reinforcement learning (DRL) has been widely studied for radio resource management \cite{chen2020RL4RRM,zangooei2023RL4RRMsurvey}. Constrained RL further allows long-term resource and delay constraints to be included in a constrained Markov decision process (CMDP), and existing methods such as CPO-type methods, SCAOPO, and SLDAC provide principled constrained policy optimization mechanisms \cite{CPO,SCAOPO,wang2024sldac}. Nevertheless, directly training a DRL-based allocation policy in a wireless scenario remains difficult. Online interactions are costly since exploratory actions consume radio resources and may increase delay or violate long-term constraints before the policy has collected sufficient data. Moreover, a randomly initialized DNN policy without policy reuse often has slow early-stage convergence and limited transferability across related communication scenarios.

Transfer learning (TL) and policy reuse offer a method to exploit allocation experience learned in related scenarios \cite{taylor2009transferlearning,fernandez2006policyreuse,fernandez2010policyreuse}. In wireless resource allocation, two types of prior policies are especially useful: source policies trained in related communication scenarios, and DK policies derived from expert allocation rules. The former can provide a warm start for the target scenario, while the latter can provide stable and interpretable fallback decisions in unfamiliar states. However, existing policy-reuse mechanisms often select prior policies through fixed or heuristic reuse probabilities. 

To address this gap, we formulate a Fused Constrained Policy Reuse Optimization (CPRO) problem for wireless resource allocation and develop the Fused-CPRO algorithm to solve it. Fused-CPRO is designed around three key requirements of future wireless resource allocation: (i) fast convergence with minimal online interactions, achieved by adaptively fusing source policies, DK policies, and a learnable target policy into a mixed policy with learnable reuse probabilities, and by leveraging offline data for warm-starting; (ii) principled handling of non-convex stochastic constraints, provided by a constrained stochastic successive convex approximation (CSSCA) actor; and (iii) rigorous reliability guarantees, established through an almost-sure convergence proof to a KKT point. Our main contributions are summarized below.

\begin{itemize}
\item \textbf{A knowledge-fused CMDP formulation for wireless resource allocation:}
We formulate constrained wireless resource allocation as a knowledge-fused CMDP in which a mixed allocation policy is constructed from a learnable target policy, source policies trained in related wireless scenarios, and DK policies derived from model- or rule-based allocation methods. The policy reuse probabilities are optimized together with the target policy, so they serve as adaptive and interpretable weights over heterogeneous allocation priors.

\item \textbf{The Fused-CPRO actor-critic algorithm:}
We develop Fused-CPRO algorithm to solve the proposed formulation. In the actor, constrained stochastic successive convex approximation (CSSCA) jointly updates the target DNN policy and the reuse probabilities, providing principled handling of non-convex stochastic constraints. In the critic, offline data generated by source and DK policies are reused together with recent online samples from the target scenario. This design transfers related allocation experience, injects communication-domain priors, and significantly reduces costly online exploration, enabling fast convergence in dynamic wireless environments.

\item \textbf{Convergence guarantees under heterogeneous policy and data reuse:}
We prove that Fused-CPRO converges to a KKT point under standard regularity and step-size assumptions. The analysis explicitly accounts for the bias induced by policy reuse and offline data reuse, and shows that these bias terms can be controlled without destroying the asymptotic consistency of the actor-critic updates. This theoretical guarantee is essential for deploying learning-based resource allocation in practical systems where reliability is critical.
\end{itemize}
The remainder of the paper is organized as follows: Section~\ref{sec:Preliminaries} introduces a representative MU-MIMO resource allocation model and formulates the constrained policy reuse optimization problem. Sections~\ref{sec:Algorithmic-Framework} and~\ref{sec:CONVERGENCE-ANALYSIS} present the Fused-CPRO algorithm and its theoretical analysis, respectively. Section~\ref{sec:simulation-results} provides simulation results, and Section~\ref{sec:conclusion} concludes the paper.

\section{Constrained Policy Reuse Formulation for Wireless Resource Allocation}
\label{sec:Preliminaries}
This section uses two representative scenarios to illustrate the proposed formulation: delay-constrained downlink MU-MIMO resource allocation and sensing-constrained MIMO-ISAC beamforming. The Fused-CPRO formulation is not restricted to these specific physical-layer models and can be applied to wireless resource allocation tasks that can be represented as CMDPs and equipped with reusable source or domain-knowledge policies.

\subsection{Delay-Constrained Power Control for Downlink MU-MIMO}
\label{subsec:mimo-model}
We consider a downlink MU-MIMO system where a base station (BS) equipped with $N_t$ antennas serves $K$ single-antenna users. At time slot $t$, the BS observes the channel state information (CSI) $\{\mathbf{h}_{k}(t)\}_{k=1}^{K}$ and the queue state information (QSI) $\{Q_{k}(t)\}_{k=1}^{K}$. It then selects a resource-allocation action, such as the transmit-power vector $\mathbf{p}(t)=[p_{1}(t),\ldots,p_{K}(t)]^{T}$ and a precoding-control parameter $\alpha_{Z}(t)$ for regularized zero-forcing (RZF) transmission.

Let $\mathbf{v}_{k}(\alpha_{Z}(t))$ denote the normalized RZF precoder for user $k$ under $\alpha_{Z}(t)$ \cite{RZF}. The downlink rate of user $k$ is
\begin{equation}
R_{k}(t)=B\log_{2}\!\left(1+\frac{p_{k}(t)\left|\mathbf{h}_{k}^{H}(t)\mathbf{v}_{k}(\alpha_{Z}(t))\right|^{2}}{\sum_{j\neq k}p_{j}(t)\left|\mathbf{h}_{k}^{H}(t)\mathbf{v}_{j}(\alpha_{Z}(t))\right|^{2}+\sigma_{k}^{2}}\right),
\label{eq:mimo_rate_model}
\end{equation}
where $B$ is the system bandwidth and $\sigma_k^2$ is the noise power. If $A_k(t)$ denotes the random data arrival rate and $\tau_0$ denotes the slot duration, the queue evolves as
\begin{equation}
Q_{k}(t+1)=\max\!\left\{Q_{k}(t)+A_{k}(t)\tau_{0}-R_{k}(t)\tau_{0},0\right\}.
\label{eq:mimo_queue_model}
\end{equation}

A representative delay-constrained resource-allocation objective is to minimize the long-term average transmit power while satisfying user-wise average delay constraints:
\begin{align}
\underset{\pi}{\mathrm{min}}\quad &\lim_{T\rightarrow\infty}\frac{1}{T}\mathbb{E}_{p_{\pi}}\!\left[\sum_{t=0}^{T-1}\sum_{k=1}^{K}p_{k}(t)\right],\nonumber\\
\mathrm{s.t.}\quad &\lim_{T\rightarrow\infty}\frac{1}{T}\mathbb{E}_{p_{\pi}}\!\left[\sum_{t=0}^{T-1}\frac{Q_{k}(t)}{\lambda_{k}}\right]\leq c_{k},\quad k=1,\ldots,K,
\label{eq:mimo_power_delay_problem}
\end{align}
where $\lambda_k=\mathbb{E}[A_k(t)]$ and $c_k$ is the delay limit of user $k$. 

\subsection{Sum-Rate Maximization for MIMO-ISAC Beamforming with Sensing-Accuracy Constraints}
\label{subsec:isac-model}
We consider a MIMO-ISAC system where a BS equipped with $M$ transmit antennas serves $K$ single-antenna communication users while simultaneously sensing $L$ targets using the same transmit array \cite{huang2024beamformingisac}. The system operates in a target-tracking scenario, where the BS continuously updates estimates of the targets' angles and reflection coefficients over successive frames. We adopt a block-fading model where the communication channel and sensing parameters remain constant within each frame of $N_{\mathrm{symbol}}$ symbols.

At frame $t$, the BS transmits a combination of communication and sensing signals. The transmitted signal matrix $\mathbf{X}(t)\in\mathbb{C}^{M\times N_{\mathrm{symbol}}}$ is given by
\[
\mathbf{X}(t)=\mathbf{W}_{c}(t)\mathbf{S}_{c}(t)+\mathbf{x}_{s}(t)\mathbf{1}^{T},
\]
where $\mathbf{W}_{c}(t)=[\mathbf{w}_{1}(t),\ldots,\mathbf{w}_{K}(t)]\in\mathbb{C}^{M\times K}$ is the communication precoding matrix, $\mathbf{S}_{c}(t)\in\mathbb{C}^{K\times N_{\mathrm{symbol}}}$ contains the i.i.d. communication data symbols with unit energy, and $\mathbf{x}_{s}(t)\in\mathbb{C}^{M}$ is the dedicated sensing signal repeated over the frame. The sensing signal is parameterized by steering vectors toward the predicted target directions
\[
\mathbf{x}_{s}(t)=\sqrt{p_{s}(t)}\frac{\sum_{l=1}^{L}\omega_{l}(t)\mathbf{a}(\hat{\theta}_{l}(t))}{\left\Vert \sum_{l=1}^{L}\omega_{l}(t)\mathbf{a}(\hat{\theta}_{l}(t))\right\Vert },
\]
where $p_{s}(t)$ is the sensing power, $\omega_{l}(t)$ is the beam weight for target $l$, and $\mathbf{a}(\theta)=[1,e^{j\pi\sin\theta},\ldots,e^{j\pi(M-1)\sin\theta}]^{T}$ is the steering vector of a half-wavelength uniform linear array. The communication and sensing functions share a total power budget $p_{s}(t)+\sum_{k=1}^{K}p_{k}(t)\leq P_{\max}$, where $p_{k}(t)=\|\mathbf{w}_{k}(t)\|^{2}$ is the communication power allocated to user $k$.

The state comprises the communication channel matrix $\mathbf{H}(t)=[\mathbf{h}_{1}(t),\ldots,\mathbf{h}_{K}(t)]^{T}$ and the predicted sensing parameters, where $\mathbf{h}_{k}(t)$ is the channel for user $k$. For each target $l$, the BS maintains a prediction of its angle $\hat{\theta}_{l}(t)$ and complex reflection coefficient $\hat{\beta}_{l}(t)$, which can be obtained via dynamic Bayesian learning algorithms such as extended Kalman filtering or dynamic Turbo-CS \cite{huang2024beamformingisac}. The predicted parameter vector is defined as
\[
\hat{\boldsymbol{\eta}}_{l}(t)=\big[\operatorname{Re}\{\hat{\beta}_{l}(t)\},\operatorname{Im}\{\hat{\beta}_{l}(t)\},\hat{\theta}_{l}(t)\big]^{T}.
\]

The action $\boldsymbol{a}^{\mathrm{ISAC}}(t)$ allocates both sensing and communication resources, i.e.,
\[
\boldsymbol{a}^{\mathrm{ISAC}}(t)=\big[\boldsymbol{\omega}^{T}(t),p_{s}(t),p_{1}(t),\ldots,p_{K}(t)\big]^{T},
\]
where $\boldsymbol{\omega}(t)=[\omega_{1}(t),\ldots,\omega_{L}(t)]^{T}$. To reduce the action-space dimensionality, we parameterize the communication precoder as a zero-forcing (ZF) beamformer: $\mathbf{w}_{k}(t)=\sqrt{p_{k}(t)}\mathbf{w}^{\mathrm{ZF}}_{k}(t)/\|\mathbf{w}^{\mathrm{ZF}}_{k}(t)\|$, where $\mathbf{W}^{\mathrm{ZF}}(t)=\mathbf{H}^{H}(t)(\mathbf{H}(t)\mathbf{H}^{H}(t))^{-1}$ \cite{huang2024beamformingisac}. The raw actor output is projected to satisfy the total power constraint.

For communication, since the sensing signal $\mathbf{x}_{s}(t)$ is known at the users after control signaling and channel estimation, its interference can be canceled from the received signal. The signal-to-interference-plus-noise ratio (SINR) of user $k$ is
\[
\mathrm{SINR}_{k}(t)=\frac{\left|\mathbf{h}^{H}_{k}(t)\mathbf{w}_{k}(t)\right|^{2}}{\sum_{j\neq k}\left|\mathbf{h}^{H}_{k}(t)\mathbf{w}_{j}(t)\right|^{2}+\sigma^{2}_{c}},
\]
where $\sigma^{2}_{c}$ is the communication noise power. The sum-rate at frame $t$ is
\begin{equation}
R_{\mathrm{sum}}(t)
=
\sum_{k=1}^{K}
\log_{2}\!\left(
1+
\frac{\left|\mathbf{h}_{k}^{H}(t)\mathbf{w}_{k}(t)\right|^{2}}
{\sum_{j\neq k}\left|\mathbf{h}_{k}^{H}(t)\mathbf{w}_{j}(t)\right|^{2}+\sigma_c^{2}}
\right).
\label{eq:isac_sum_rate}
\end{equation}

For sensing, the targets' parameters are estimated from the echo signal received at the BS. The received sensing signal matrix is
\[
\mathbf{Y}_{s}(t)=\sum_{l=1}^{L}\beta_{l}(t)\mathbf{a}(\theta_{l}(t))\mathbf{a}^{H}(\theta_{l}(t))\mathbf{X}(t)+\mathbf{N}_{s}(t),
\]
where $\beta_{l}(t)$ and $\theta_{l}(t)$ are the true reflection coefficient and angle of target $l$, and $\mathbf{N}_{s}(t)$ is the sensing noise matrix with i.i.d. $\mathcal{CN}(0,\sigma^{2}_{s})$ entries. Let $\boldsymbol{\eta}_{l}(t)$ denote the true parameters reorganized from $\beta_{l}(t)$ and $\theta_{l}(t)$ similar to $\hat{\boldsymbol{\eta}}_{l}(t)$. The Fisher information matrix (FIM) with respect to the unknown parameters $\{\boldsymbol{\eta}_{l}(t)\}_{l=1}^{L}$ quantifies the estimation accuracy. The Cramer-Rao bound (CRB) matrix for target $l$ is given by the corresponding diagonal block of the inverse FIM \cite{huang2024beamformingisac}. Let $\Gamma_{l}(t)$ denote the CRB for $\boldsymbol{\eta}_{l}(t)$, which is a function of $\boldsymbol{\eta}_{l}(t)$ with the detailed expression given in \cite{huang2024beamformingisac}. A smaller sensing cost $\Gamma_{l}(t)$ indicates better sensing accuracy. Since the true parameters $\boldsymbol{\eta}_{l}(t)$ are unknown, we follow \cite{huang2024beamformingisac} and compute the cost using the predicted parameters $\hat{\boldsymbol{\eta}}_{l}(t)$ in place of $\boldsymbol{\eta}_{l}(t)$ during algorithm execution.

The resulting CMDP objective is to maximize the long-term communication sum-rate subject to per-target sensing-accuracy constraints:
\begin{align}
\underset{\pi}{\mathrm{max}}\quad
&\lim_{T\rightarrow\infty}\frac{1}{T}\mathbb{E}_{p_{\pi}}\!\left[
\sum_{t=0}^{T-1}R_{\mathrm{sum}}(t)
\right],\nonumber\\
\mathrm{s.t.}\quad
&\lim_{T\rightarrow\infty}\frac{1}{T}\mathbb{E}_{p_{\pi}}\!\left[
\sum_{t=0}^{T-1}\Gamma_l(t)
\right]\leq \bar{\Gamma}_{l},\quad l=1,\ldots,L.
\label{eq:isac_sumrate_crb_problem}
\end{align}

Note that the action $\boldsymbol{\omega}(t)$ affects the sensing beam directions, which in turn influence the quality of the predicted parameters $\hat{\boldsymbol{\eta}}_{l}(t+1)$ in the next frame. Therefore, this is a CMDP whose state distribution depends on the action, naturally capturing the closed-loop coupling between beamforming decisions and sensing accuracy.

\subsection{Knowledge-Fused CMDP Formulation}
\label{subsec:Problem-Formulation}

\subsubsection{CMDP Representation of Wireless Resource Allocation}
Both examples above fit the same constrained sequential-decision structure. In each slot or frame, the controller observes the scenario-specific wireless state, selects a continuous allocation action, and incurs one objective cost together with several constraint costs. We therefore represent both scenarios using the CMDP framework \cite{CMDP}. More generally, a CMDP is defined by the tuple $(\mathcal{S},\mathcal{A},P,\mathcal{C})$:

\begin{itemize}
\item $\mathcal{S}\subseteq\mathbb{R}^{n_{s}}$: The state space, which contains scenario-specific wireless observations such as CSI, QSI, or sensing-state estimates.
\item $\mathcal{A}\subseteq\mathbb{R}^{n_{a}}$: The action space, which may include transmit powers, beam weights, precoding parameters, and other continuous allocation variables.
\item $P:\mathcal{S}\times\mathcal{A}\times\mathcal{S}\rightarrow[0,1]$: The state transition probability function, where $P(\boldsymbol{s}'\mid\boldsymbol{s},\boldsymbol{a})$ denotes the probability of transitioning to state $\boldsymbol{s}'$ upon taking action $\boldsymbol{a}$ in state $\boldsymbol{s}$.
\item $\mathcal{C}=\{C_{i}\}_{i=0}^{I}$: A set of per-stage cost functions. Here, $C_{0}$ represents the primary objective cost and $\{C_{i}\}_{i=1}^{I}$ represent auxiliary costs to be constrained.
\end{itemize}

A policy $\pi:\mathcal{S}\rightarrow\mathcal{P}(\mathcal{A})$ maps states to probability distributions over actions, with $\pi(\boldsymbol{a}\mid\boldsymbol{s})$ being the probability of selecting action $\boldsymbol{a}$ in state $\boldsymbol{s}$. The policy $\pi$ and the transition dynamics $P$ together induce a probability distribution over trajectories $\tau=\{\boldsymbol{s}_{0},\boldsymbol{a}_{0},\boldsymbol{s}_{1},\boldsymbol{a}_{1},\ldots\}$, denoted by $p_{\pi}$, such that $\boldsymbol{s}_{t+1}\sim P(\cdot\mid\boldsymbol{s}_{t},\boldsymbol{a}_{t})$ and $\boldsymbol{a}_{t}\sim\pi(\cdot\mid\boldsymbol{s}_{t})$.

\subsubsection{Policy Reuse for Heterogeneous Knowledge Integration}
Policy reuse considers $N\geq1$ old policies $\pi_{1},\ldots,\pi_{N}$ and a learnable target policy $\pi_{0}$. The reuse vector $\boldsymbol{\rho}=[\rho_{0},\rho_{1},\ldots,\rho_{N}]^{T}$ lies on the probability simplex, i.e., $\rho_n\geq0$ and $\sum_{n=0}^{N}\rho_n=1$. At each time step, a policy index is sampled according to $\boldsymbol{\rho}$, and the action is then sampled from the selected policy \cite{fernandez2006policyreuse, fernandez2010policyreuse}.

The old-policy pool contains two types of reusable knowledge. The first subset $\{\pi_n\}_{n=1}^{N_1}$ contains source policies trained in related tasks, for example through transfer learning. The second subset $\{\pi_n\}_{n=N_1+1}^{N}$ contains DK policies constructed from domain knowledge, such as rule-based or iterative allocation algorithms. The old policies are frozen, while Fused-CPRO updates the target-policy parameters and the reuse probabilities.

To enable gradient-based optimization and keep a unified representation, we model the source policies and the target policy as Gaussian policies \cite{Gaussianpolicy}. For policy $n\in\{0,1,\ldots,N_{1}\}$, the mean $\boldsymbol{\mu}_{n}$ and the diagonal covariance matrix $\boldsymbol{\Sigma}_{n}$ are output by DNNs parameterized by $\boldsymbol{\psi}_{n}\in\boldsymbol{\Psi}$. Thus, the policy is defined as:
\[
\pi_{n}(\boldsymbol{a}\mid\boldsymbol{s})\propto|\boldsymbol{\Sigma}_{n}|^{-\frac{1}{2}}\exp\left(-\frac{1}{2}(\boldsymbol{\mu}_{n}-\boldsymbol{a})^{\intercal}\boldsymbol{\Sigma}_{n}^{-1}(\boldsymbol{\mu}_{n}-\boldsymbol{a})\right).
\]

DK policies are often deterministic. To incorporate them into the stochastic policy-gradient framework, we apply Gaussian smoothing: each DK policy is represented as $\mathcal{N}(\boldsymbol{\mu}_{n},\epsilon\mathbf{I})$, where $\boldsymbol{\mu}_{n}$ is the deterministic DK action and $\epsilon$ is a small fixed variance. This surrogate gives each DK policy a differentiable density while keeping sampled actions close to the original deterministic output.

\subsubsection{Constrained Policy Reuse Optimization Problem}
The policy-reuse scheme defines a mixed policy $\pi_{\boldsymbol{\theta}}$ parameterized by $\boldsymbol{\theta}=[\boldsymbol{\rho};\boldsymbol{\psi}_{0}]\in\boldsymbol{\Theta}$, where $\boldsymbol{\psi}_{0}$ parameterizes the target policy. The joint optimization of $\boldsymbol{\rho}$ and $\boldsymbol{\psi}_{0}$ is formulated as the following infinite-horizon average-cost constrained problem:
\begin{align}
\underset{\boldsymbol{\theta}\in\boldsymbol{\Theta}}{\mathrm{min}}\quad
J_{0}(\boldsymbol{\theta})
\triangleq &\lim_{T\to\infty}\frac{1}{T}
\mathbb{E}_{p_{\pi_{\boldsymbol{\theta}}}}\!\left[
\sum_{t=0}^{T-1}C_{0}(\boldsymbol{s}_{t},\boldsymbol{a}_{t})
\right]\label{eq:problem 1-1}\\
\mathrm{s.t.}\quad
J_{i}(\boldsymbol{\theta})
\triangleq &\lim_{T\to\infty}\frac{1}{T}
\mathbb{E}_{p_{\pi_{\boldsymbol{\theta}}}}\!\left[
\sum_{t=0}^{T-1}C_{i}(\boldsymbol{s}_{t},\boldsymbol{a}_{t})
\right]-c_{i}\leq0,\nonumber\\
& i=1,\ldots,I,\nonumber
\end{align}
where $p_{\pi_{\boldsymbol{\theta}}}$ denotes the trajectory distribution under the mixed policy $\pi_{\boldsymbol{\theta}}$.

The two wireless examples instantiate these costs directly. For the MU-MIMO model in \eqref{eq:mimo_power_delay_problem}, we set $C_0(\boldsymbol{s}_t,\boldsymbol{a}_t)=\sum_{k=1}^{K}p_k(t)$ and $C_k(\boldsymbol{s}_t,\boldsymbol{a}_t)=Q_k(t)/\lambda_k$ for $k=1,\ldots,K$. For the MIMO-ISAC model in \eqref{eq:isac_sumrate_crb_problem}, we set $C_0(\boldsymbol{s}_t,\boldsymbol{a}_t)=-R_{\mathrm{sum}}(t)$ and $C_l(\boldsymbol{s}_t,\boldsymbol{a}_t)=\Gamma_l(t)$ for $l=1,\ldots,L$. The paper focuses on the infinite-horizon average-cost case; the extension to discounted objectives is straightforward \cite{TwoTimescaleAC, average1, average8}.

\section{Fused-CPRO Algorithm}
\label{sec:Algorithmic-Framework}
\begin{figure}[t]
\noindent
\centering

\includegraphics[width=\columnwidth]{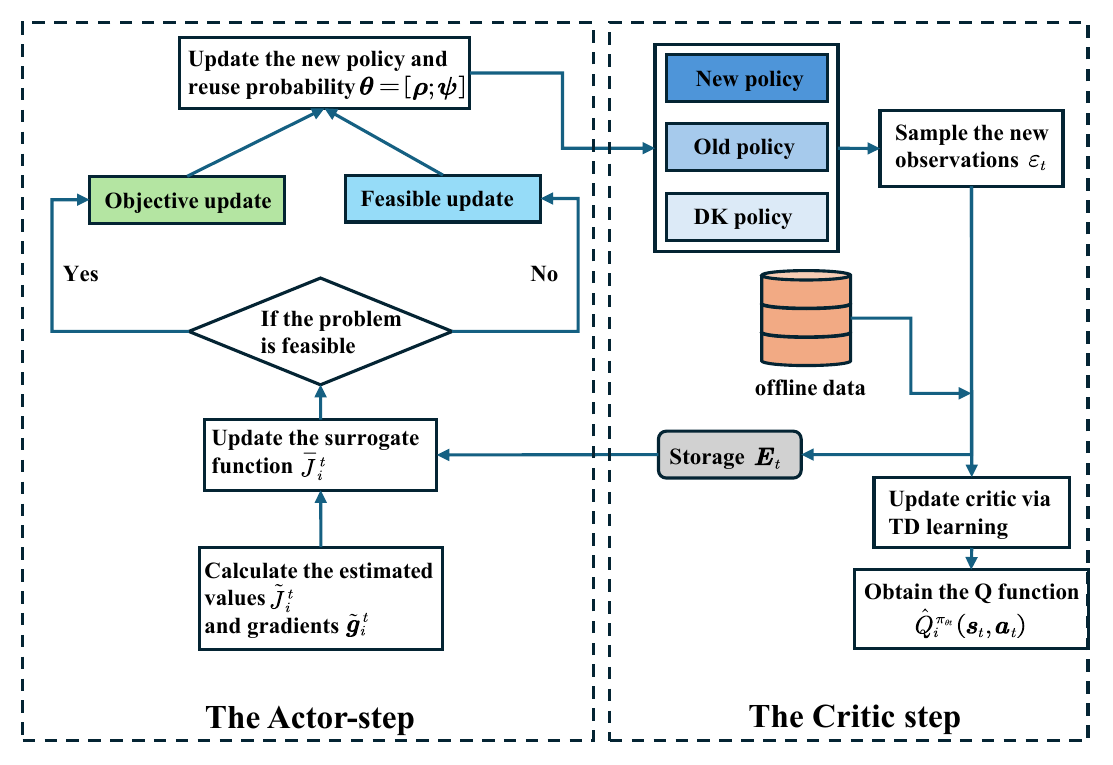}

\caption{The framework of the proposed Fused-CPRO algorithm}\label{fig:diagram of the algorithm}
\end{figure}

This section presents the proposed Fused-CPRO algorithm. As illustrated in Fig.~\ref{fig:diagram of the algorithm}, Fused-CPRO follows a single-loop actor-critic structure: the critic estimates the state-action value functions (Q-functions), while the actor updates both the target policy and the reuse probabilities. The algorithm is designed to operate efficiently in wireless resource allocation scenarios where online interactions incur non-negligible signaling overhead and energy costs. By exploiting offline datasets collected from source and DK policies, Fused-CPRO reduces the number of online samples required before reaching a feasible and high-performing allocation policy. At iteration $t$, the current mixed policy collects a new online sample $\varepsilon_t$ and stores it in the replay buffer. Both critic and actor updates then use a mixed offline-online data reuse scheme to reduce online interaction cost and improve sample efficiency. We first describe the mixed data collection procedure, then derive the policy-reuse gradient estimator, and finally present the critic and actor updates. The full procedure is summarized in Algorithm~\ref{alg:the-Single-Loop-Deep}.

\subsection{Mixed Offline-Online Data Collection}
Fused-CPRO uses a mixed offline-online data collection scheme. We first define the adjusted costs as $C_{0}^{\prime}(\boldsymbol{s},\boldsymbol{a})=C_{0}(\boldsymbol{s},\boldsymbol{a})$ and $C_{i}^{\prime}(\boldsymbol{s},\boldsymbol{a})=C_{i}(\boldsymbol{s},\boldsymbol{a})-c_{i}$ for $i=1,\ldots,I$. In wireless resource allocation, these costs correspond to physical quantities such as transmit power, queue backlog, or sensing error, and the constants $c_i$ represent the operational constraints (e.g., delay limits or CRB thresholds) that the allocation policy must satisfy. The data used by the algorithm come from two sources.

\begin{itemize}
\item Offline datasets ($\mathcal{D}_{n}$):
The agent has access to $N$ pre-collected datasets. Each dataset $\mathcal{D}_{n}$ contains transitions $\left\{ \boldsymbol{s},\boldsymbol{a},\{C_{i}^{\prime}(\boldsymbol{s},\boldsymbol{a})\}_{i=0}^{I},\boldsymbol{s}'\right\}$ generated by executing the $n$-th old policy $\pi_{n}$ in an environment related to the current target task. For instance, these may include trajectories from a source policy trained under different channel distributions or a DK policy derived from a queue-aware power-allocation rule.

\item Online dataset ($\mathcal{D}_{0}^{t}$):
The agent also interacts with the current environment using the latest mixed policy $\pi_{\boldsymbol{\theta}_{t}}$. At iteration $t$, it obtains a new online sample $\varepsilon_{t}=\{\boldsymbol{s}_{t},\boldsymbol{a}_{t},\{C_{i}^{\prime}(\boldsymbol{s}_{t},\boldsymbol{a}_{t})\}_{i=0}^{I},\boldsymbol{s}_{t+1}\}$ and maintains a replay buffer $\mathcal{D}_{0}^{t}=\{\varepsilon_{t-T_{t}+1},\ldots,\varepsilon_{t}\}$ containing the latest $T_{t}$ samples. In a typical wireless deployment, each online sample corresponds to one scheduling frame or time slot, and its acquisition consumes radio resources and time.
\end{itemize}
For clarity, the algorithm description and convergence analysis use one online sample and one critic update per iteration; in implementation, a mini-batch of $B$ online samples can be used for $q$ critic updates with $B\geq q\geq1$. The parameters $B$ and $q$ control the trade-off between convergence speed and environment interaction cost: a larger $q$ improves per-sample critic learning at the expense of more gradient computations, while a larger $B$ increases the online data collection burden per iteration. In resource-constrained wireless settings, this flexibility allows the operator to balance learning efficiency against the signaling and energy cost of online exploration.

\subsection{Policy Reuse Gradient with Mixed Data}
Let $\pi_{\boldsymbol{\psi}_{0}}=\pi_{0}$ denote the new, learnable policy. The mixed policy is
$\pi_{\boldsymbol{\theta}}(\boldsymbol{a}|\boldsymbol{s})=\sum_{n=0}^{N}\rho_{n}\pi_{n}(\boldsymbol{a}|\boldsymbol{s})$, where $\boldsymbol{\theta}=[\boldsymbol{\rho};\boldsymbol{\psi}_{0}]$. The reuse probability $\rho_n$ reflects the algorithm's current confidence in each prior policy for the target resource allocation task, and these probabilities are adaptively updated as more online experience is gathered. The goal is to estimate the gradients of the objective and constraint functions $J_{i}(\boldsymbol{\theta})$ with respect to $\boldsymbol{\theta}$.

Following the policy gradient theorem, the gradient can be expressed as:
\begin{equation}
\begin{aligned}
\nabla_{\boldsymbol{\theta}}J_{i}(\boldsymbol{\theta})
& =\mathbb{E}_{\sigma_{\pi_{\boldsymbol{\theta}}}}\left[Q_{i}^{\pi_{\boldsymbol{\theta}}}(\boldsymbol{s},\boldsymbol{a})\nabla_{\boldsymbol{\theta}}\log\pi_{\boldsymbol{\theta}}(\boldsymbol{a}|\boldsymbol{s})\right],
\end{aligned}
\end{equation}
where $\sigma_{\pi_{\boldsymbol{\theta}}}$ is the stationary state-action distribution and $Q_{i}^{\pi_{\boldsymbol{\theta}}}$ is the state-action value function. Applying this identity to the mixed policy gives the gradients for the reuse probabilities and the target-policy parameters:

\begin{equation}
\begin{aligned}
\nabla_{\rho_{n}}J_{i}(\boldsymbol{\theta}) & =\mathbb{E}_{\sigma_{\pi_{\boldsymbol{\theta}}}}\left[\frac{Q_{i}^{\pi_{\boldsymbol{\theta}}}(\boldsymbol{s},\boldsymbol{a})\pi_{n}(\boldsymbol{a}|\boldsymbol{s})}{\sum_{k=0}^{N}\rho_{k}\pi_{k}(\boldsymbol{a}|\boldsymbol{s})}\right],\quad n=0,\ldots,N,\\
\nabla_{\boldsymbol{\psi}_{0}}J_{i}(\boldsymbol{\theta}) & =\mathbb{E}_{\sigma_{\pi_{\boldsymbol{\theta}}}}\left[\frac{Q_{i}^{\pi_{\boldsymbol{\theta}}}(\boldsymbol{s},\boldsymbol{a})\rho_{0}\nabla_{\boldsymbol{\psi}_{0}}\pi_{\boldsymbol{\psi}_{0}}(\boldsymbol{a}|\boldsymbol{s})}{\sum_{k=0}^{N}\rho_{k}\pi_{k}(\boldsymbol{a}|\boldsymbol{s})}\right].
\end{aligned}
\end{equation}

The true Q-functions and the expectation under $\sigma_{\pi_{\boldsymbol{\theta}}}$ are unknown. We therefore approximate each $Q_{i}^{\pi_{\boldsymbol{\theta}}}$ using dual critic DNNs $f(\boldsymbol{\omega}^{i};\boldsymbol{s},\boldsymbol{a})$ and $f(\bar{\boldsymbol{\omega}}^{i};\boldsymbol{s},\boldsymbol{a})$, whose updates are detailed in Section~\ref{subsec:critic-module}, and estimate the expectation using samples from the mixed offline-online data. This gives the mixed-data gradient estimator at iteration $t$:

\begin{equation}
\tilde{\boldsymbol{g}}_{i}^{t}
=\left[\tilde{g}_{i,\rho_{0}}^{t};\tilde{g}_{i,\rho_{1}}^{t};\ldots;\tilde{g}_{i,\rho_{N}}^{t};\tilde{\boldsymbol{g}}_{i,\boldsymbol{\psi}_{0}}^{t}\right]\label{eq:gtlt}
\end{equation}
whose components for the reuse probabilities and target-policy parameters are given by
\begin{align}
\tilde{g}_{i,\rho_{n}}^{t}= & \xi_{t}\hat{\mathbb{E}}_{\mathcal{D}_{\text{off}}^{t}}\left[\frac{f\left(\bar{\boldsymbol{\omega}}_{t}^{i};\boldsymbol{s},\boldsymbol{a}\right)\pi_{n}\left(\boldsymbol{a}\mid\boldsymbol{s}\right)}{\sum_{k=0}^{N}\rho_{k}\pi_{k}\left(\boldsymbol{a}\mid\boldsymbol{s}\right)}\right]\label{eq:g-tilde-1}\\
& +\left(1-\xi_{t}\right)\hat{\mathbb{E}}_{\mathcal{D}_{0}^{t}}\left[\frac{f\left(\bar{\boldsymbol{\omega}}_{t}^{i};\boldsymbol{s},\boldsymbol{a}\right)\pi_{n}\left(\boldsymbol{a}\mid\boldsymbol{s}\right)}{\sum_{k=0}^{N}\rho_{k}\pi_{k}\left(\boldsymbol{a}\mid\boldsymbol{s}\right)}\right],\nonumber
\end{align}
for $n=0,\ldots,N$, and
\begin{align}
\tilde{\boldsymbol{g}}_{i,\boldsymbol{\psi}_{0}}^{t}= & \xi_{t}\hat{\mathbb{E}}_{\mathcal{D}_{\text{off}}^{t}}\left[\frac{f\left(\bar{\boldsymbol{\omega}}_{t}^{i};\boldsymbol{s},\boldsymbol{a}\right)\rho_{0}\nabla_{\boldsymbol{\psi}_{0}}\pi_{\boldsymbol{\psi}_{0}}\left(\boldsymbol{a}\mid\boldsymbol{s}\right)}{\sum_{k=0}^{N}\rho_{k}\pi_{k}\left(\boldsymbol{a}\mid\boldsymbol{s}\right)}\right]\label{eq:g-tilde}\\
& +\left(1-\xi_{t}\right)\hat{\mathbb{E}}_{\mathcal{D}_{0}^{t}}\left[\frac{f\left(\bar{\boldsymbol{\omega}}_{t}^{i};\boldsymbol{s},\boldsymbol{a}\right)\rho_{0}\nabla_{\boldsymbol{\psi}_{0}}\pi_{\boldsymbol{\psi}_{0}}\left(\boldsymbol{a}\mid\boldsymbol{s}\right)}{\sum_{k=0}^{N}\rho_{k}\pi_{k}\left(\boldsymbol{a}\mid\boldsymbol{s}\right)}\right].\nonumber
\end{align}
Here, $\hat{\mathbb{E}}_{\mathcal{D}}$ denotes the sample mean over dataset $\mathcal{D}$. The mixed offline dataset $\mathcal{D}_{\text{off}}^{t}$ is constructed only from old-policy datasets: each of its $T_{\text{off}}^{t}$ samples is drawn from $\mathcal{D}_{n}$, $n=1,\ldots,N$, with probability $\rho_{n}/\sum_{k=1}^{N}\rho_{k}$. The decreasing weight $\xi_{t}$ balances offline and online samples. Correspondingly, the function value $J_{i}(\boldsymbol{\theta}_{t})$ is estimated using the same mixture as
\begin{align}
\tilde{J}_{i}^{t} & =\xi_{t}\hat{\mathbb{E}}_{\mathcal{D}_{\text{off}}^{t}}\left[C_{i}^{\prime}(\boldsymbol{s},\boldsymbol{a})\right]+\left(1-\xi_{t}\right)\hat{\mathbb{E}}_{\mathcal{D}_{0}^{t}}\left[C_{i}^{\prime}(\boldsymbol{s},\boldsymbol{a})\right],\forall i.\label{eq:J-new}
\end{align}
This estimator uses offline data from old policies to warm-start learning and recent online samples to reduce variance, thereby improving early-stage adaptation without extra environment interactions. In wireless resource allocation, this is particularly beneficial: the critic can form reasonable value estimates for actions near those suggested by source or DK policies (e.g., a queue-aware power allocation) before the target policy has been sufficiently trained, enabling the allocator to avoid severely suboptimal or constraint-violating decisions during early online exploration.

\subsection{Critic Module}
\label{subsec:critic-module}
The critic module estimates the state-action value functions (Q-functions) used in the policy-gradient estimator in \eqref{eq:gtlt}. Since the true average costs $J_{i}(\boldsymbol{\theta}_{t})$ are unknown during learning, we define surrogate Q-functions $\{\hat{Q}_{i}^{\pi_{\boldsymbol{\theta}_{t}}}\}_{i=0}^{I}$ by replacing $J_{i}(\boldsymbol{\theta}_{t})$ with the running estimate $\hat{J}_{i}^{t}$:
\begin{align}
\hat{Q}_{i}^{\pi_{\boldsymbol{\theta}_{t}}}\left(\boldsymbol{s},\boldsymbol{a}\right)= & \mathbb{E}_{p_{t}}\Bigl[\sum_{l=0}^{\infty}\Bigl(C_{i}^{\prime}(\boldsymbol{s}_{l},\boldsymbol{a}_{l})-\hat{J}_{i}^{t}\Bigr)\nonumber \\
& \bigl|\boldsymbol{s}_{0}=\boldsymbol{s},\boldsymbol{a}_{0}=\boldsymbol{a}\Bigr],\forall i,\boldsymbol{s},\boldsymbol{a},\label{eq:definition of Qhat}
\end{align}
where $p_{t}\triangleq p_{\pi_{\boldsymbol{\theta}_{t}}}$ denotes the trajectory distribution under the current policy. We approximate these surrogate Q-functions with two sets of critic networks: primary critics for TD learning and target critics for the actor-gradient estimator.

\subsubsection{Primary Critic Networks}

The primary critics $\{f(\boldsymbol{\omega}^{i};\boldsymbol{s},\boldsymbol{a})\}_{i=0}^{I}$ are trained to approximate $\hat{Q}_{i}^{\pi_{\boldsymbol{\theta}_{t}}}$. Their parameters $\boldsymbol{\omega}^{i}$ are updated by minimizing the mean-squared Bellman error (MSBE):
\begin{equation}
\underset{\boldsymbol{\omega}_{t}^{i}}{\mathrm{min}}\ \mathbb{E}_{\sigma_{t}}\Bigl[\Bigl(f\left(\boldsymbol{\omega}_{t}^{i};\boldsymbol{s}_{t},\boldsymbol{a}_{t}\right)-\mathcal{T}_{t}\ f\left(\boldsymbol{\omega}_{t}^{i};\boldsymbol{s}_{t},\boldsymbol{a}_{t}\right)\Bigr)^{2}\Bigr],\forall i,\label{eq:MSBE}
\end{equation}
where $\sigma_{t}\triangleq\sigma_{\pi_{\boldsymbol{\theta}_{t}}}$ is the stationary state-action distribution, and $\mathcal{T}_{t}$ is the Bellman operator for the current policy defined as
\begin{align}
\mathcal{T}_{t}\ f\left(\boldsymbol{\omega}_{t}^{i};\boldsymbol{s}_{t},\boldsymbol{a}_{t}\right)= & \mathbb{E}_{p_{t}}\Bigl[f\left(\boldsymbol{\omega}_{t}^{i};\boldsymbol{s}_{t+1},\boldsymbol{a}_{t+1}'\right)\Bigr]\\
& +C_{i}^{\prime}(\boldsymbol{s}_{t},\boldsymbol{a}_{t})-\hat{J}_{i}^{t},\forall i,\nonumber
\end{align}
Here, $\boldsymbol{a}_{t+1}'$ is the action selected by $\pi_{\boldsymbol{\theta}_{t}}$ at the next state $\boldsymbol{s}_{t+1}$. The critic update uses a mixed TD estimator built from the online replay buffer $\mathcal{D}_{0}^{t}$ and the sampled offline set $\mathcal{D}_{\text{off}}^{t}$. To keep the critic parameters bounded, we define
\[
\mathbb{B}(\boldsymbol{\omega}_{0}^{i},R_{\boldsymbol{\omega}})=\{\boldsymbol{\omega}^{i}:\|\boldsymbol{\omega}^{i}-\boldsymbol{\omega}_{0}^{i}\|\leq R_{\boldsymbol{\omega}}\}
\]
around the random initialization $\boldsymbol{\omega}_{0}^{i}$ and update the primary critic by projected TD learning:
\begin{equation}
\boldsymbol{\omega}_{t}^{i}=\Pi_{\Omega_{i}}\bigl(\boldsymbol{\omega}_{t-1}^{i}-\eta_{t}\boldsymbol{\Delta}_{i}^{\boldsymbol{\omega}_{t-1}^{i}}\bigr),\forall i,\label{eq:TD-learning}
\end{equation}
where $\{\eta_{t}\}$ is a decreasing step-size sequence satisfying Assumption~\ref{assumption:step_size}, $\Pi_{\Omega_{i}}$ projects the parameter onto the constraint set $\Omega_{i}\triangleq\mathbb{B}(\boldsymbol{\omega}_{0}^{i},R_{\boldsymbol{\omega}})$, and the TD gradient is computed from a mixed offline-online replay estimator:
\begin{equation}
\boldsymbol{\Delta}_{i}^{\boldsymbol{\omega}_{t-1}^{i}}=\left(1-\xi_{t}\right)\boldsymbol{\Delta}_{i,\mathrm{on}}^{\boldsymbol{\omega}_{t-1}^{i}}+\xi_{t}\boldsymbol{\Delta}_{i,\mathrm{off}}^{\boldsymbol{\omega}_{t-1}^{i}},\forall i,\label{eq:gradient term}
\end{equation}
with the online and offline TD-gradient components given by
\begin{equation}\label{eq:gradient term on}
\begin{aligned}
\boldsymbol{\Delta}_{i,\mathrm{on}}^{\boldsymbol{\omega}_{t-1}^{i}}
&=\hat{\mathbb{E}}_{\mathcal{D}_{0}^{t}}\Big[\Big(f\left(\boldsymbol{\omega}_{t-1}^{i};\boldsymbol{s},\boldsymbol{a}\right)-\big(C_{i}^{\prime}(\boldsymbol{s},\boldsymbol{a})-\hat{J}_{i}^{t-1}\\
&\qquad+f\left(\boldsymbol{\omega}_{t-1}^{i};\boldsymbol{s}',\boldsymbol{a}'\right)\big)\Big)\nabla_{\boldsymbol{\omega}}f\left(\boldsymbol{\omega}_{t-1}^{i};\boldsymbol{s},\boldsymbol{a}\right)\Big],
\end{aligned}
\end{equation}

\begin{equation}\label{eq:gradient term off}
\begin{aligned}
\boldsymbol{\Delta}_{i,\mathrm{off}}^{\boldsymbol{\omega}_{t-1}^{i}}
&=\hat{\mathbb{E}}_{\mathcal{D}_{\mathrm{off}}^{t}}\Big[\Big(f\left(\boldsymbol{\omega}_{t-1}^{i};\boldsymbol{s},\boldsymbol{a}\right)-\big(C_{i}^{\prime}(\boldsymbol{s},\boldsymbol{a})-\hat{J}_{i}^{t-1}\\
&\qquad+f\left(\boldsymbol{\omega}_{t-1}^{i};\boldsymbol{s}',\boldsymbol{a}'\right)\big)\Big)\nabla_{\boldsymbol{\omega}}f\left(\boldsymbol{\omega}_{t-1}^{i};\boldsymbol{s},\boldsymbol{a}\right)\Big],
\end{aligned}
\end{equation}
The decaying offline weight $\xi_{t}$ suppresses offline-induced bias asymptotically while retaining its warm-start effect in the early stage.

\subsubsection{Target Critic Networks}
Using the rapidly updated primary critics directly in the actor gradient \eqref{eq:gtlt} can increase variance. Fused-CPRO therefore maintains target critics $\{f(\bar{\boldsymbol{\omega}}^{i};\boldsymbol{s},\boldsymbol{a})\}_{i=0}^{I}$ whose parameters are updated by a slow-moving average of the primary critic parameters:
\begin{equation}
\bar{\boldsymbol{\omega}}_{t}^{i}=\left(1-\gamma_{t}\right)\bar{\boldsymbol{\omega}}_{t-1}^{i}+\gamma_{t}\boldsymbol{\omega}_{t}^{i},\forall i,\label{eq:w^_}
\end{equation}
initialized with $\bar{\boldsymbol{\omega}}_{0}^{i}=\boldsymbol{\omega}_{0}^{i}$. Here, $\{\gamma_{t}\}$ is a decreasing sequence. The target critics $f(\bar{\boldsymbol{\omega}}_{t}^{i};\boldsymbol{s},\boldsymbol{a})$ are used in \eqref{eq:g-tilde-1} and \eqref{eq:g-tilde} to compute the actor-gradient estimates.

\subsection{Actor Module}
The actor module updates the joint policy parameter $\boldsymbol{\theta}=[\boldsymbol{\rho};\boldsymbol{\psi}_{0}]$ in the constrained optimization problem \eqref{eq:problem 1-1}. Because $J_{0}(\boldsymbol{\theta})$ and $\{J_i(\boldsymbol{\theta})\}_{i=1}^{I}$ are non-convex stochastic functions without closed forms, Fused-CPRO adopts the CSSCA framework \cite{CSSCA}. At each iteration, CSSCA builds convex local surrogates from smoothed function-value and gradient estimates, solves the resulting surrogate problem, and moves the actor toward the surrogate solution.

The CSSCA actor proceeds as follows.

\textbf{Step 1: Construct Convex Surrogate Functions.} For each $i=0,\ldots,I$, we construct a convex quadratic surrogate $\bar{J}_{i}^{t}(\boldsymbol{\theta})$ around the current iterate $\boldsymbol{\theta}_{t}$:
\begin{equation}
\bar{J}_{i}^{t}\left(\boldsymbol{\theta}\right)=\hat{J}_{i}^{t}+\left(\hat{\boldsymbol{g}}_{i}^{t}\right)^{T}\left(\boldsymbol{\theta}-\boldsymbol{\theta}_{t}\right)+\zeta_{i}\left\Vert \boldsymbol{\theta}-\boldsymbol{\theta}_{t}\right\Vert _{2}^{2},\forall i,\label{eq:surrogate functions}
\end{equation}
where $\zeta_{i}>0$ controls the quadratic regularization. The terms $\hat{J}_{i}^{t}$ and $\hat{\boldsymbol{g}}_{i}^{t}$ are smoothed estimates of $J_{i}(\boldsymbol{\theta}_{t})$ and $\nabla J_{i}(\boldsymbol{\theta}_{t})$, obtained by averaging the instantaneous estimates $\tilde{J}_{i}^{t}$ in \eqref{eq:J-new} and $\tilde{\boldsymbol{g}}_{i}^{t}$ in \eqref{eq:gtlt}:
\begin{equation}
\hat{J}_{i}^{t}=\left(1-\alpha_{t}\right)\hat{J}_{i}^{t-1}+\alpha_{t}\tilde{J}_{i}^{t},\forall i,\label{eq:J-average}
\end{equation}

\begin{equation}
\hat{\boldsymbol{g}}_{i}^{t}=\left(1-\alpha_{t}\right)\hat{\boldsymbol{g}}_{i}^{t-1}+\alpha_{t}\tilde{\boldsymbol{g}}_{i}^{t},\forall i,\label{eq:g-average}
\end{equation}
where $\{\alpha_{t}\}$ is a decreasing step-size sequence satisfying Assumption~\ref{assumption:step_size}. This averaging reduces the variance of the stochastic estimates used in the surrogate functions. Together with the mixed offline-online estimators in \eqref{eq:g-tilde-1}--\eqref{eq:J-new}, it stabilizes the actor update.

\textbf{Step 2: Solve the Surrogate Optimization Problem.} Using the convex surrogates $\{\bar{J}_{i}^{t}(\boldsymbol{\theta})\}$, the actor solves
\begin{align}
\bar{\boldsymbol{\theta}}_{t}=\underset{\boldsymbol{\theta}\in\boldsymbol{\Theta}}{\textrm{argmin}}\  & \bar{J}_{0}^{t}\left(\boldsymbol{\theta}\right)\label{eq:objective update}\\
\textrm{s.t.}\  & \bar{J}_{i}^{t}\left(\boldsymbol{\theta}\right)\leq0,i=1,\cdots,I.\nonumber
\end{align}
If \eqref{eq:objective update} is infeasible, Fused-CPRO applies a feasibility-restoration step by introducing a slack variable $y$:
\begin{align}
\bar{\boldsymbol{\theta}}_{t}=\underset{\boldsymbol{\theta}\in\boldsymbol{\Theta},y}{\textrm{argmin}}\  & y\label{eq:feasible update}\\
\textrm{s.t.}\  & \bar{J}_{i}^{t}\left(\boldsymbol{\theta}\right)\leq y,i=1,\cdots,I.\nonumber
\end{align}
This problem minimizes the maximum surrogate constraint violation and always admits a feasible solution. When \eqref{eq:objective update} is feasible, the solution of \eqref{eq:feasible update} yields $y\leq0$ and is consistent with the feasible surrogate solution.

The objective update in \eqref{eq:objective update} is a strongly convex QP because of the quadratic regularization, while \eqref{eq:feasible update} is a convex feasibility-restoration QP. Both can be solved efficiently with standard convex optimization tools or Lagrange-dual methods.

\textbf{Step 3: Update the Policy Parameter.} The actor then moves from the current iterate toward the surrogate solution:
\begin{equation}
\boldsymbol{\theta}_{t+1}=(1-\beta_{t})\boldsymbol{\theta}_{t}+\beta_{t}\bar{\boldsymbol{\theta}}_{t}.\label{eq:theta update}
\end{equation}
where $\{\beta_{t}\}$ is a decreasing step-size sequence satisfying Assumption~\ref{assumption:step_size}. Since $\boldsymbol{\Theta}$ is convex, this update keeps $\boldsymbol{\theta}_{t+1}$ in the parameter set whenever $\boldsymbol{\theta}_{t}$ and $\bar{\boldsymbol{\theta}}_{t}$ belong to $\boldsymbol{\Theta}$. The convergence properties of this CSSCA actor update are analyzed in Section~\ref{sec:CONVERGENCE-ANALYSIS}.

\begin{algorithm}
\caption{Fused-CPRO Algorithm}\label{alg:the-Single-Loop-Deep}
\begin{algorithmic}[1]
\STATE \textbf{Input:} decreasing sequences $\{\xi_{t}\}$, $\{\alpha_{t}\}$, $\{\beta_{t}\}$, $\{\eta_{t}\}$, and $\{\gamma_{t}\}$.
\STATE Randomly initialize $\boldsymbol{\omega}_{0}^{i}$ and $\boldsymbol{\psi}_{0}$ from $\mathcal{N}(0,1/m^{2})$; initialize $\boldsymbol{\rho}_{0}$ uniformly over the $N+1$ policies.
\FOR{$t=0,1,\cdots$}
\STATE Obtain $T_{\text{off}}^{t}$ offline samples $\mathcal{D}_{\text{off}}^{t}$ by sampling from $D_{1},\ldots,D_{N}$ according to $\boldsymbol{\rho}_{t}$.
\STATE Sample the new observation $\varepsilon_{t}$ and update the online dataset $\mathcal{D}_{0}^{t}$.
\STATE \textbf{Critic step:} sample mixed TD transitions from $\mathcal{D}_{\text{off}}^{t}$ and $\mathcal{D}_{0}^{t}$, then update $\boldsymbol{\omega}_{t}^{i}$ and $\bar{\boldsymbol{\omega}}_{t}^{i}$ by (\ref{eq:TD-learning})--(\ref{eq:gradient term off}) and (\ref{eq:w^_}), respectively.
\STATE \textbf{Actor step:} calculate $\hat{J}_{i}^{t}$ via (\ref{eq:J-new}) and (\ref{eq:J-average}).
\STATE Estimate $\hat{\boldsymbol{g}}_{i}^{t}$ via (\ref{eq:g-tilde}) and (\ref{eq:g-average}).
\STATE Update $\{\bar{J}_{i}^{t}(\boldsymbol{\theta})\}_{i=0,\ldots,I}$ via (\ref{eq:surrogate functions}).
\IF{Problem (\ref{eq:objective update}) is feasible}
\STATE Solve (\ref{eq:objective update}) to obtain $\bar{\boldsymbol{\theta}}_{t}$.
\ELSE
\STATE Solve (\ref{eq:feasible update}) to obtain $\bar{\boldsymbol{\theta}}_{t}$.
\ENDIF
\STATE Update $\boldsymbol{\theta}_{t+1}$ according to (\ref{eq:theta update}).
\ENDFOR
\end{algorithmic}
\end{algorithm}

\subsection{Relation to Prior Constrained RL and Policy Reuse Methods}
Fused-CPRO is related to two lines of methods: CSSCA-based constrained RL and policy reuse. SCAOPO \cite{SCAOPO} uses CSSCA to update an actor under long-term constraints, while SLDAC \cite{wang2024sldac} adds critic-based TD learning to reduce the variance of policy-gradient estimation. These methods provide convergence-aware constrained policy optimization mechanisms, but they optimize a single target policy and do not explicitly reuse source or DK policies. In contrast, HRL \cite{zhang2025hrl} introduces policy reuse for source and DK policies, but it remains actor-only and does not train a critic from mixed offline-online data.

Fused-CPRO combines these two directions. It keeps the CSSCA actor update and critic-based value estimation from the constrained actor-critic line, while introducing an explicit reuse probability vector over the target policy, source policies, and DK policies. The offline datasets generated by old policies are also used in critic training, so policy reuse affects not only action selection but also value estimation. This design gives two concrete differences from SLDAC and HRL.

\begin{enumerate}
\item \textbf{Joint Optimization of Policy and Reuse Probabilities:}
SLDAC optimizes only a single new policy, whereas HRL reuses prior policies without a critic module. Fused-CPRO jointly updates the target-policy parameters and the reuse probability vector $\boldsymbol{\rho}$ over all candidate policies. Thus, policy selection becomes part of the constrained actor update rather than a fixed or heuristic reuse rule. In wireless resource allocation, this allows the algorithm to adaptively trade off the learnable target policy, source policies trained in related scenarios, and DK policies derived from communication-domain rules.

\item \textbf{Systematic Offline Data Reuse:}
Beyond reusing recent online samples, Fused-CPRO trains the critic with offline datasets generated by the old policies. These offline samples provide early value-estimation information before extensive online interaction and help the critic evaluate actions near reusable source or DK behaviors. The decaying weight $\xi_t$ then reduces the influence of offline distribution mismatch as online samples from the target scenario accumulate. This differs from actor-only policy reuse, where reused policies affect action generation but do not directly provide critic-side TD information.
\end{enumerate}

These two components also introduce additional analytical issues. The mixed policy produces mixture-ratio terms in the policy-gradient estimator, and offline data reuse introduces distribution mismatch between old-policy samples and the current target policy. The next section makes these issues explicit and analyzes their effect on convergence.

\section{Convergence and Performance Analysis}
\label{sec:CONVERGENCE-ANALYSIS}
The convergence analysis has three parts: a critic-tracking bound for the surrogate $Q$-functions, asymptotic consistency of the function-value and gradient surrogates used by the actor, and the resulting KKT convergence of Algorithm~\ref{alg:the-Single-Loop-Deep}.

\subsection{Key Assumptions and New Technical Challenges}

Compared with single-policy methods such as \cite{wang2024sldac}, Fused-CPRO introduces two new sources of estimation bias arising from its knowledge-fusion design, both of which must be explicitly controlled in the convergence proof.

\begin{itemize}
\item \textbf{Policy reuse optimization:}
The mixed policy $\pi_{\boldsymbol{\theta}}(\boldsymbol{a}\mid\boldsymbol{s})=\sum_{n=0}^{N}\rho_{n}\pi_{n}(\boldsymbol{a}\mid\boldsymbol{s})$ induces policy-gradient terms that involve mixture ratios of the form $\pi_{n}(\boldsymbol{a}\mid\boldsymbol{s})/\pi_{\boldsymbol{\theta}}(\boldsymbol{a}\mid\boldsymbol{s})$. To ensure numerical stability and avoid ill-conditioned gradient estimates, these ratios are required to be uniformly bounded.
\item \textbf{Offline dataset reuse:}
The mixed-data estimators introduced in equations \eqref{eq:g-tilde-1}--\eqref{eq:J-new} combine online samples with offline samples collected from distributions different from that of the current policy. This leads to a persistent distribution mismatch. The proposed Fused-CPRO employs a diminishing offline weight sequence $\{\xi_{t}\}$ to mitigate this mismatch and ensure that its cumulative effect on the estimation bias remains bounded.
\end{itemize}
The following assumptions are used for the convergence analysis of Fused-CPRO. They are standard in the analysis of infinite-horizon average-cost methods with deep actor-critic frameworks.

\begin{assumption}[Assumptions on the Problem Structure]~

\label{assumption:problem_structure}

\begin{itemize}
\item \textbf{Ergodicity and geometric mixing:}
For any feasible parameter $\boldsymbol{\theta}\in\boldsymbol{\Theta}$, the Markov chain induced by $\pi_{\boldsymbol{\theta}}$ admits a unique stationary state-action distribution $\mathbf{P}_{\pi_{\boldsymbol{\theta}}}$. For all $t\geq0$, there exist constants $\lambda>0$ and $\varrho\in(0,1)$ such that
\[
\sup_{\boldsymbol{s}\in\mathcal{S}}d_{\mathrm{TV}}\!\left(\mathbf{P}_{\pi_{\boldsymbol{\theta}}}^{t}(\cdot\mid\boldsymbol{s}),\,\mathbf{P}_{\pi_{\boldsymbol{\theta}}}\right)\leq\lambda\varrho^{t}.
\]
\item \textbf{Compactness and boundedness:}
The state space $\mathcal{S}$ and the action space $\mathcal{A}$ are compact, and the costs/rewards $\{C_{i}^{\text{'}}(\boldsymbol{s},\boldsymbol{a})\}$ are uniformly bounded. The DNN parameter spaces $\boldsymbol{\Theta}$ and $\Omega_{i}$ are compact and convex. The outputs of all DNNs are uniformly bounded. The policy $\pi_{\theta}$ is Lipschitz continuous over $\boldsymbol{\theta}\in\boldsymbol{\Theta}$.
\end{itemize}
\end{assumption}

\begin{assumption}[Step sizes and offline weight schedule]~

\label{assumption:step_size}

The sequences $\{\eta_{t}\}$, $\{\gamma_{t}\}$, $\{\alpha_{t}\}$, $\{\beta_{t}\}$, and $\{\xi_{t}\}$ are deterministic, positive, and non-increasing.

\begin{enumerate}
\item The sequence $\{\alpha_{t}\}$ satisfies $\alpha_{t}\to0$, $\sum_{t}\alpha_{t}=\infty$ and $\sum_{t}\alpha_{t}^{2}<\infty$.
\item The sequence $\{\beta_{t}\}$ satisfies $\beta_{t}\to0$, $\sum_{t}\beta_{t}=\infty$ and $\sum_{t}\beta_{t}^{2}<\infty$. The window-drift term is summable: $\sum_{t=0}^{\infty}\alpha_{t}\beta_{t}\log(t+1)<\infty$.
\item The critic step size $\{\eta_{t}\}$ and the target-network step size $\{\gamma_{t}\}$ satisfy $\eta_{t}\to0,\gamma_{t}\to0$ and $\sum_{t}\eta_{t}=\infty,\sum_{t}\gamma_{t}=\infty$.
\item The actor evolves on the slowest time scale:
\[
\lim_{t\rightarrow\infty}\frac{\beta_{t}}{\alpha_{t}}=0,\qquad\lim_{t\rightarrow\infty}\frac{\beta_{t}}{\eta_{t}}=0.
\]
\item The offline weight $\xi_{t}\to0$ and satisfies $\sum_{t}\alpha_{t}\xi_{t}<\infty$, ensuring that the cumulative bias introduced by offline data reuse is summable and does not impair asymptotic consistency.
\item Define $n_{t}=t-\lceil t^{0.43}\rceil$. Then the following series are finite:

\begin{equation}
\begin{aligned}
& \sum_{t=0}^{\infty}\alpha_{t}(1-\gamma_{t})^{t^{0.215}}\gamma_{n_{t}}^{-1/2}<\infty,\;
  \sum_{t=0}^{\infty}\alpha_{t}\gamma_{n_{t}}^{1/2}\eta_{n_{t}}^{-1/2}<\infty,\\
& \sum_{t=0}^{\infty}m_{Q}\alpha_{t}\gamma_{n_{t}}^{1/2}\eta_{n_{t}}^{1/2}t^{0.215}<\infty,\;
  \sum_{t=0}^{\infty}m_{Q}\alpha_{t}\eta_{n_{t}}t^{-0.57}<\infty,\\
& \sum_{t=0}^{\infty}m_{Q}\alpha_{t}\eta_{n_{t}}\beta_{n_{t}}t^{0.86}<\infty,\;
  \sum_{t=0}^{\infty}m_{Q}\alpha_{t}\gamma_{n_{t}}^{1/2}t^{0.215}\xi_{n_{t}}<\infty.
\end{aligned}
\end{equation}
where $m_{Q}$ denotes the width of the DNN used in the critic module.
\end{enumerate}
\end{assumption}

Assumptions~\ref{assumption:step_size}.1--\ref{assumption:step_size}.3 are standard stochastic-approximation conditions, while Assumption~\ref{assumption:step_size}.4 makes the actor slower than the running estimators and critic, as in \cite{wang2024sldac}. Assumptions~\ref{assumption:step_size}.5--\ref{assumption:step_size}.6 control the additional offline-data bias by making its running-estimator and critic-tracking contributions summable. For polynomial step sizes of the form

\begin{align}
\alpha_{t} & =\mathcal{O}\!\big(m_{Q}^{-1/2}(t+1)^{-\kappa_{1}}\big), & \beta_t & =\mathcal{O}\!\big(m_{Q}^{-1/2}(t+1)^{-\kappa_{2}}\big),\nonumber \\
\eta_{t} & =\mathcal{O}\!\big(m_{Q}^{-1/2}(t+1)^{-\kappa_{3}}\big), & \gamma_{t} & =\mathcal{O}\!\big((t+1)^{-\kappa_{4}}\big),\nonumber \\
\xi_{t} & =\mathcal{O}\!\big((t+1)^{-\kappa_{5}}\big),\label{eq:poly_stepsize}
\end{align}

Assumption \ref{assumption:step_size} can then be satisfied when $\kappa_{1},\kappa_{2},\kappa_{3},\kappa_{4},\kappa_{5}$ lie in the following region:

\begin{equation}
\left\{ \begin{aligned} & 1>2\kappa_{2}-1>\kappa_{1}>0.43>\kappa_{4}>0,\\
& \min\{0.5\kappa_{1}+0.5\kappa_{2}+0.5\kappa_{3}-0.5,\ \kappa_{1}+\kappa_{3}\}>0.43,\\ & \kappa_{1}+0.5\kappa_{4}-0.5\kappa_{3}>1,\\ & \kappa_{1}+0.5\kappa_{3}+0.5\kappa_{4}>1.215,\\ & \kappa_{1}+\kappa_{5}>1,\\ & \kappa_{3}+2\kappa_{5}>1,\\ & \kappa_{1}+\kappa_{5}+0.5\kappa_{4}>1.215.
\end{aligned}
\right.\label{eq:kappa_region_core}
\end{equation}

The region is nonempty; for example, $\left(\kappa_{1},\kappa_{2},\kappa_{3},\kappa_{4},\kappa_{5}\right)=(0.9,\,0.96,\,0.21,\,0.425,\,0.4)$ is feasible. These conditions are sufficient rather than necessary and may be mildly relaxed in finite-horizon implementations, provided stability is maintained.

\subsection{Convergence of the Critic Module}
We establish a finite-time tracking bound for the critic module. Policy reuse and offline data introduce additional terms absent from the single-policy SLDAC analysis. The critic evaluates the surrogate $Q$-functions $\{\hat{Q}_{i}^{\pi_{\boldsymbol{\theta}_{t}}}\}_{i=0}^{I}$ defined in $~\eqref{eq:definition of Qhat}$ with DNNs $\{f(\bar{\boldsymbol{\omega}}^{i};\boldsymbol{s},\boldsymbol{a})\}_{i=0}^{I}$ under the current mixed policy $\pi_{\boldsymbol{\theta}_{t}}$. It also uses the running estimate $\hat{J}_{i}^{t}$ in place of the unknown average reward/cost. The estimation error is defined as
\begin{equation}
\epsilon_{t,i}^{\mathrm{cri}}
\triangleq
\left|\,
\mathbb{E}_{p_t}\!\left[f\!\left(\bar{\boldsymbol{\omega}}_{t}^{i};\boldsymbol{s},\boldsymbol{a}\right)\,\middle|\,\bar{\boldsymbol{\omega}}_{t-1}^{i}\right]
-\hat{Q}_{i}^{\pi_{\boldsymbol{\theta}_{t}}}(\boldsymbol{s},\boldsymbol{a})
\right|,\qquad \forall i.
\label{eq:critic_error}
\end{equation}
For each critic DNN with initialization $\boldsymbol{\omega}_{0}^{i}$ and constraint set $\Omega_{i}\triangleq\mathbb{B}(\boldsymbol{\omega}_{0}^{i},R_{\boldsymbol{\omega}})$, its local linearization function class is defined as
\begin{equation}
\hat{f}(\boldsymbol{\omega}^{i};\boldsymbol{s},\boldsymbol{a})\triangleq f(\boldsymbol{\omega}_{0}^{i};\boldsymbol{s},\boldsymbol{a})+\bigl\langle\nabla_{\boldsymbol{\omega}}f(\boldsymbol{\omega}_{0}^{i};\boldsymbol{s},\boldsymbol{a}),\,\boldsymbol{\omega}^{i}-\boldsymbol{\omega}_{0}^{i}\bigr\rangle,\quad\boldsymbol{\omega}^{i}\in\Omega_{i},\label{eq:local_linearization}
\end{equation}
and the induced function class is defined as $\hat{\mathcal{F}}_{i}\triangleq\{\hat{f}(\boldsymbol{\omega}^{i};\cdot):\boldsymbol{\omega}^{i}\in\Omega_{i}\}$. This linearization satisfies the identity
\begin{equation}
\hat{f}(\boldsymbol{\omega}_{a}^{i};\boldsymbol{s},\boldsymbol{a})-\hat{f}(\boldsymbol{\omega}_{b}^{i};\boldsymbol{s},\boldsymbol{a})=\bigl\langle\nabla_{\boldsymbol{\omega}}f(\boldsymbol{\omega}_{0}^{i};\boldsymbol{s},\boldsymbol{a}),\,\boldsymbol{\omega}_{a}^{i}-\boldsymbol{\omega}_{b}^{i}\bigr\rangle,\quad\forall i.
\end{equation}

\begin{assumption}[Assumptions on the target Q]~
\label{assumption:target_Q}
\begin{enumerate}
\item \textbf{Representability:}
For each $t$ and $i$, there exists a point $\dot{\boldsymbol{\omega}}_{t}^{i}\in\Omega_{i}$ such that $\hat{f}(\dot{\boldsymbol{\omega}}_{t}^{i};\boldsymbol{s},\boldsymbol{a})=\hat{Q}_{i}^{\pi_{\boldsymbol{\theta}_{t}}}(\boldsymbol{s},\boldsymbol{a})$.
\item \textbf{Uniform TD stability:}
As illustrated in \cite{tsitsiklis1999avgtd}, the TD feature matrix is defined as:
\begin{align}
\mathbf{A}_{\boldsymbol{\theta}_{t}} & \triangleq\mathbb{E}_{p_{t}}\Bigl[\nabla_{\boldsymbol{\omega}}f(\boldsymbol{\omega}_{0}^{i};\boldsymbol{s}_{t},\boldsymbol{a}_{t})\nonumber \\
& \bigl(\nabla_{\boldsymbol{\omega}}f(\boldsymbol{\omega}_{0}^{i};\boldsymbol{s}_{t},\boldsymbol{a}_{t})-\nabla_{\boldsymbol{\omega}}f(\boldsymbol{\omega}_{0}^{i};\boldsymbol{s}_{t+1},\boldsymbol{a}_{t+1}')\bigr)^{\intercal}\Bigr]
\end{align}
where $\boldsymbol{a}_{t+1}'\sim\pi_{\boldsymbol{\theta}_{t}}(\cdot\mid\boldsymbol{s}_{t+1})$. The matrix $\mathbf{A}_{\boldsymbol{\theta}_{t}}$ satisfies $\boldsymbol{\omega}^{\intercal}\mathbf{A}_{\boldsymbol{\theta}_{t}}\boldsymbol{\omega}>0$ for all nonzero $\boldsymbol{\omega}$. Also, the symmetric part $\mathbf{A}_{\boldsymbol{\theta}_{t}}+\mathbf{A}_{\boldsymbol{\theta}_{t}}^{\intercal}$ is uniformly positive definite over $\boldsymbol{\theta}_{t}\in\boldsymbol{\Theta}$, i.e., $\lambda_{\min}(\mathbf{A}_{\boldsymbol{\theta}_{t}}+\mathbf{A}_{\boldsymbol{\theta}_{t}}^{\intercal})>\varsigma$ for constant $\varsigma>0$.
\end{enumerate}
\end{assumption}

Assumption~\ref{assumption:target_Q}.1 ensures that the local linearization class can represent the surrogate target $Q$-functions. Assumption~\ref{assumption:target_Q}.2 guarantees a unique, stable TD fixed point. Similar realizability and stability assumptions are standard in analyses of neural TD/Q-learning and actor-critic methods with over-parameterized networks; see, e.g., \cite{Allen2019b,Cao2019b,DQlearning,Assumption31,Assumption32,linnerAC}. The critic error in \eqref{eq:critic_error} can be decomposed and upper bounded by the sum of the local linearization error and the approximation error between the linearized function $\hat{f}(\bar{\boldsymbol{\omega}}_{t}^{i})$ and the true Q-function $\hat{Q}_{i}^{\pi_{\theta_{t}}}$:

\begin{equation}
\epsilon_{t,i}^{\mathrm{cri}}\le\left|\mathbb{E}\big[f(\bar{\boldsymbol{\omega}}_{t}^{i})\big]-\mathbb{E}\big[\hat{f}(\bar{\boldsymbol{\omega}}_{t}^{i})\big]\right|+\left|\mathbb{E}\big[\hat{f}(\bar{\boldsymbol{\omega}}_{t}^{i})\big]-\hat{Q}_{i}^{\pi_{\theta_{t}}}\right|,
\end{equation}
where, with high probability, $\epsilon_{m_Q}$ denotes a uniform local-linearization bound of order $\mathcal{O}(R_{\boldsymbol{\omega}}^{4/3}L^{4}\sqrt{m_Q\log m_Q}+R_{\boldsymbol{\omega}}^{2}L^{5})$. For fixed depth $L$ and the NTK scaling $R_{\boldsymbol{\omega}}=\mathcal{O}(m_Q^{-1/2})$, this gives $\epsilon_{m_Q}=\mathcal{O}(m_Q^{-1/6}\sqrt{\log m_Q}+m_Q^{-1})\to0$. Appendix~\ref{sec:Technical-Bounds-DNN} states the width and radius conditions for this bound. The second term can be further expressed as
\begin{equation}
\left|\mathbb{E}\big[\hat{f}(\bar{\boldsymbol{\omega}}_{t}^{i})\big]-\hat{Q}_{i}^{\pi_{\theta_{t}}}\right|\le\|\nabla_{\omega}f(\boldsymbol{\omega}_{0})\|_{2}\cdot\left\Vert \mathbb{E}[\bar{\boldsymbol{\omega}}_{t}^{i}]-\dot{\boldsymbol{\omega}}_{t}^{i}\right\Vert _{2},\ \forall i,
\end{equation}
where the same DNN bounds give $\|\nabla_{\omega}f(\boldsymbol{\omega}_{0})\|_{2}=\mathcal{O}(\sqrt{m_Q})$ with high probability.

To bound $\|\mathbb{E}[\bar{\boldsymbol{\omega}}_t^i]-\dot{\boldsymbol{\omega}}_t^i\|_2$, introduce a reference critic trajectory that matches the actual recursion up to $n_t$ and then evolves under the frozen policy $\pi_{\theta_{n_t}}$ over $\{n_t+1,\ldots,t\}$. The resulting tracking bound is stated below.

\begin{lemma}[Convergence Rate of the Critic]~

\label{lemma-convergence-rate-critic}

Fix $\kappa_{6}\in(0,1)$ and define $n_{t}\triangleq t-\lceil t^{\kappa_{6}}\rceil$. Under the assumptions \ref{assumption:problem_structure}-\ref{assumption:target_Q}, the critic error admits the bound
\begin{align}
\epsilon_{t,i}^{\mathrm{cri}} & \le\mathcal{O}\Bigl(\epsilon_{m_{Q}}+\frac{(1-\gamma_{t})^{t^{\kappa_{6}/2}}}{\sqrt{\gamma_{t}}}+\sqrt{\tfrac{\gamma_{n_{t}}}{\eta_{n_{t}}}}\nonumber\\
&\quad +m_{Q}\sqrt{\gamma_{n_{t}}\eta_{n_{t}}}t^{\kappa_{6}/2}+m_{Q}\eta_{n_{t}}t^{\kappa_{6}-1}\nonumber\\
&\quad +m_{Q}\eta_{n_{t}}\beta_{n_{t}}t^{2\kappa_{6}}+m_{Q}\sqrt{\gamma_{n_{t}}}t^{\kappa_{6}/2}\,\xi_{n_{t}}\Bigr),\quad\forall i,
\end{align}
almost surely. Here $\xi_{t}$ is the offline weight used in the mixed estimation of $\hat{J}_{i}^{t}$. The last term, proportional to $\xi_{n_t}$, is the additional tracking error caused by offline data reuse. $\epsilon_{m_{Q}}$ is the local linearization error. Consequently, if $\xi_{t}\to0$ and the step sizes are chosen so that every term on the right-hand side vanishes, then $\epsilon_{t,i}^{\mathrm{cri}}\to\mathcal{O}(\epsilon_{m_{Q}})$. Moreover, $\epsilon_{m_{Q}}\to0$ as $m_{Q}\to\infty$, meaning the critic tracking error can be made arbitrarily small with a sufficiently wide DNN.

\end{lemma}

\emph{Proof sketch.} Equations~\eqref{eq:critic_error}--\eqref{eq:local_linearization} reduce the critic error to $\mathcal{O}(\epsilon_{m_Q})+\mathcal{O}(\sqrt{m_Q})\|\mathbb{E}[\bar{\boldsymbol{\omega}}_t^i]-\dot{\boldsymbol{\omega}}_t^i\|_2$. Adding and subtracting the frozen-policy reference critic splits the parameter error into frozen-policy TD tracking and policy drift. Uniform TD stability contracts the former; geometric mixing and slow actor motion bound the latter, while offline samples introduce the $\mathcal{O}(\sqrt{m_Q}\xi_{n_t})$ drift. Applying the target-critic averaging weights yields the terms in the stated bound, whose cumulative effect is controlled by Assumption~\ref{assumption:step_size}.6. Appendix~\ref{sec:proof-of-critic} gives the full derivation.

\subsection{Asymptotic Consistency of the Surrogate Functions}
The actor update relies on the running estimates $\hat{J}_{i}^{t}$ and $\hat{\boldsymbol{g}}_{i}^{t}$ in (\ref{eq:J-average})--(\ref{eq:g-average}). Reusing online-buffer samples induces off-policy bias because those samples are generated by slightly outdated policy parameters. Offline data reuse adds another bias term whose distribution mismatch does not vanish without the decay of $\xi_{t}$. The following lemma states that the running estimates $\hat{J}_{i}^{t}$ and $\hat{\boldsymbol{g}}_{i}^{t}$ remain asymptotically consistent under these two bias sources.

\begin{lemma}[Asymptotic consistency of running estimates]~

\label{lemma-asymptotic-consistency}

Suppose Assumptions \ref{assumption:problem_structure}-\ref{assumption:step_size} hold. Then, for each $i\in\{0,1,\ldots,I\}$, the running estimates satisfy
\begin{align}
\lim_{t\rightarrow\infty}\bigl|\hat{J}_{i}^{t}-J_{i}(\boldsymbol{\theta}_{t})\bigr| & =0,\\
\lim_{t\rightarrow\infty}\bigl\|\hat{\boldsymbol{g}}_{i}^{t}-\nabla_{\boldsymbol{\theta}}J_{i}(\boldsymbol{\theta}_{t})\bigr\|_{2} & \leq c_{\pi}\epsilon_{m_{Q}},
\end{align}
almost surely, where $c_{\pi}>0$ is a constant independent of $t$ and $m_{Q}$.

\end{lemma}

\emph{Proof sketch.} Regard \eqref{eq:J-average} and \eqref{eq:g-average} as stochastic-approximation recursions tracking $J_i(\boldsymbol{\theta}_t)$ and $\nabla_{\boldsymbol{\theta}}\hat J_i(\boldsymbol{\theta}_t)$ \cite{Lemma3}. Mixing and policy drift bound their conditional biases by $\mathcal{O}(t^{-2}+T_t\beta_{t-T_t+1}+\xi_t)$, with an additional $\mathcal{O}(\bar\epsilon_{\mathrm{cri}}(t))$ for the gradient. Assumption~\ref{assumption:step_size} gives bounded moments, weighted bias summability, and target drift $\mathcal{O}(\beta_t/\alpha_t)=o(1)$, so the tracking lemma applies. Finally, $\|\nabla\hat J_i-\nabla J_i\|_2\le c_\pi(|\hat J_i^t-J_i|+\epsilon_{m_Q})$ gives the stated residual gradient error. Appendix~\ref{app:sa-tracking} states the tracking lemma, and Appendix~\ref{sec:proof-of-asymptotic-consistency} verifies its conditions for both running estimates.

\subsection{Finite-Time Convergence Rate of the Surrogate Functions}
Lemma \ref{lemma-asymptotic-consistency} establishes the asymptotic accuracy of $\hat{J}_{i}^{t}$ and $\hat{\boldsymbol{g}}_{i}^{t}$. We next quantify the finite-time convergence rate of the surrogate functions through the averaged estimation errors:
\begin{align}
\epsilon_{J}(t) & \triangleq\frac{1}{t+1}\sum_{k=0}^{t}\max_{0\le i\le I}\mathbb{E}\!\left[\left|\hat{J}_{i}^{k}-J_{i}(\boldsymbol{\theta}_{k})\right|^{2}\right],\label{eq:epsJ_def}\\
\epsilon_{g}(t) & \triangleq\frac{1}{t+1}\sum_{k=0}^{t}\max_{0\le i\le I}\mathbb{E}\!\left[\left\Vert \hat{\boldsymbol{g}}_{i}^{k}-\nabla_{\boldsymbol{\theta}}J_{i}(\boldsymbol{\theta}_{k})\right\Vert _{2}^{2}\right].\label{eq:epsg_def}
\end{align}

\begin{lemma}[Convergence Rate of the Surrogate Functions]~

\label{lemma-convergence-rate-surrogate}

Suppose Assumptions \ref{assumption:problem_structure}-\ref{assumption:target_Q} hold and the online buffer length $T_{t}=\mathcal{O}(\log(t+1))$. The finite-time convergence rate of the surrogate functions is:

\begin{align}
\epsilon_{J}(t)\le & \ \mathcal{O}\!\Big(\frac{1}{(t+1)\alpha_{t+1}}\Big)+\mathcal{O}\!\Big(\frac{1}{t+1}\sum_{k=0}^{t}\big(\alpha_{k+1}+T_{k}\alpha_{k}\nonumber \\
& +T_{k}^{2}\beta_{k}+\beta_{k}^{2}(\alpha_{k+1}^{-2}+\alpha_{k+1}^{-1})\big)\Big)\nonumber \\ & +\mathcal{O}\!\Big(\frac{1}{t+1}\sum_{k=0}^{t}\xi_{k+1}\Big),\label{eq:epsJ_rate}\\[1mm]
\epsilon_{g}(t)\le & \ \mathcal{O}\!\Big(\frac{1}{(t+1)\alpha_{t+1}}\Big)+\mathcal{O}\!\Big(\frac{1}{t+1}\sum_{k=0}^{t}\big(\bar{\epsilon}_{\mathrm{cri}}(k)+k^{-1}\nonumber \\
& +T_{k}\beta_{k-T_{k}+1}+\beta_{k}^{2}(\alpha_{k+1}^{-2}+\alpha_{k+1}^{-1})+\alpha_{k+1}\big)\Big)\nonumber \\ & +\mathcal{O}\!\Big(\sqrt{\epsilon_{J}(t)}+\frac{1}{t+1}\sum_{k=0}^{t}\xi_{k+1}+\epsilon_{m_{Q}}\Big).
\end{align}
where the aggregated critic error $\bar{\epsilon}_{\mathrm{cri}}(k)\triangleq\max_{0\le i\le I}\epsilon_{k,i}^{\mathrm{cri}}$.

In both bounds, the term involving $\xi_{k+1}$ captures the finite-time effect of offline data reuse and vanishes asymptotically under Assumption~\ref{assumption:step_size}.5.

\end{lemma}

\emph{Proof sketch.} For fixed $i$, square the recursion for $\boldsymbol y_k=\hat{\boldsymbol g}_i^k-\nabla J_i(\boldsymbol{\theta}_k)$. Averaging separates telescoping, target-drift, estimator-variance, conditional-bias, and cross terms. Lipschitz continuity makes the target drift $\mathcal{O}(\beta_k)$, and Lemma~\ref{lemma-convergence-rate-critic} controls the critic, buffer, and offline contributions. Cauchy--Schwarz gives $F_i(t)\le c_1\sqrt{F_i(t)N(t)}+c_2R(t)$, where $N(t)$ collects $\beta_k^2/\alpha_{k+1}^2$ and $R(t)$ the remaining displayed rates; solving it yields the $\epsilon_g(t)$ bound. The scalar value recursion yields the $\epsilon_J(t)$ bound from its telescoping, drift, variance, and offline-bias terms. Appendix~\ref{sec:proof-of-surrogate} supplies the term-by-term bounds.
\subsection{Convergence of the Actor Module}

We now analyze the limiting behavior of the actor sequence $\{\boldsymbol{\theta}_t\}$ and show that every limiting point is an $\epsilon_{m_Q}$-KKT point of Problem~\eqref{eq:problem 1-1}, where $\epsilon_{m_Q}\to 0$ as $m_Q\to\infty$.

Consider any subsequence $\{\boldsymbol{\theta}_{t_j}\}_{j=1}^{\infty}$ converging to a limiting point $\boldsymbol{\theta}^{\star}$. Since $\boldsymbol{\Theta}$ is compact and the sequences $\{\hat{J}_{i}^{t_j}\}_{j}$, $\{\hat{\boldsymbol g}_{i}^{t_j}\}_{j}$, and $\{\boldsymbol{\theta}_{t_j}\}_{j}$ are bounded, there exists a further subsequence, still indexed by $\{t_j\}$ for notational simplicity, and limits $J_i^\infty\in\mathbb{R}$, $\boldsymbol g_i^\infty\in\mathbb{R}^{n_\theta}$ such that $\hat{J}_{i}^{t_j}\to J_i^\infty$ and $\hat{\boldsymbol g}_{i}^{t_j}\to \boldsymbol g_i^\infty$. Hence, the surrogate functions $\bar J_i^{t_j}(\boldsymbol\theta)$ converge uniformly on $\boldsymbol\Theta$ to
\[
\bar J_i^\infty(\boldsymbol\theta)
=
J_i^\infty
+
(\boldsymbol g_i^\infty)^T(\boldsymbol\theta-\boldsymbol\theta^\star)
+
\zeta_i\|\boldsymbol\theta-\boldsymbol\theta^\star\|_2^2.
\]
By Lemma~\ref{lemma-asymptotic-consistency}, we have
\begin{equation}
\begin{aligned}
\left|
\bar J_i^\infty(\boldsymbol\theta^\star)-J_i(\boldsymbol\theta^\star)
\right| &= 0,\\
\left\|
\nabla \bar J_i^\infty(\boldsymbol\theta^\star)
-
\nabla J_i(\boldsymbol\theta^\star)
\right\|_2
&\le c_\pi \epsilon_{m_Q}.
\end{aligned}
\end{equation}

Combining Lemmas~\ref{lemma-convergence-rate-critic}-- \ref{lemma-convergence-rate-surrogate} with Assumptions \ref{assumption:problem_structure}-- \ref{assumption:target_Q} and the CSSCA KKT-convergence argument in \cite{CSSCA,SCAOPO} yields the following theorem.

\begin{theorem}[Global Convergence of Algorithm \ref{alg:the-Single-Loop-Deep}]~
\label{Theorem:GlobalConvergence}

Suppose Assumptions \ref{assumption:problem_structure}-\ref{assumption:target_Q} are satisfied and the initial point $\boldsymbol{\theta}_{0}$ is feasible, i.e., $\max_{i\in\{1,\ldots,I\}}J_{i}\!\left(\boldsymbol{\theta}_{0}\right)\leq0$, and the number of sampled data is set to $T_{t}=\mathcal{O}(\log t)$. Denote $\{\boldsymbol{\theta}_{t}\}_{t=1}^{\infty}$ as the iterates generated by Algorithm \ref{alg:the-Single-Loop-Deep} with a sufficiently small initial step size $\beta_{0}$. Then every limiting point $\boldsymbol{\theta}^{\star}$ of $\{\boldsymbol{\theta}_{t}\}_{t=1}^{\infty}$ satisfying the Slater condition \cite{Slater} is an $\epsilon_{m_{Q}}$-KKT point of Problem~\eqref{eq:problem 1-1}; namely, there exist multipliers $\lambda=[\lambda_{1},\ldots,\lambda_{I}]^{T}\succeq0$ such that
\begin{gather}
\left\Vert \nabla J_{0}\!\left(\boldsymbol{\theta}^{\star}\right)+\sum_{i=1}^{I}\lambda_{i}\nabla J_{i}\!\left(\boldsymbol{\theta}^{\star}\right)\right\Vert _{2}\leq\epsilon_{m_{Q}},\\
J_{i}\!\left(\boldsymbol{\theta}^{\star}\right)\leq\epsilon_{m_{Q}},\qquad i=1,\ldots,I,\\
\left|\lambda_{i}J_{i}\!\left(\boldsymbol{\theta}^{\star}\right)\right|\leq\epsilon_{m_{Q}},\qquad i=1,\ldots,I,
\end{gather}
where $\lim_{m_{Q}\rightarrow\infty}\epsilon_{m_{Q}}=0$.
\end{theorem}

\emph{Proof sketch.} Lemma~\ref{lemma-asymptotic-consistency} transfers the limiting surrogate conditions above to the original objectives up to $\mathcal{O}(\epsilon_{m_Q})$, while Lemmas~\ref{lemma-convergence-rate-critic} and~\ref{lemma-convergence-rate-surrogate} provide the required consistency. Feasible initialization and the Slater condition permit the CSSCA KKT argument in \cite{CSSCA,SCAOPO}; absorbing fixed constants into $\epsilon_{m_Q}$ gives the stated stationarity, feasibility, and complementary-slackness bounds.

\begin{remark}
The feasibility of the initial point $\boldsymbol{\theta}_{0}$ is a sufficient condition for the CSSCA-based convergence argument to exclude undesired stationary points of the constraint-violation minimization problem and establish convergence to a KKT point of Problem~\eqref{eq:problem 1-1}. Since Algorithm~\ref{alg:the-Single-Loop-Deep} includes the feasibility-restoration step in \eqref{eq:feasible update}, the iterates can still be pushed toward the feasible region even when $\boldsymbol{\theta}_{0}$ is infeasible. In practice, the algorithm may still converge from an infeasible initialization unless it is trapped near an undesired stationary point of the constraint-violation minimization problem. Related discussions can be found in the original CSSCA framework and its RL extensions \cite{CSSCA,SCAOPO,wang2024sldac}.
\end{remark}

\section{Simulation Results}
\label{sec:simulation-results}

In this section, we evaluate the proposed Fused-CPRO on two continuous-control wireless CMDP scenarios: delay-constrained power control for downlink MU-MIMO and sum-rate maximization for MIMO-ISAC beamforming with sensing-accuracy constraints. The experiments are designed to validate three key claims from the preceding analysis: (i) Fused-CPRO achieves faster convergence and better final performance by fusing heterogeneous priors, (ii) the mixed offline-online data reuse reduces the number of costly online interactions required to reach a feasible and high-performing policy, and (iii) the CSSCA-based actor reliably enforces long-term constraints under stochastic non-convexity. Across the two benchmarks, we compare Fused-CPRO with the following baselines, whose parameters are carefully tuned to achieve best possible empirical performance for each benchmark.
\begin{itemize}
\item \textbf{PPO-Lag}: a typical Lagrangian policy-gradient baseline based on PPO-style updates \cite{PPOLagTRPOLag}, which handles the safety constraint by augmenting the PPO objective with an adaptive penalty coefficient for cost-limit violation. It uses fresh on-policy trajectories at each update and does not exploit policy reuse or offline data. As a fully online method, it requires a large number of environment interactions to converge.
\item \textbf{SLDAC}: the single-loop constrained actor-critic algorithm \cite{wang2024sldac}, which applies CSSCA in the actor step and updates the critic only once or a few finite times per iteration. It does not exploit old policy and offline data. SLDAC shares the same CSSCA backbone as Fused-CPRO but relies entirely on online samples, serving as an ablation that isolates the benefit of policy reuse and offline data.
\item \textbf{SCAOPO}: the off-policy constrained actor-only method \cite{SCAOPO}, which replaces the original problem with a sequence of convex surrogate objective/feasibility subproblems and solves them through a Lagrange-dual procedure. It does not exploit old policy and offline data. Moreover, as an actor-only method without a learned critic, its gradient estimates may suffer from higher variance.
\item \textbf{HRL}: a policy reuse-based constrained RL method from \cite{zhang2025hrl,zhang2025prcrl}, which uses the same source-policy pool and the same DK policy as Fused-CPRO but does not exploit the offline data. Furthermore, it is an actor-only method similar to SCAOPO, but lacks a complete and rigorous convergence proof. Comparing Fused-CPRO with HRL isolates the benefit of offline data reuse and the critic-actor structure. For Fused-CPRO and HRL, we keep the source-policy pool and the DK policy identical.
\item \textbf{CPO}: the on-policy constrained policy-optimization method in \cite{CPO}, which performs policy updates within a KL-based trust region and targets near-constraint satisfaction at each iteration. In practice, it approximately solves the resulting trust-region step with conjugate gradient and backtracking line search, and it does not exploit policy reuse or offline data. Like PPO-Lag, CPO requires fresh on-policy data per update and is therefore sample-intensive.
\end{itemize}

\subsection{Delay-Constrained Power Control for Downlink MU-MIMO}
We first instantiate the delay-constrained MU-MIMO resource-allocation model in Section~\ref{subsec:mimo-model}. The BS is equipped with $N_{t}=8$ antennas and serves $K=4$ single-antenna users. The action consists of the user power vector $\mathbf{p}(t)$ and the RZF regularization factor $\alpha_Z(t)$. The channel of user $k$ is generated as follows:
\begin{equation}
\mathbf{h}_{k}=\sum_{i=1}^{N_{p}}\bar{\alpha}_{k,i}\mathbf{a}(\psi_{k,i}),
\end{equation}
where $N_{p}=4$, the angles of departure $\{\psi_{k,i}\}$ follow a Laplacian distribution with angular spread $5^{\circ}$, and the path gains are normalized such that the large-scale fading gains are uniformly distributed in $[-10,10]$ dB. The system bandwidth is $10$ MHz, the slot duration is $1$ ms, and the noise power density is $-100$ dBm/Hz. The rate, queue evolution, and delay-constrained power-minimization objective follow \eqref{eq:mimo_rate_model}--\eqref{eq:mimo_power_delay_problem}. The DK policy is chosen as a queue-aware proportional allocation rule \cite{stolyar2001largestLWDF}:
\begin{equation}
p_{k}^{\mathrm{DK}}(t)=P_{\max}\frac{Q_{k}(t)/\lambda_{k}}{\sum_{j=1}^{K}Q_{j}(t)/\lambda_{j}},
\end{equation}
combined with the classical RZF regularization heuristic $\alpha_{Z}^{\mathrm{DK}}(t)=\sigma^{2}/\bar{p}(t)$, where $\bar{p}(t)=P_{\max}/K$.

The following MU-MIMO simulation results in Figs.~\ref{fig:mimo_power}--\ref{fig:mimo_mix} use mini-batches of $100$ samples and one critic update per iteration ($q=1$). The offline data are collected in advance from the source and DK policies and require no additional online interactions during learning. In a practical deployment, this offline collection can be performed once per scenario and reused across multiple deployments or fine-tuning runs. The reused-policy pool consists of one SLDAC source policy trained under a different random channel environment and one DK policy, with uniform initial reuse probabilities over the new policy and the reused policies. In this case, we choose the step sizes as $\alpha_t=\frac{1}{t^{0.6}}$ and $\gamma_t=\frac{1}{t^{0.3}}$. For the actor parameters, we separate the step sizes of $\boldsymbol{\rho}$ and $\boldsymbol{\psi}$ in the simulations, and set them as $\beta_t^{(\rho)}=\frac{1}{t^{0.2}}$ and $\beta_t^{(\psi)}=\frac{1}{t^{0.7}}$, respectively. The offline weight is set as $\xi_t=\frac{0.5}{t^{0.7}}$.

\begin{figure}[!t]
\centering
\subfloat[Average total transmit power.\label{fig:mimo_power}]{\includegraphics[width=0.8\linewidth]{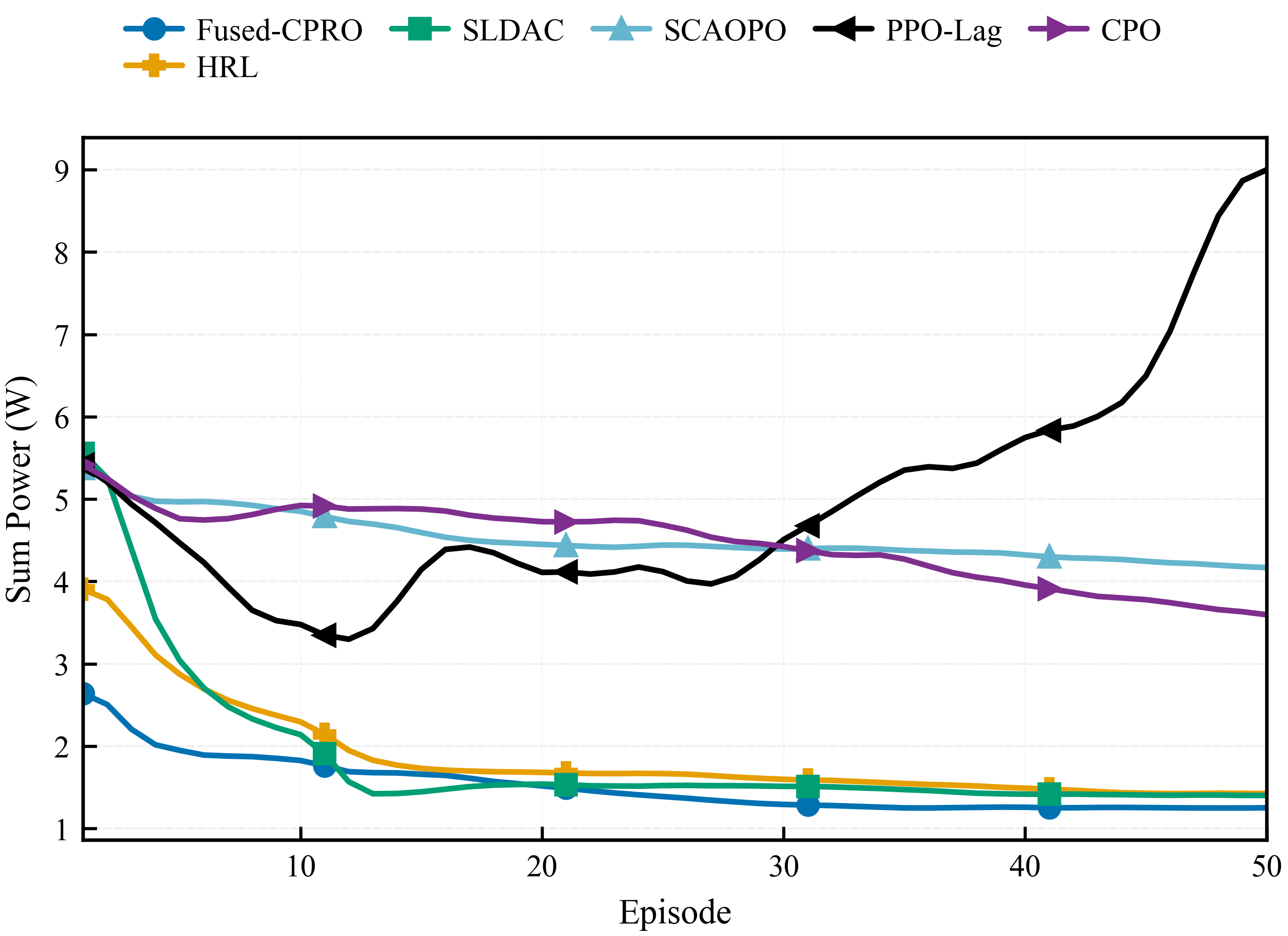}}
\par\smallskip
\subfloat[Average user delay backlog.\label{fig:mimo_delay}]{\includegraphics[width=0.8\linewidth]{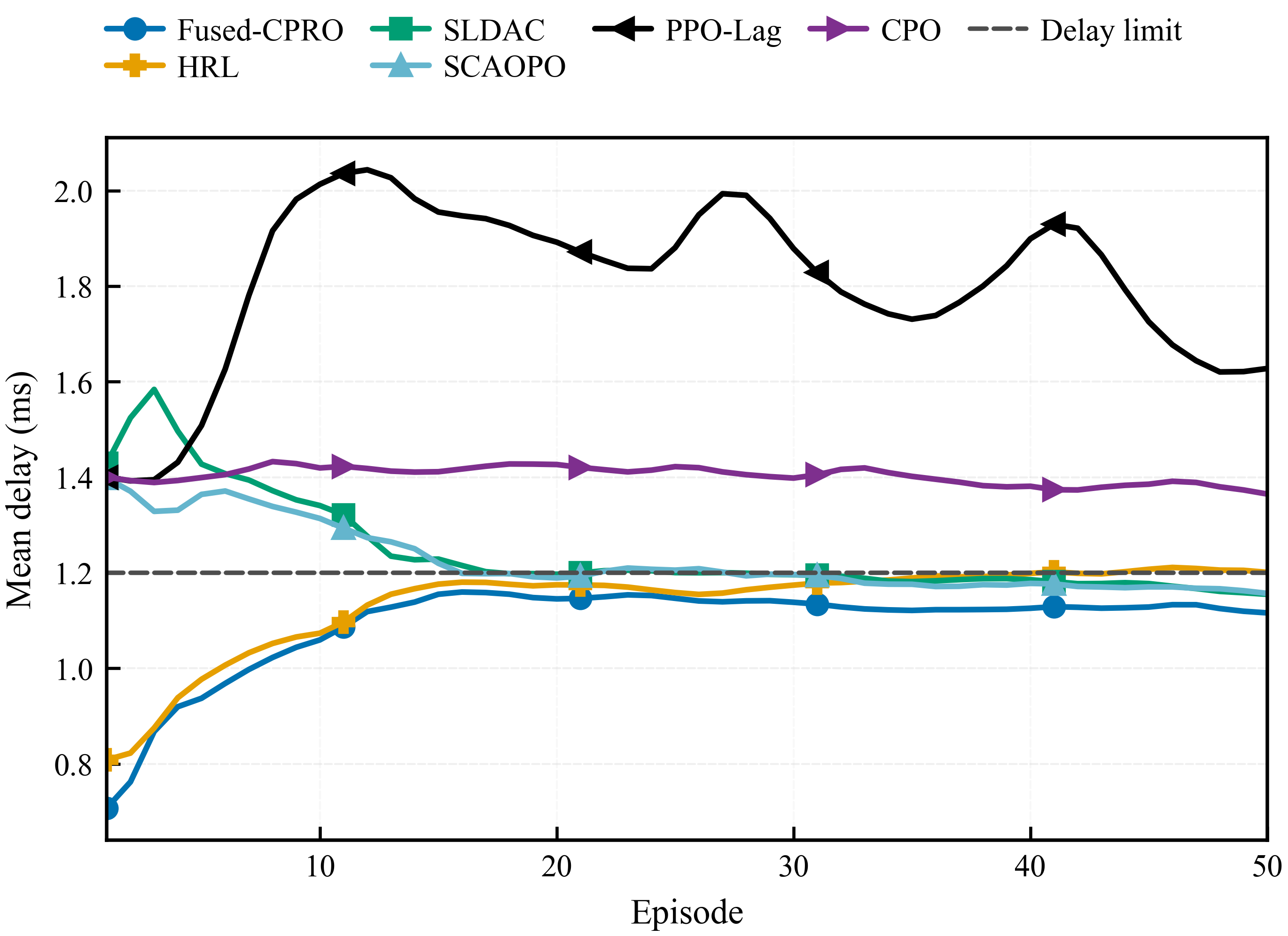}}
\caption{Simulation Results of the delay-constrained MU-MIMO task}
\label{fig:mimo_main}
\end{figure}

Figs.~\ref{fig:mimo_power} and \ref{fig:mimo_delay} show the learning curves of transmit power and average delay, respectively. Fused-CPRO achieves the lowest transmit power among all methods while keeping the average delay backlog below the constraint limit, demonstrating that the CSSCA actor effectively enforces long-term constraints even with policy reuse and mixed offline-online data. The performance gap is particularly pronounced in the early stage: Fused-CPRO improves much faster than SLDAC, which uses the same CSSCA backbone but trains only a single target policy from scratch. This early-stage advantage directly translates to reduced online interaction cost: Fused-CPRO reaches a feasible low-power region in significantly fewer iterations, consistent with the claim that fusing source and DK policies with offline data warm-starts learning and reduces costly online exploration. Traditional baselines such as CPO and PPO-Lag, which rely entirely on on-policy data without any form of prior knowledge, are even slower to approach the feasible region in this scenario.

\begin{figure}[!t]
\centering
\includegraphics[width=0.8\columnwidth]{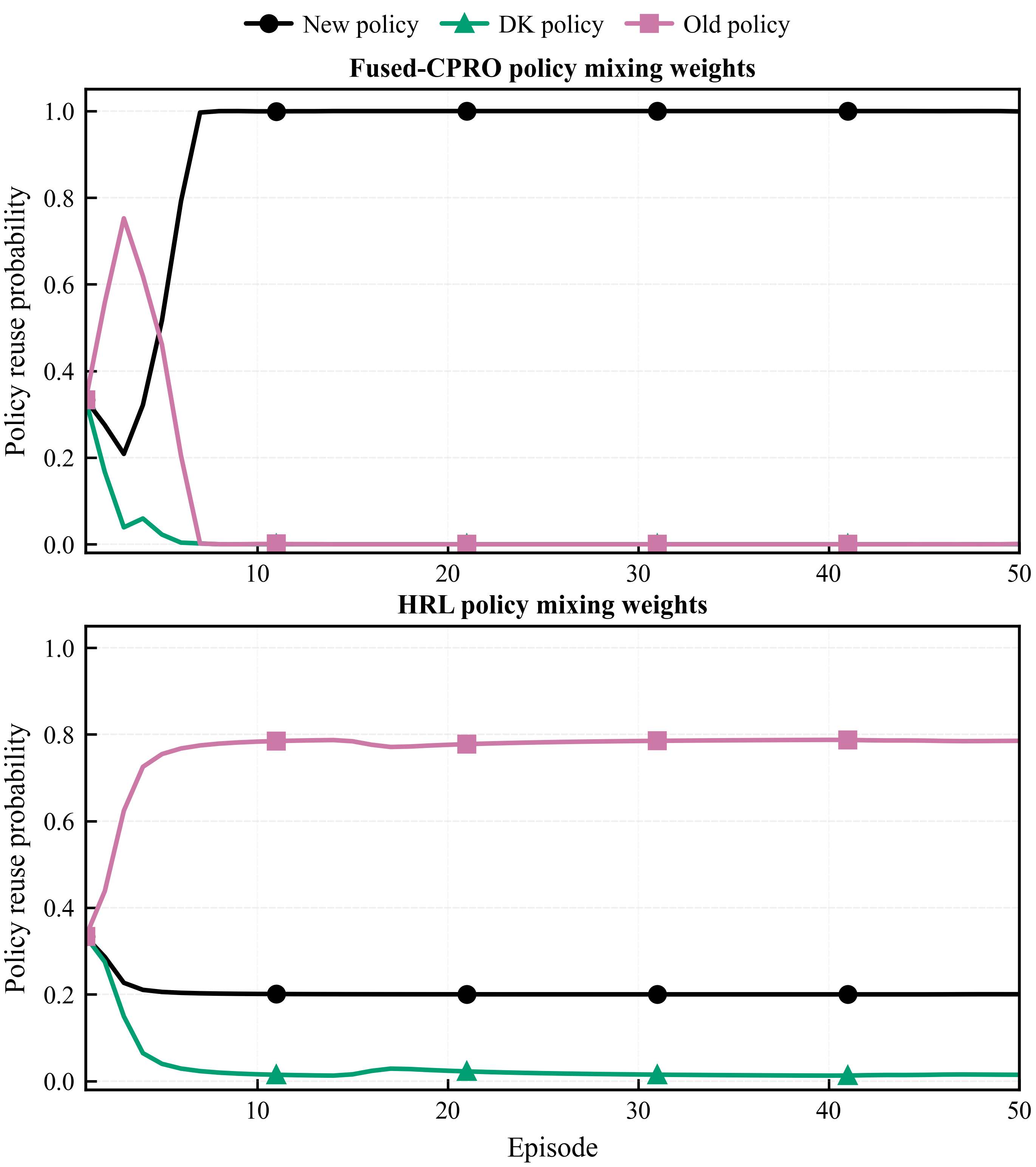}
\caption{Evolution of the reuse probabilities in the MU-MIMO task.}
\label{fig:mimo_mix}
\end{figure}

Comparing Fused-CPRO with HRL further isolates the role of offline data reuse and the critic-actor structure. Although both methods share the same source and DK policies, HRL reaches a higher final power and converges more slowly. The gap arises because HRL is an actor-only method that estimates Q-functions via Monte Carlo returns, which suffer from higher variance and become less stable when the mixed policy changes. In contrast, Fused-CPRO's critic, trained from mixed offline-online data, provides lower-variance value estimates that accelerate and stabilize the actor update.

The reuse probabilities in Fig.~\ref{fig:mimo_mix} are consistent with this behavior. In Fused-CPRO, the old and DK policies are mainly used in the initial stage, and the probability mass then shifts quickly to the new policy. In HRL, the reused policies remain dominant for much longer. This difference explains why Fused-CPRO adapts more rapidly: with offline data providing informative value estimates early on and faster learning speed due to the critic-actor structure, the actor can confidently shift toward the target policy, whereas HRL, lacking such critic support, must rely more heavily on the reused policies for a longer period.

\subsection{Sum-Rate Maximization for MIMO-ISAC Beamforming with Sensing-Accuracy Constraints}
We next instantiate the MIMO-ISAC beamforming model in Section~\ref{subsec:isac-model}. The BS is equipped with $M=16$ transmit antennas and serves $K=2$ single-antenna communication users while sensing $L=3$ targets. The channel contains $N_p=4$ paths, the large-scale fading gains are uniformly distributed in $[-10,10]$ dB, and each frame contains $N_{\mathrm{symbol}}=10$ transmit symbols. The target angles are initialized as $[-45^{\circ},0^{\circ},60^{\circ}]$, and the target reflection coefficients are initialized as $0.6+0.8j$, $0.8+0.6j$, and $\sqrt{2}/2+j\sqrt{2}/2$, respectively. The sensing and communication functions share a normalized total transmit-power budget $P_{\max}=1$.

The action follows the structure in Section~\ref{subsec:isac-model}: it contains $L$ sensing beam weights, the sensing power, and $K$ communication powers. The objective is to maximize the long-term communication sum-rate. The constraint costs are the per-target CRB values, and the fixed CRB limit is indicated by the dashed line in Fig.~\ref{fig:isac_crb}. For the DK policy, we use a risk-balanced zero-forcing (ZF) allocation rule. It estimates a sensing-risk score from the predicted target reflection coefficients, angles, and steering-vector coherence, allocates more sensing power when the predicted sensing risk is high, and allocates the remaining communication power according to ZF channel gains. This gives Fused-CPRO and HRL the same communication-aware DK policy in the MIMO-ISAC task.

The following MIMO-ISAC simulation results in Figs.~\ref{fig:isac_rate}--\ref{fig:isac_mix} use mini-batches of $100$ samples and one critic update per iteration ($q=1$). The offline data are collected in advance from the source and DK policies and require no additional online interactions during learning. The reused-policy pool consists of one SLDAC source policy checkpoint and one DK policy, with uniform initial reuse probabilities over the new policy and the reused policies. The step-size powers are set as $\alpha_t=\frac{1}{t^{0.6}}$, $\gamma_t=\frac{1}{t^{0.3}}$, $\beta_t^{(\rho)}=\frac{1}{t^{0.8}}$, and $\beta_t^{(\psi)}=\frac{1}{t^{0.7}}$.

\begin{figure}[!t]
\centering
\subfloat[Episode sum-rate.\label{fig:isac_rate}]{\includegraphics[width=0.8\linewidth]{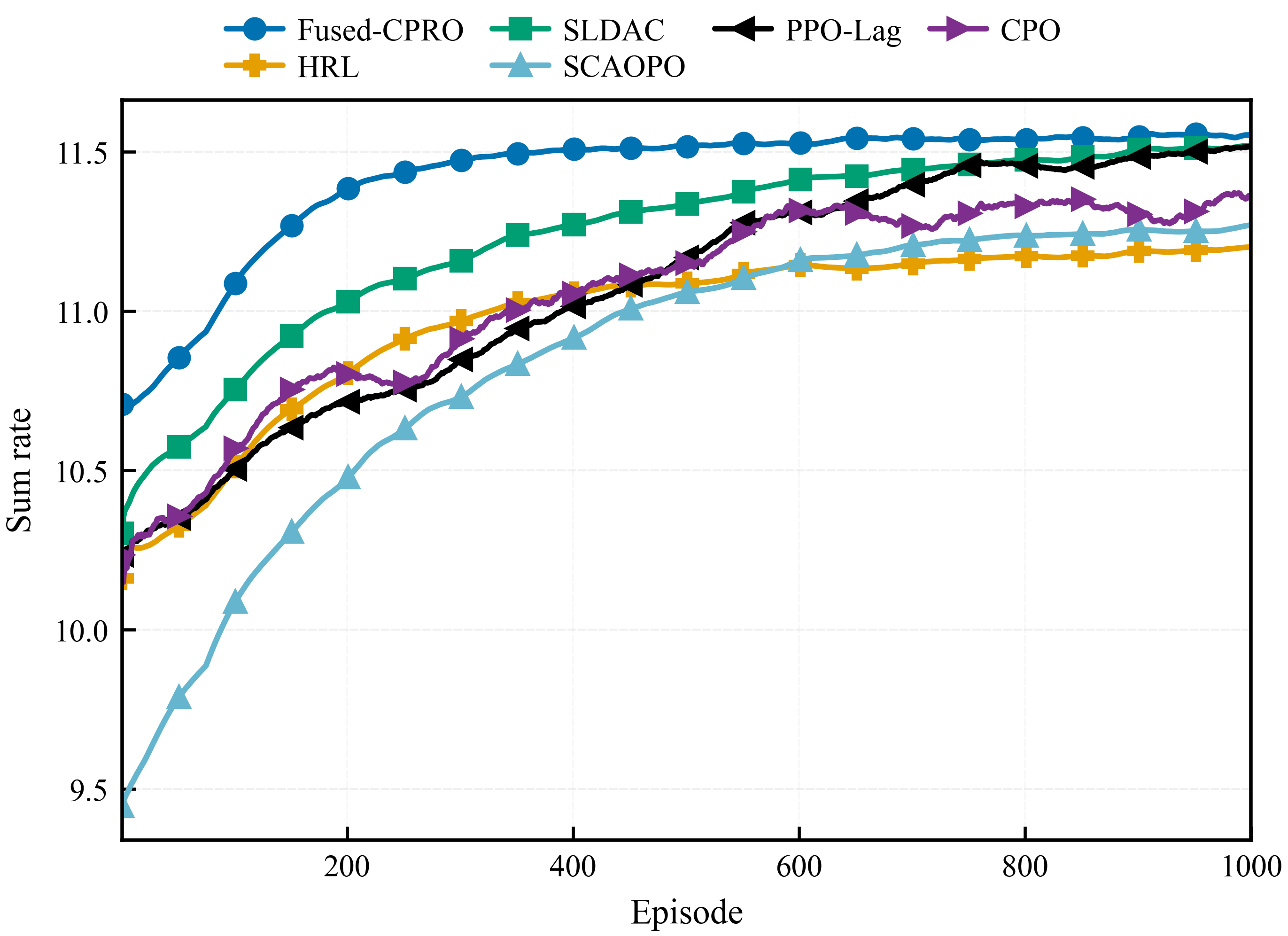}}
\par\smallskip
\subfloat[Mean CRB.\label{fig:isac_crb}]{\includegraphics[width=0.8\linewidth]{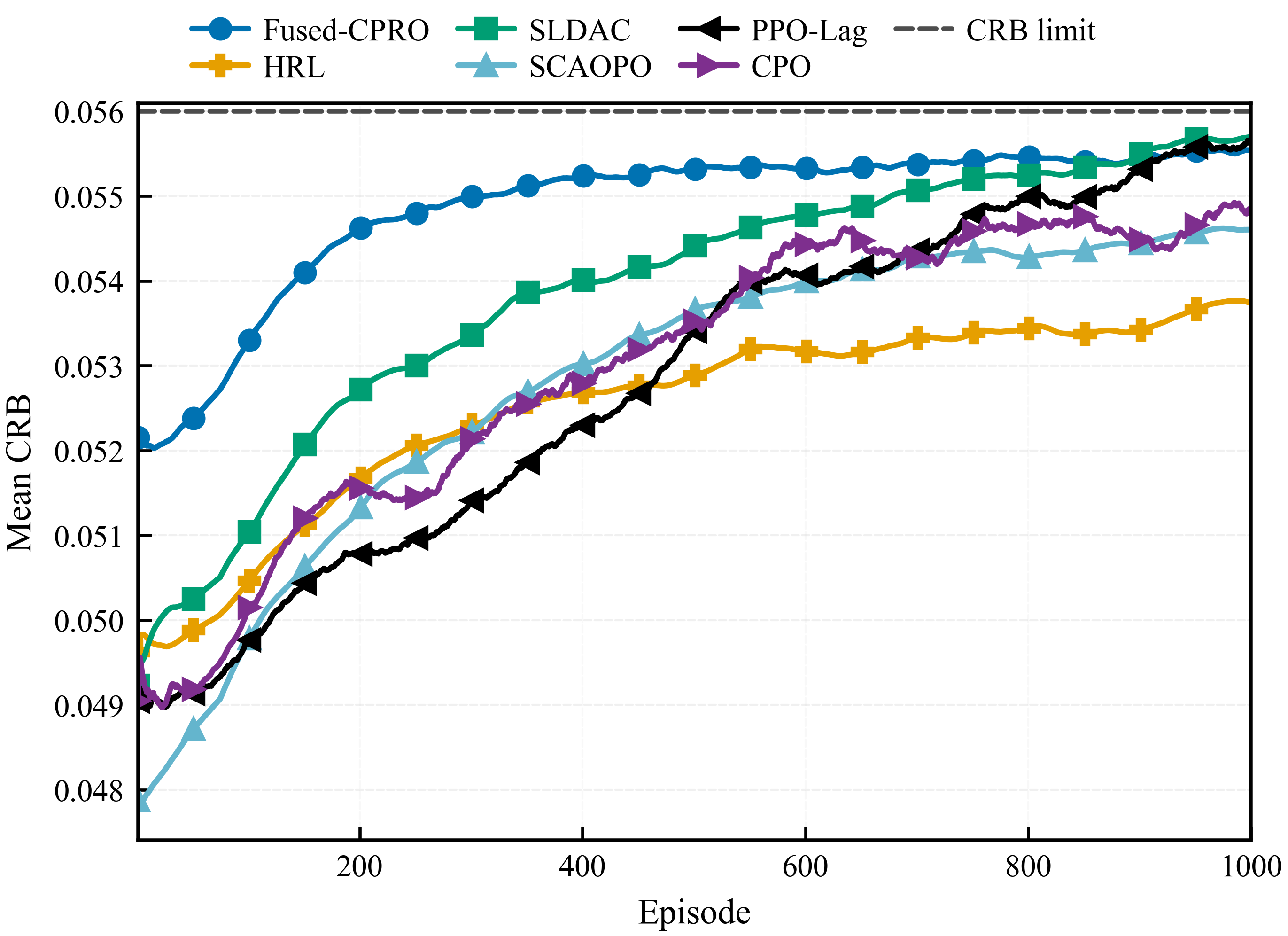}}
\caption{Performance of Fused-CPRO and baseline algorithms in the MIMO-ISAC beamforming task.}
\label{fig:isac_main}
\end{figure}

Fig.~\ref{fig:isac_rate} shows that Fused-CPRO reaches a high sum-rate region earlier than the baselines, while Fig.~\ref{fig:isac_crb} confirms that the mean CRB remains below the sensing-accuracy limit. Together, these results demonstrate that Fused-CPRO simultaneously improves communication performance and maintains reliable sensing constraint satisfaction in the more complex ISAC setting, where the state depends on past sensing actions and the constraints involve non-convex CRB functions. Compared with SLDAC, the faster rise mainly comes from policy reuse and offline data, which provide a useful warm start and more informative critic estimates. As in the MU-MIMO task, Fused-CPRO's early-stage advantage directly reflects the reduced number of online interactions required to reach a high-performance region. Compared with HRL, the gain comes from using the offline data in critic learning, which replaces high-variance Monte Carlo returns with reusable value estimates and enables more stable policy optimization in the high-dimensional ISAC action space.

\begin{figure}[!t]
\centering
\includegraphics[width=0.85\columnwidth]{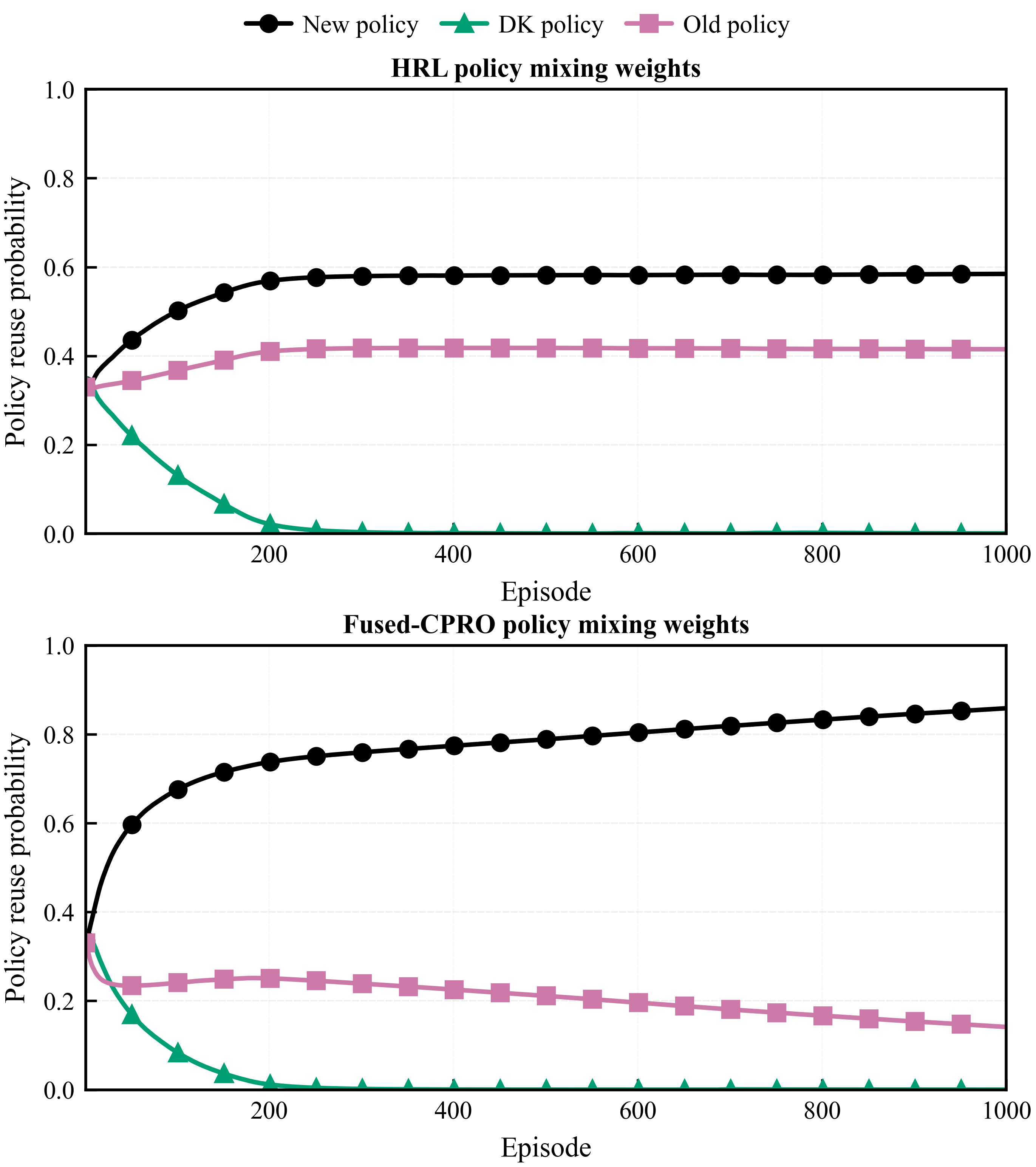}
\caption{Evolution of the reuse probabilities in the MIMO-ISAC task.}
\label{fig:isac_mix}
\end{figure}

The reuse probabilities in Fig.~\ref{fig:isac_mix} are consistent with this explanation. Fused-CPRO uses the source and DK policies mainly in the initial stage and then transfers probability mass to the new actor. This early reuse explains the convergence-speed advantage over SLDAC, which trains a single target policy from online data. The faster shift away from reused policies also explains the advantage over HRL: with critic-side offline data reuse, Fused-CPRO can exploit source experience early without keeping the target actor tied to the reused policies for too long, achieving both faster learning and better final performance in the ISAC beamforming task.

\section{Conclusion}
\label{sec:conclusion}
We proposed Fused-CPRO, a knowledge-fused constrained policy reuse optimization method for wireless resource allocation, motivated by the need for fast-converging, sample-efficient algorithms in dynamic environments with costly online interactions. Fused-CPRO constructs the allocation policy as a mixture of a learnable target policy, source policies, and DK policies, jointly optimizing the reuse probabilities and target actor under a CMDP framework. A CSSCA actor handles non-convex stochastic objectives and long-term constraints, while a critic trained from mixed offline-online data reuses pre-collected experience to reduce online exploration. We proved almost-sure convergence to a KKT point, providing the theoretical guarantee essential for practical deployment. Simulations on delay-constrained MU-MIMO power control and CRB-constrained MIMO-ISAC beamforming demonstrate that Fused-CPRO improves empirical performance and converges substantially faster than representative baselines. The results validate that fusing heterogeneous knowledge and reusing offline data effectively reduces the number of online interactions required to reach a feasible and high-performing policy, while the CSSCA actor reliably enforces long-term constraints throughout learning.

\appendices
\section{Technical Bounds about DNN}
\label{sec:Technical-Bounds-DNN}
The appendices collect the auxiliary technical results and full proof details referenced in Section~\ref{sec:CONVERGENCE-ANALYSIS}. Appendix~\ref{sec:Technical-Bounds-DNN} states the DNN local-linearization bounds, Appendix~\ref{app:sa-tracking} states the projected stochastic-approximation tracking lemma, and Appendices~\ref{sec:proof-of-critic}--\ref{sec:proof-of-surrogate} give the complete proofs of the critic-tracking and surrogate-consistency results.

The convergence analysis uses the following DNN local-linearization lemma \cite{Cao2019b,DQlearning, Allen2019b}.
\begin{lemma}[Technical Bounds about DNN]~
\label{lemma:technical-bounds-dnn}
Let $d$ be the input dimension, and let $\{a_{i}\}_{i=0,1,\ldots}$ denote universal constants that are independent of the proposed algorithm parameters. For any $\sigma\in(0,1)$, if the radius of the constraint set satisfies
\begin{align}
a_{1}d^{3/2}L^{-1}m_{Q}^{-3/4}\leq R_{\boldsymbol{\omega}}\leq a_{2}L^{-6}(\log m_{Q})^{-3},
\end{align}
and the width of the DNN satisfies
\begin{align}
m_{Q}\geq a_{3}\max\left\{ dL^{2}\log\left(\frac{m_{Q}}{\sigma}\right),\,R_{\boldsymbol{\omega}}^{-4/3}L^{-8/3}\log\left(\frac{m_{Q}}{R_{\boldsymbol{\omega}}\sigma}\right)\right\},
\end{align}
it holds that the bias between $f(\boldsymbol{\omega})$ and its local linearization $\hat{f}(\boldsymbol{\omega})$ satisfies
\begin{align}
\left|f(\boldsymbol{\omega};\boldsymbol{s},\boldsymbol{a})-\hat{f}(\boldsymbol{\omega};\boldsymbol{s},\boldsymbol{a})\right|
\leq a_{4}R_{\boldsymbol{\omega}}^{4/3}L^{4}\sqrt{m_{Q}\log m_{Q}}+a_{5}R_{\boldsymbol{\omega}}^{2}L^{5},
\end{align}
with probability at least $1-\sigma-\exp\{-a_6 m_QR^{2/3}_\omega L\}$.

The gradient of the DNN is also bounded as
\begin{align}
\left\Vert \nabla_{\boldsymbol{\omega}}f(\boldsymbol{\omega};\boldsymbol{s},\boldsymbol{a})\right\Vert _{2}\leq a_{7}m_{Q}^{1/2},
\end{align}
with probability at least $1-L^{2}\exp\{-a_{8}m_{Q}R_{\boldsymbol{\omega}}^{2/3}L\}$. In the theoretical analysis, all critics are initialized with the same distribution $\mathcal{N}(0,1/m^2_Q)$.

\end{lemma}

\section{Projected Stochastic Approximation Tracking Lemma}
\label{app:sa-tracking}
The proof also uses the following projected stochastic-approximation tracking lemma from \cite{Lemma3}.
\begin{lemma}[A projected stochastic approximation tracking lemma]
\label{lemma:appC_SA_tracking}
Let $\{\mathcal{F}_t\}_{t\ge 0}$ be an increasing sequence of $\sigma$-fields. Let $\{\boldsymbol{z}^{t}\}$ and $\{\boldsymbol{w}^{t}\}$ be $\mathcal{F}_{t}$-measurable random vectors satisfying, for $t\ge 1$,
\begin{equation}
\boldsymbol{w}^{t}=\Pi_{\mathcal{W}}\!\Big(\boldsymbol{w}^{t-1}+\alpha_t\big(\boldsymbol{\varrho}^{t}-\boldsymbol{w}^{t-1}\big)\Big),
\label{eq:appC_SA_recursion}
\end{equation}
where $\Pi_{\mathcal{W}}$ is the projection onto a convex and closed set $\mathcal{W}$.
Assume that the following conditions hold:
\begin{enumerate}
\item all accumulation points of the target sequence $\{\boldsymbol{z}^{t}\}$ belong to $\mathcal{W}$ w.p.1;
\item there exists a constant $C<\infty$ such that $\mathbb{E}[\|\boldsymbol{\varrho}^{t}\|_2^2\mid\mathcal{F}_{t-1}]\le C$ a.s. for all $t\ge 1$;
\item there exists a bias sequence $\{\boldsymbol{o}^{t}\}$ such that, for $t\ge 1$,
\begin{equation}
\begin{aligned}
\mathbb{E}[\boldsymbol{\varrho}^{t}\mid\mathcal{F}_{t-1}] = \boldsymbol{z}^{t-1}+\boldsymbol{o}^{t-1},\\
\sum_{t=1}^{\infty}\mathbb{E}\!\left[(\alpha_t)^2+\alpha_t\,\|\boldsymbol{o}^{t-1}\|_2\right]<\infty;
\end{aligned}
\label{eq:appC_cond_mean_a}
\end{equation}
\item $\alpha_t>0$, $\sum_{t=1}^{\infty}\alpha_t=\infty$, and $\sum_{t=1}^{\infty}\alpha_t^2<\infty$;
\item the target drifts slowly:
\begin{equation}
\frac{\|\boldsymbol{z}^{t}-\boldsymbol{z}^{t-1}\|_2}{\alpha_t}\to 0,\qquad \text{w.p.1.}
\label{eq:appC_target_drift}
\end{equation}
\end{enumerate}
Then $\|\boldsymbol{w}^{t}-\boldsymbol{z}^{t}\|_2\to 0$ almost surely.
\end{lemma}

\section{\texorpdfstring{Proof of Lemma~\ref{lemma-convergence-rate-critic}}{Proof of the critic convergence lemma}}
\label{sec:proof-of-critic}
Fix an arbitrary $i\in\{0,1,\ldots,I\}$. To keep notation light, suppress the index $i$ on $\bar{\boldsymbol{\omega}}_{t}^{i}$, $\dot{\boldsymbol{\omega}}_{t}^{i}$, $\Omega_{i}$, and $\hat{Q}_{i}^{\pi_{\boldsymbol{\theta}_t}}$, while keeping the critic-error notation explicit. The critic error in \eqref{eq:critic_error} is
\begin{equation}
\epsilon_{t,i}^{\mathrm{cri}}
\triangleq\bigl|\,\mathbb{E}_{p_t}\big[f(\bar{\boldsymbol{\omega}}_{t};\boldsymbol{s},\boldsymbol{a})\,\big|\,\bar{\boldsymbol{\omega}}_{t-1}\big]-\hat Q^{\pi_{\boldsymbol{\theta}_t}}(\boldsymbol{s},\boldsymbol{a})\,\bigr|.
\end{equation}
With the local linearization \eqref{eq:local_linearization} and the triangle inequality,
\begin{align}
\epsilon_{t,i}^{\mathrm{cri}}
&\le \underbrace{\bigl|\,\mathbb{E}[f(\bar{\boldsymbol{\omega}}_{t})]-\mathbb{E}[\hat f(\bar{\boldsymbol{\omega}}_{t})]\,\bigr|}_{\mathrm{(bias\ 1)}}
+\bigl|\,\mathbb{E}[\hat f(\bar{\boldsymbol{\omega}}_{t})]-\hat Q^{\pi_{\boldsymbol{\theta}_t}}\,\bigr|.
\label{eq:appB_decomp1}
\end{align}
By the local linearization error bound provided in Lemma~\ref{lemma:technical-bounds-dnn}, $\mathrm{bias\ 1}=\mathcal{O}(\epsilon_{m_Q})$.

According to Assumption~\ref{assumption:target_Q}.1, there exists a point $\dot{\boldsymbol{\omega}}_{t}\in\Omega$ such that
$\hat f(\dot{\boldsymbol{\omega}}_{t})\equiv \hat Q^{\pi_{\boldsymbol{\theta}_t}}$.
Thus the second term in \eqref{eq:appB_decomp1} becomes
\begin{equation}
\bigl|\,\mathbb{E}[\hat f(\bar{\boldsymbol{\omega}}_{t})]-\hat Q^{\pi_{\boldsymbol{\theta}_t}}\,\bigr|
=\bigl|\,\mathbb{E}[\hat f(\bar{\boldsymbol{\omega}}_{t})]-\hat f(\dot{\boldsymbol{\omega}}_{t})\,\bigr|.
\label{eq:appB_decomp2}
\end{equation}
Using the linear identity induced by \eqref{eq:local_linearization},
$\hat f(\boldsymbol{\omega}_a)-\hat f(\boldsymbol{\omega}_b)=\langle\nabla_{\boldsymbol{\omega}}f(\boldsymbol{\omega}_0;\cdot),\,\boldsymbol{\omega}_a-\boldsymbol{\omega}_b\rangle$,
which yields
\begin{equation}
\bigl|\,\mathbb{E}[\hat f(\bar{\boldsymbol{\omega}}_{t})]-\hat f(\dot{\boldsymbol{\omega}}_{t})\,\bigr|
\le \bigl\|\nabla_{\boldsymbol{\omega}}f(\boldsymbol{\omega}_0)\bigr\|_2\cdot\bigl\|\mathbb{E}[\bar{\boldsymbol{\omega}}_{t}]-\dot{\boldsymbol{\omega}}_{t}\bigr\|_2.
\label{eq:appB_reduce_to_param}
\end{equation}
Lemma~\ref{lemma:technical-bounds-dnn} bounds $\|\nabla_{\boldsymbol{\omega}}f(\boldsymbol{\omega}_0)\|_2$, hence
\begin{equation}
\epsilon_{t,i}^{\mathrm{cri}}
\le \mathcal{O}(\epsilon_{m_Q})+\mathcal{O}(\sqrt{m_Q})\cdot\bigl\|\mathbb{E}[\bar{\boldsymbol{\omega}}_{t}]-\dot{\boldsymbol{\omega}}_{t}\bigr\|_2.
\label{eq:appB_core_goal}
\end{equation}
It suffices to bound $\|\mathbb{E}[\bar{\boldsymbol{\omega}}_{t}]-\dot{\boldsymbol{\omega}}_{t}\|_2$.

We fix $\kappa_6\in(0,1)$ and define $n_t\triangleq t-\lceil t^{\kappa_6}\rceil$ to construct the frozen-window reference critic trajectory. For a time step $t$, let $\{\tilde{\varepsilon}_{k}^{(t)}\}_{k=1}^{t}$ denote the surrogate observation sequence, where
\begin{equation}
\tilde{\varepsilon}_{k}^{(t)}=\bigl(\tilde{\boldsymbol{s}}_{k},\tilde{\boldsymbol{a}}_{k},\{C_{i}^{\text{'}}(\tilde{\boldsymbol{s}}_{k},\tilde{\boldsymbol{a}}_{k})\}_{i=0}^{I},\tilde{\boldsymbol{s}}_{k+1}\bigr).
\end{equation}
Specifically, the reference trajectory matches the actual critic recursion up to time $n_t$, and then evolves under the frozen policy $\pi_{\theta_{n_t}}$ over the window $\{n_t+1,\ldots,t\}$. For each $i$, $\boldsymbol{m}_{k}^{i}$ is the auxiliary critic parameter obtained by applying the projected TD recursion to the local linearization $\hat{f}$ along a surrogate observation sequence. Its moving-average counterpart $\bar{\boldsymbol{m}}_{k}^{i}$ mirrors the target-critic recursion of $\bar{\boldsymbol{\omega}}_{k}^{i}$.
\begin{equation}
\begin{aligned}
\boldsymbol{m}_{k}^{i}
&=\Pi_{\Omega_{i}}\!\left(\boldsymbol{m}_{k-1}^{i}-\eta_{k}\boldsymbol{M}_{k-1}^{i}\right),\\
\bar{\boldsymbol{m}}_{k}^{i}
&=(1-\gamma_{k})\bar{\boldsymbol{m}}_{k-1}^{i}+\gamma_{k}\boldsymbol{m}_{k}^{i},
\end{aligned}
\label{eq:appB_aux_update}
\end{equation}
where $\boldsymbol{M}_{k-1}^{i}$ is the surrogate TD-gradient induced by the surrogate trajectory $\{\tilde\varepsilon_{k}^{(t)}\}$:
\begin{align}
\boldsymbol{M}_{k-1}^{i} & =\Bigl(\hat{f}(\boldsymbol{m}_{k-1}^{i};\tilde{\boldsymbol{s}}_{k},\tilde{\boldsymbol{a}}_{k})-\bigl(C_{i}^{\text{'}}(\tilde{\boldsymbol{s}}_{k},\tilde{\boldsymbol{a}}_{k})-\hat{J}_{i}^{k-1}\nonumber \\
& +\hat{f}(\boldsymbol{m}_{k-1}^{i};\tilde{\boldsymbol{s}}_{k+1},\tilde{\boldsymbol{a}}_{k+1}')\bigr)\Bigr)\nabla_{\boldsymbol{\omega}}f(\boldsymbol{\omega}_{0}^{i};\tilde{\boldsymbol{s}}_{k},\tilde{\boldsymbol{a}}_{k}),
\end{align}
Suppress the superscript $i$ below. Adding and subtracting $\mathbb{E}[\bar{\boldsymbol{m}}_{t}]$ decomposes $\|\mathbb{E}[\bar{\boldsymbol{\omega}}_{t}]-\dot{\boldsymbol{\omega}}_{t}\|_2$ as
\begin{equation}
\bigl\|\mathbb{E}[\bar{\boldsymbol{\omega}}_{t}]-\dot{\boldsymbol{\omega}}_{t}\bigr\|_2^2
\le 2\underbrace{\bigl\|\mathbb{E}[\bar{\boldsymbol{m}}_{t}]-\dot{\boldsymbol{\omega}}_{t}\bigr\|_2^2}_{\mathrm{(bias\ 2)}}
+2\underbrace{\bigl\|\mathbb{E}[\bar{\boldsymbol{\omega}}_{t}]-\mathbb{E}[\bar{\boldsymbol{m}}_{t}]\bigr\|_2^2}_{\mathrm{(bias\ 3)}}.
\label{eq:appB_bias23_split}
\end{equation}
where the first term is the finite-time TD tracking error under the frozen policy, and the second term captures the policy-drift error. The next two subsections bound bias~2 and bias~3.

\subsection{Bounds about bias~2}
Unfolding the recursion $\bar{\boldsymbol{m}}_{k}=(1-\gamma_{k})\bar{\boldsymbol{m}}_{k-1}+\gamma_{k}\boldsymbol{m}_{k}$, there exist weights
$w_{t,t'}\triangleq \gamma_{t'}\prod_{j=t'+1}^{t}(1-\gamma_j)$ for $t'=0,1,\ldots,t$ such that
\begin{equation}
\bar{\boldsymbol{m}}_{t}=\sum_{t'=0}^{t} w_{t,t'}\,\boldsymbol{m}_{t'},\qquad \sum_{t'=0}^{t} w_{t,t'}=1.
\label{eq:appB_bar_m_weights}
\end{equation}
Jensen's inequality gives
\begin{equation}
\bigl\|\mathbb{E}[\bar{\boldsymbol{m}}_{t}]-\dot{\boldsymbol{\omega}}_{t}\bigr\|_2^2
\le \sum_{t'=0}^{t} w_{t,t'}\,\bigl\|\mathbb{E}[\boldsymbol{m}_{t'}]-\dot{\boldsymbol{\omega}}_{t}\bigr\|_2^2.
\label{eq:appB_jensen_bias2}
\end{equation}
Split the sum into early part $t'\le n_t$ and recent part $t'\in\{n_t+1,\ldots,t\}$. Since $\boldsymbol{m}_{t'}\in\Omega=\mathbb{B}(\boldsymbol{\omega}_0,R_{\boldsymbol{\omega}})$ and $\dot{\boldsymbol{\omega}}_{t}\in\Omega$,
$\|\mathbb{E}[\boldsymbol{m}_{t'}]-\dot{\boldsymbol{\omega}}_{t}\|_2\le 2R_{\boldsymbol{\omega}}$. Moreover, using the monotonicity of $\{\gamma_t\}$,
\begin{equation}
\sum_{t'=0}^{n_t} w_{t,t'}
\le \prod_{j=n_t+1}^{t}(1-\gamma_j)
\le (1-\gamma_t)^{t-n_t}
\le (1-\gamma_t)^{t^{\kappa_6}}.
\label{eq:appB_weight_early}
\end{equation}
Thus,
\begin{equation}
\sum_{t'=0}^{n_t} w_{t,t'}\,\bigl\|\mathbb{E}[\boldsymbol{m}_{t'}]-\dot{\boldsymbol{\omega}}_{t}\bigr\|_2^2
\le 4R_{\boldsymbol{\omega}}^{2}(1-\gamma_t)^{t^{\kappa_6}}
=\mathcal{O}\!\left(\frac{(1-\gamma_t)^{t^{\kappa_6}}}{m_Q\,\gamma_t}\right),
\label{eq:appB_bias2_early}
\end{equation}
where we used the standard NTK scaling $R_{\boldsymbol{\omega}}=\mathcal{O}(m_Q^{-1/2})$ in Lemma~\ref{lemma:technical-bounds-dnn} and the fact that $\lim_{t\to\infty}\gamma_t\to 0$.

Define the tracking error in the frozen window as
$\boldsymbol{e}_{k}\triangleq \mathbb{E}[\boldsymbol{m}_{k}]-\dot{\boldsymbol{\omega}}_{t}$ for $k\in\{n_t+1,\ldots,t\}$.
The error can be expanded using the non-expansive property of the projection operation:
\begin{equation}
\begin{aligned}
\|\boldsymbol{e}_{k+1}\|_2^2
&=\bigl\|\mathbb{E}[\boldsymbol{m}_{k+1}]-\dot{\boldsymbol{\omega}}_{t}\bigr\|_2^2\\
&\le \bigl\|\boldsymbol{e}_{k}-\eta_{k+1}\,\mathbb{E}[\boldsymbol{M}_{k}(\boldsymbol{m}_{k})]\bigr\|_2^2\\
&=\|\boldsymbol{e}_{k}\|_2^2-2\eta_{k+1}\big\langle \mathbb{E}[\boldsymbol{M}_{k}(\boldsymbol{m}_{k})],\,\boldsymbol{e}_{k}\big\rangle\\
&\quad+\eta_{k+1}^2\bigl\|\mathbb{E}[\boldsymbol{M}_{k}(\boldsymbol{m}_{k})]\bigr\|_2^2.
\end{aligned}
\label{eq:appB_recursion_raw}
\end{equation}
Decompose the inner product $\big\langle \mathbb{E}[\boldsymbol{M}_{k}(\boldsymbol{m}_{k})],\,\boldsymbol{e}_{k}\big\rangle$ as
\begin{equation}
\begin{aligned}
\big\langle \mathbb{E}[\boldsymbol{M}_{k}(\boldsymbol{m}_{k})],\,\boldsymbol{e}_{k}\big\rangle
&=\big\langle \mathbb{E}[\boldsymbol{M}_{k}(\boldsymbol{m}_{k})-\boldsymbol{M}_{k}(\dot{\boldsymbol{\omega}}_{t})],\,\boldsymbol{e}_{k}\big\rangle\\
&\quad+\big\langle \mathbb{E}[\boldsymbol{M}_{k}(\dot{\boldsymbol{\omega}}_{t})],\,\boldsymbol{e}_{k}\big\rangle.
\end{aligned}
\label{eq:appB_add_subtract}
\end{equation}
The first term is controlled by Assumption~\ref{assumption:target_Q}.2:
there exists $\varsigma>0$ such that
\begin{equation}
\big\langle \mathbb{E}[\boldsymbol{M}_{k}(\boldsymbol{m}_{k})-\boldsymbol{M}_{k}(\dot{\boldsymbol{\omega}}_{t})],\,\boldsymbol{e}_{k}\big\rangle
\ge \frac{\varsigma}{2}\,\|\boldsymbol{e}_{k}\|_2^2.
\label{eq:appB_strong_mono}
\end{equation}
The last term in \eqref{eq:appB_add_subtract} is more challenging. In SLDAC, the fixed point satisfies
$\mathbb{E}[\boldsymbol{M}_{k}(\dot{\boldsymbol{\omega}}_{t})]=\boldsymbol{0}$.
However, owing to the use of the mixed offline/online data, $\mathbb{E}[\boldsymbol{M}_{k}(\dot{\boldsymbol{\omega}}_{t})]$ is no longer exactly zero.
We decompose the surrogate mixed TD gradient as
\begin{equation}
\begin{aligned}
\boldsymbol{M}_{k}^{\mathrm{mix}}(\boldsymbol{\omega})
&\triangleq (1-\xi_{k+1})\,\boldsymbol{M}_{k}^{\mathrm{on}}(\boldsymbol{\omega})\\
&\quad+\xi_{k+1}\,\boldsymbol{M}_{k}^{\mathrm{off}}(\boldsymbol{\omega}),\\
\boldsymbol{M}_{k}(\cdot)
&\equiv \boldsymbol{M}_{k}^{\mathrm{mix}}(\cdot).
\end{aligned}
\label{eq:appB_mix_decomp}
\end{equation}
Since $\dot{\boldsymbol{\omega}}_{t}$ is the projection fixed point of the frozen online operator (for $k\in\{n_t+1,\ldots,t\}$), we have
$\mathbb{E}[\boldsymbol{M}_{k}^{\mathrm{on}}(\dot{\boldsymbol{\omega}}_{t})]=\boldsymbol{0}$. Hence,
\begin{equation}
\begin{aligned}
\mathbb{E}[\boldsymbol{M}_{k}(\dot{\boldsymbol{\omega}}_{t})]
&=(1-\xi_{k+1})\underbrace{\mathbb{E}[\boldsymbol{M}_{k}^{\mathrm{on}}(\dot{\boldsymbol{\omega}}_{t})]}_{=\boldsymbol{0}}
+\xi_{k+1}\,\mathbb{E}[\boldsymbol{M}_{k}^{\mathrm{off}}(\dot{\boldsymbol{\omega}}_{t})]\\
&=\xi_{k+1}\,\boldsymbol{b}_{\mathrm{off}}(t),
\qquad
\boldsymbol{b}_{\mathrm{off}}(t)\triangleq \mathbb{E}[\boldsymbol{M}_{k}^{\mathrm{off}}(\dot{\boldsymbol{\omega}}_{t})].
\end{aligned}
\label{eq:appB_offline_drift}
\end{equation}
By boundedness of the gradient in Lemma~\ref{lemma:technical-bounds-dnn}, there exists a constant $a_{11}>0$ such that
\begin{equation}
\|\boldsymbol{b}_{\mathrm{off}}(t)\|_2\le a_{11}m_{Q}^{1/2}.
\label{eq:appB_boff_bound}
\end{equation}
Applying Young's inequality to the drift inner product yields, for any $\varsigma>0$,
\begin{flalign}
&2\eta_{k+1}\big|\langle \mathbb{E}[\boldsymbol{M}_{k}(\dot{\boldsymbol{\omega}}_{t})],\,\boldsymbol{e}_{k}\rangle\big| &&\nonumber\\
&\le \eta_{k+1}\frac{\varsigma}{2}\,\|\boldsymbol{e}_{k}\|_2^2+\frac{2\eta_{k+1}}{\varsigma}\,\bigl\|\mathbb{E}[\boldsymbol{M}_{k}(\dot{\boldsymbol{\omega}}_{t})]\bigr\|_2^2 &&\nonumber\\
&=\eta_{k+1}\frac{\varsigma}{2}\,\|\boldsymbol{e}_{k}\|_2^2+\mathcal{O}\!\left(\eta_{k+1}\,m_Q\,\xi_{n_t}^{2}\right). &&\label{eq:appB_young_drift}
\end{flalign}

Using \eqref{eq:appB_add_subtract}--\eqref{eq:appB_young_drift} in \eqref{eq:appB_recursion_raw} gives
\begin{equation}
\begin{aligned}
\eta_{k+1}\|\boldsymbol{e}_{k}\|_2^2
&\le \mathcal{O}(\|\boldsymbol{e}_{k}\|_2^2-\|\boldsymbol{e}_{k+1}\|_2^2)+\mathcal{O}(m_Q\eta_{k+1}^2)\\
&\quad+\mathcal{O}(m_Q\eta_{k+1}\xi_{n_t}^{2}).
\end{aligned}
\label{eq:appB_telescope_form}
\end{equation}
Summing \eqref{eq:appB_telescope_form} for $k=n_t+1,\ldots,t$, using monotonicity of $\{\eta_k\}$ and $\{\xi_k\}$, and substituting the result into \eqref{eq:appB_jensen_bias2} gives
\begin{equation}
\begin{aligned}
\sum_{t'=n_t+1}^{t} w_{t,t'}\,\bigl\|\mathbb{E}[\boldsymbol{m}_{t'}]-\dot{\boldsymbol{\omega}}_{t}\bigr\|_2^2
&\le \frac{\gamma_{n_t}}{\eta_{n_t}}\cdot
\mathcal{O}\Bigl(
\frac{1}{m_Q}
+m_Q\,t^{\kappa_6}\eta_{n_t}^{2}\\
&\qquad
+m_Q\,t^{\kappa_6}\eta_{n_t}\xi_{n_t}^{2}
\Bigr).
\end{aligned}
\label{eq:appB_bias2_recent}
\end{equation}
Combining the early-window and recent-window estimates gives the bound for bias~2:
\begin{equation}
\begin{aligned}
bias~2
&\le \mathcal{O}\Bigl(
\frac{(1-\gamma_t)^{t^{\kappa_6/2}}}{\sqrt{\gamma_t}}
+\sqrt{\tfrac{\gamma_{n_t}}{\eta_{n_t}}}\\
&\qquad
+m_Q\sqrt{\gamma_{n_t}\eta_{n_t}}\,t^{\kappa_6/2}
+m_Q\sqrt{\gamma_{n_t}}\,t^{\kappa_6/2}\,\xi_{n_t}
\Bigr).
\end{aligned}
\label{eq:appB_bias2_final}
\end{equation}

\subsection{Bounds about bias~3}
The difference $\mathbb{E}[\bar{\boldsymbol{\omega}}_{t}]-\mathbb{E}[\bar{\boldsymbol{m}}_{t}]$ can also be presented as the exponentially weighted averages:
\begin{equation}
\begin{aligned}
\bigl\|\mathbb{E}[\bar{\boldsymbol{\omega}}_{t}]-\mathbb{E}[\bar{\boldsymbol{m}}_{t}]\bigr\|_2
&\le \mathcal{O}\!\left(\frac{\gamma_{n_t}}{\gamma_t}\right)\cdot e_{n_t}^{\mathrm{b}},\\
e_{n_t}^{\mathrm{b}}
&\triangleq \max_{k\in\{n_t+1,\ldots,t\}}\bigl\|\mathbb{E}[\boldsymbol{\omega}_{k}]-\mathbb{E}[\boldsymbol{m}_{k}]\bigr\|_2.
\end{aligned}
\label{eq:appB_bias3_reduce}
\end{equation}
Using the non-expansiveness of projection operation and the recursions formulation in \eqref{eq:TD-learning} and \eqref{eq:appB_aux_update}, it can be expanded as:
\begin{equation}
\begin{aligned}
\bigl\|\mathbb{E}[\boldsymbol{\omega}_{k}]-\mathbb{E}[\boldsymbol{m}_{k}]\bigr\|_2
&\le \sum_{j=n_t+1}^{k}\eta_j\cdot
\Bigl\|\mathbb{E}[\boldsymbol{\Delta}(\boldsymbol{\omega}_{j-1})] \\
&\qquad\qquad
-\mathbb{E}[\boldsymbol{M}^{\mathrm{mix}}_{j-1}(\boldsymbol{m}_{j-1})]\Bigr\|_2.
\end{aligned}
\label{eq:appB_bias3_sumgrad}
\end{equation}
Define the TD-gradient integrands
\begin{equation}
\begin{aligned}
H_f(\boldsymbol{\omega};\varepsilon)
&\triangleq \delta(\boldsymbol{\omega};\varepsilon)\,
\nabla_{\boldsymbol{\omega}}f(\boldsymbol{\omega};\boldsymbol{s},\boldsymbol{a}),\\
H_{\hat f}(\boldsymbol{m};\varepsilon)
&\triangleq \hat\delta(\boldsymbol{m};\varepsilon)\,
\nabla_{\boldsymbol{\omega}}\hat f(\boldsymbol{m};\boldsymbol{s},\boldsymbol{a}).
\end{aligned}
\end{equation}

where $\varepsilon=(\boldsymbol{s},\boldsymbol{a},\boldsymbol{s}',\boldsymbol{a}')$, and $\boldsymbol{a}'\sim\pi_{\boldsymbol{\theta}_t}(\cdot\mid\boldsymbol{s}')$.
Let $\mu_j^{\mathrm{on}}$ denote distribution of online data generated under $\pi_{\boldsymbol{\theta}_j}$, while $\tilde\mu_t^{\mathrm{on}}$ denote online data's distribution under the frozen policy $\pi_{\boldsymbol{\theta}_t}$, and let $\nu^{\mathrm{off}}$ denote the offline data distribution. Then \eqref{eq:appB_bias3_sumgrad} can be decomposed as:
\begin{equation}
\begin{aligned}
&\bigl\|\mathbb{E}[\boldsymbol{\Delta}_{j-1}]-\mathbb{E}[\boldsymbol{M}^{\mathrm{mix}}_{j-1}(\boldsymbol{m}_{j-1})]\bigr\|_2
\\
&\le (1-\xi_j)\bigl\|\mathbb{E}_{\mu_j^{\mathrm{on}}}\!\big[H_f(\boldsymbol{\omega}_{j-1};\varepsilon)\big]
-\mathbb{E}_{\tilde\mu_t^{\mathrm{on}}}\!\big[H_f(\boldsymbol{\omega}_{j-1};\varepsilon)\big]\bigr\|_2
\\
&\quad+\bigl\|\mathbb{E}_{\tilde\mu_t^{\mathrm{on}}}\!\big[H_f(\boldsymbol{\omega}_{j-1};\varepsilon)-H_{\hat f}(\boldsymbol{m}_{j-1};\varepsilon)\big]\bigr\|_2
\\
&\quad+\xi_j\bigl\|\mathbb{E}_{\nu^{\mathrm{off}}}\!\big[H_f(\boldsymbol{\omega}_{j-1};\varepsilon)-H_{\hat f}(\boldsymbol{m}_{j-1};\varepsilon)\big]\bigr\|_2.
\label{eq:appB_grad_decomp}
\end{aligned}
\end{equation}

We further obtain
\begin{equation}
\begin{aligned}
&\bigl\|\mathbb{E}[\boldsymbol{\Delta}_{j-1}]
-\mathbb{E}[\boldsymbol{M}^{\mathrm{mix}}_{j-1}(\boldsymbol{m}_{j-1})]\bigr\|_2\\
&\le m_Q^{1/2}\cdot\mathcal{O}\!\Bigl(
\|\mu_j^{\mathrm{on}}-\tilde\mu_t^{\mathrm{on}}\|_{\mathrm{TV}}
+\|\boldsymbol{\theta}_j-\boldsymbol{\theta}_t\|_2
\Bigr)\\
&\quad+\mathcal{O}(\epsilon_{m_Q}).
\end{aligned}
\label{eq:appB_grad_tv}
\end{equation}

Following the ergodicity assumption,
\begin{equation}
\|\mu_j^{\mathrm{on}}-\tilde\mu_t^{\mathrm{on}}\|_{\mathrm{TV}}
\le \lambda\chi^{\tau_j}
+\mathcal{O}\!\left(\sum_{k=j-\tau_j+1}^{t}\beta_k\right),
\label{eq:appB_tv_bound}
\end{equation}
where $\tau_j$ is the mixing-time parameter and $\chi^{\tau_j}=\mathcal{O}(1/t)$. Moreover, by Lipschitz continuity of $\pi_{\boldsymbol{\theta}}$,
\begin{equation}
\|\boldsymbol{\theta}_j-\boldsymbol{\theta}_t\|_2
\le \sum_{k=j+1}^{t}\beta_k
\le \mathcal{O}\!\bigl(\beta_{n_t}t^{\kappa_6}\bigr),
\label{eq:appB_theta_drift}
\end{equation}
where the last step uses $j\ge n_t+1=t-\Theta(t^{\kappa_6})$. Substituting \eqref{eq:appB_tv_bound}--\eqref{eq:appB_theta_drift} into \eqref{eq:appB_grad_tv}, and then into \eqref{eq:appB_bias3_sumgrad}, yields
\begin{equation}
\begin{aligned}
e_{n_t}^{\mathrm{b}}
&\le m_Q^{1/2}\eta_{n_t}\cdot\mathcal{O}\!\Bigl(
t^{\kappa_6-1}+\beta_{n_t}t^{2\kappa_6}
\Bigr)
+\mathcal{O}(\eta_{n_t}\epsilon_{m_Q}t^{\kappa_6}).
\end{aligned}
\label{eq:appB_ebound}
\end{equation}

Substitution of \eqref{eq:appB_ebound} into \eqref{eq:appB_bias3_reduce} yields
\begin{equation}
\begin{aligned}
Bias~3
&\le \mathcal{O}\Bigl(
m_Q\eta_{n_t}t^{\kappa_6-1}
+m_Q\eta_{n_t}\beta_{n_t}t^{2\kappa_6}
\Bigr)\\
&\quad+\mathcal{O}(\epsilon_{m_Q}).
\end{aligned}
\label{eq:appB_bias3_final}
\end{equation}

Substitution of \eqref{eq:appB_bias2_final} and \eqref{eq:appB_bias3_final} into \eqref{eq:appB_core_goal} yields the critic-error bound
\begin{equation}
\begin{aligned}
\epsilon_{t,i}^{\mathrm{cri}}
&\le \mathcal{O}\Bigl(
\epsilon_{m_Q}
+\frac{(1-\gamma_{t})^{t^{\kappa_{6}/2}}}{\sqrt{\gamma_{t}}}
+\sqrt{\tfrac{\gamma_{n_{t}}}{\eta_{n_{t}}}}\\
&\qquad
+m_Q\sqrt{\gamma_{n_{t}}\eta_{n_{t}}}\,t^{\kappa_{6}/2}
+m_Q\eta_{n_{t}}\,t^{\kappa_{6}-1}
+m_Q\eta_{n_{t}}\beta_{n_{t}}\,t^{2\kappa_{6}}\\
&\qquad
+m_Q\sqrt{\gamma_{n_t}}\,t^{\kappa_6/2}\,\xi_{n_t}
\Bigr),
\end{aligned}
\end{equation}
The resulting estimate is exactly the bound stated in Lemma~\ref{lemma-convergence-rate-critic}.

\section{\texorpdfstring{Proof of Lemma~\ref{lemma-asymptotic-consistency}}{Proof of the asymptotic consistency lemma}}
\label{sec:proof-of-asymptotic-consistency}
The proof applies Lemma~\ref{lemma:appC_SA_tracking} from Section~\ref{app:sa-tracking} to the recursions (\ref{eq:J-average}) and (\ref{eq:g-average}) by verifying its technical conditions.
\subsection{\texorpdfstring{Asymptotic consistency of $\hat{J}_i^t$}{Asymptotic consistency of Ji estimates}}
Rewrite the recursion \eqref{eq:J-average} as
\begin{equation}
\hat{J}_i^{t}=\hat{J}_i^{t-1}+\alpha_t\big(\tilde{J}_i^{t}-\hat{J}_i^{t-1}\big).
\label{eq:appC_J_recursion}
\end{equation}
$\tilde{J}_i^{t}$ is the mixture of online and offline estimates:
\begin{equation}
\tilde{J}_i^{t}=(1-\xi_t)\,\hat{\mathbb{E}}_{\mathcal{D}_0^{t}}\big[C_i^{\text{'}}(\boldsymbol{s},\boldsymbol{a})\big]
+\xi_t\,\hat{\mathbb{E}}_{\mathcal{D}_{\mathrm{off}}^{t}}\big[C_i^{\text{'}}(\boldsymbol{s},\boldsymbol{a})\big].
\end{equation}
Taking conditional expectation given $\mathcal{F}_{t-1}$ yields
\begin{equation}
\mathbb{E}[\tilde{J}_i^{t}\mid\mathcal{F}_{t-1}]=J_i(\boldsymbol{\theta}_{t-1})+o_{J}^{t-1},
\label{eq:appC_J_cond_mean}
\end{equation}
where
\begin{equation}
J_i(\boldsymbol{\theta})\triangleq \mathbb{E}_{\sigma_{\pi_{\boldsymbol{\theta}}}}\big[C_i^{\text{'}}(\boldsymbol{s},\boldsymbol{a})\big]
\label{eq:appC_J_def}
\end{equation}
is the stationary average cost under $\pi_{\boldsymbol{\theta}}$ and $o_{J}^{t-1}$ denotes the bias term to be bounded.

Let $\mu_k$ denote the conditional distribution of $(\boldsymbol{s}_k,\boldsymbol{a}_k)$ given $\mathcal{F}_{t-1}$ for $k\in\{t-T_t+1,\ldots,t\}$.
The online bias can be written as
\begin{equation}
\begin{aligned}
 |b_{J,\mathrm{on}}^{t-1}|
 &\triangleq
 |\mathbb{E}\Big[\hat{\mathbb{E}}_{\mathcal{D}_0^{t}}[C_i^{\text{'}}]\,\Big|\,\mathcal{F}_{t-1}\Big]
 - J_i(\boldsymbol{\theta}_{t-1})| \\
 &= \bigg|
 \frac{1}{T_t}\sum_{k=t-T_t+1}^{t}
 \Big(\mathbb{E}_{\mu_k}[C_i^{\text{'}}]-\mathbb{E}_{\sigma_{\pi_{\boldsymbol{\theta}_{t-1}}}}[C_i^{\text{'}}]\Big)
 \bigg| \\
 &\le
 \frac{1}{T_t}\sum_{k=t-T_t+1}^{t}
 \Big|
 \mathbb{E}_{\mu_k}[C_i^{\text{'}}]-\mathbb{E}_{\sigma_{\pi_{\boldsymbol{\theta}_{t-1}}}}[C_i^{\text{'}}]
 \Big| \\
 &\le
 \frac{2C_{\max}}{T_t}\sum_{k=t-T_t+1}^{t}
 \big\|\mu_k-\sigma_{\pi_{\boldsymbol{\theta}_{t-1}}}\big\|_{\mathrm{TV}}.
\end{aligned}
\label{eq:appC_bJon_tv}
\end{equation}

For each $k$ in the window, apply the triangle inequality
\begin{equation}
\big\|\mu_k-\sigma_{\pi_{\boldsymbol{\theta}_{t-1}}}\big\|_{\mathrm{TV}}
\le
\underbrace{\big\|\mu_k-\sigma_{\pi_{\boldsymbol{\theta}_{k}}}\big\|_{\mathrm{TV}}}_{\text{mixing term}}
+
\underbrace{\big\|\sigma_{\pi_{\boldsymbol{\theta}_{k}}}-\sigma_{\pi_{\boldsymbol{\theta}_{t-1}}}\big\|_{\mathrm{TV}}}_{\text{policy-drift term}}.
\label{eq:appC_tv_split}
\end{equation}
The mixing term is controlled by Assumption~\ref{assumption:problem_structure}, for some constants $\lambda>0$ and $\varrho\in(0,1)$, $\|\mu_k-\sigma_{\pi_{\boldsymbol{\theta}_k}}\|_{\mathrm{TV}}\le \lambda\varrho^{\tau_t}$ with a mixing horizon $\tau_t=\Theta(\log t)$. Since $T_t=\mathcal{O}(\log t)$, we can take $\tau_t\asymp T_t$, so that $\varrho^{\tau_t}=\mathcal{O}(t^{-2})$. For the policy-drift term, Assumption~\ref{assumption:problem_structure} implies Lipschitz continuity of the stationary distribution in total variation:
there exists $L_{\sigma}>0$ such that
$\|\sigma_{\pi_{\boldsymbol{\theta}}}-\sigma_{\pi_{\boldsymbol{\theta}'}}\|_{\mathrm{TV}}\le L_{\sigma}\|\boldsymbol{\theta}-\boldsymbol{\theta}'\|_2$. Moreover, using the compactness of $\boldsymbol{\Theta}$, there exists $D_{\Theta}<\infty$ such that
\begin{equation}
\|\boldsymbol{\theta}_{\ell+1}-\boldsymbol{\theta}_{\ell}\|_2\le D_{\Theta}\,\beta_{\ell},\qquad \forall \ell.
\label{eq:appC_theta_step}
\end{equation}
Hence for $k\in\{t-T_t+1,\ldots,t\}$,
\begin{equation}
\begin{aligned}
\|\boldsymbol{\theta}_{k}-\boldsymbol{\theta}_{t-1}\|_2
&\le \sum_{\ell=k}^{t-2}\|\boldsymbol{\theta}_{\ell+1}-\boldsymbol{\theta}_{\ell}\|_2\\
&\le D_{\Theta}\sum_{\ell=t-T_t}^{t-1}\beta_{\ell}\\
&\le D_{\Theta}T_t\,\beta_{t-T_t+1}.
\end{aligned}
\label{eq:appC_theta_drift_window}
\end{equation}
The online bias is therefore bounded by
\begin{equation}
|b_{J,\mathrm{on}}^{t-1}|
\le c_{J,1}\Big(\varrho^{\tau_t}+T_t\,\beta_{t-T_t+1}\Big)
\le c_{J,1}\Big(t^{-2}+T_t\,\beta_{t-T_t+1}\Big),
\label{eq:appC_bJon_bound}
\end{equation}
for some constant $c_{J,1}>0$.

For the offline part, with the boundedness of $C_i^{\text{'}}$,
\begin{equation}
\bigg|\mathbb{E}\Big[\hat{\mathbb{E}}_{\mathcal{D}_{\mathrm{off}}^{t}}[C_i^{\text{'}}]\,\Big|\,\mathcal{F}_{t-1}\Big]-J_i(\boldsymbol{\theta}_{t-1})\bigg|
\le 2C_{\max}
\triangleq c_{J,2}.
\label{eq:appC_bJoff_bound}
\end{equation}
Combining \eqref{eq:appC_bJon_bound} and \eqref{eq:appC_bJoff_bound} gives the total bias in \eqref{eq:appC_J_cond_mean}:
\begin{equation}
|o_J^{t-1}|
\le (1-\xi_t)|b_{J,\mathrm{on}}^{t-1}|+\xi_t\,c_{J,2}
\le c_{J,1}\Big(t^{-2}+T_t\,\beta_{t-T_t+1}\Big)+c_{J,2}\,\xi_t.
\label{eq:appC_oJ_bound}
\end{equation}
Equation~\eqref{eq:appC_J_cond_mean} together with \eqref{eq:appC_oJ_bound} verifies condition~(3) of Lemma~\ref{lemma:appC_SA_tracking} for the target sequence $z_{J,i}^{t}\triangleq J_i(\boldsymbol{\theta}_{t})$. Since we take $\mathcal{W}=\mathbb{R}$, condition~(1) is satisfied. Condition~(2) follows from the boundedness of $C_i^{\text{'}}(\boldsymbol{s},\boldsymbol{a})$ and condition~(4) follows Assumption~\ref{assumption:step_size}. For condition~(5), since $J_i(\boldsymbol{\theta})$ is Lipschitz on the compact set $\boldsymbol{\Theta}$, there exists $L_{J,i}>0$ such that
\begin{equation}
\frac{|J_i(\boldsymbol{\theta}_{t})-J_i(\boldsymbol{\theta}_{t-1})|}{\alpha_{t-1}}
\le L_{J,i}\frac{\|\boldsymbol{\theta}_{t}-\boldsymbol{\theta}_{t-1}\|_2}{\alpha_{t-1}}
=\mathcal{O}\!\left(\frac{\beta_{t-1}}{\alpha_{t-1}}\right)\to 0,
\label{eq:appC_J_target_drift}
\end{equation}
where \eqref{eq:appC_theta_step} and Assumption~\ref{assumption:step_size}.4 are used. Applying Lemma~\ref{lemma:appC_SA_tracking} to \eqref{eq:appC_J_recursion} with $w^{t}=\hat J_i^{t}$, $\varrho^{t}=\tilde J_i^{t}$, and $z_{J,i}^{t}=J_i(\boldsymbol{\theta}_{t})$ yields
\begin{equation}
|\hat J_i^{t}-J_i(\boldsymbol{\theta}_{t})|\to 0,\qquad \text{a.s.}
\label{eq:appC_J_consistency}
\end{equation}

\subsection{\texorpdfstring{Asymptotic consistency of $\hat{\boldsymbol{g}}_i^t$ up to $\epsilon_{m_Q}$}{Asymptotic consistency of gradient estimates}}

Define the gradient of the mixed policy  score function as $\boldsymbol{\phi}_{\boldsymbol{\theta}}(\boldsymbol{s},\boldsymbol{a})\triangleq\nabla_{\boldsymbol{\theta}}\log\pi_{\boldsymbol{\theta}}(\boldsymbol{a}\mid\boldsymbol{s})$.
Under Assumption~\ref{assumption:problem_structure} and the mixture-ratio terms are uniformly bounded over $\boldsymbol{\Theta}$, there exists a constant $B_{\pi}<\infty$ such that
\begin{equation}
\sup_{\boldsymbol{\theta}\in\boldsymbol{\Theta}}\ \sup_{(\boldsymbol{s},\boldsymbol{a})\in\mathcal{S}\times\mathcal{A}}
\big\|\boldsymbol{\phi}_{\boldsymbol{\theta}}(\boldsymbol{s},\boldsymbol{a})\big\|_2\le B_{\pi}.
\label{eq:appC_score_bound}
\end{equation}

Define the auxiliary gradient target
\begin{equation}
\nabla_{\boldsymbol{\theta}}\hat{J}_i(\boldsymbol{\theta}_t)
\triangleq
\mathbb{E}_{\sigma_{\pi_{\boldsymbol{\theta}_t}}}\Big[\hat{Q}_i^{\pi_{\boldsymbol{\theta}_t}}(\boldsymbol{s},\boldsymbol{a})\,\boldsymbol{\phi}_{\boldsymbol{\theta}_t}(\boldsymbol{s},\boldsymbol{a})\Big].
\label{eq:appC_aux_grad_def}
\end{equation}
$\hat{\boldsymbol{g}}_i^t$ tracks the moving target $\nabla_{\boldsymbol{\theta}}\hat{J}_i(\boldsymbol{\theta}_t)$ as
\begin{equation}
\hat{\boldsymbol{g}}_i^{t}=\hat{\boldsymbol{g}}_i^{t-1}+\alpha_t\big(\tilde{\boldsymbol{g}}_i^{t}-\hat{\boldsymbol{g}}_i^{t-1}\big).
\label{eq:appC_g_recursion}
\end{equation}

Taking the conditional expectation of \eqref{eq:gtlt} given $\mathcal{F}_{t-1}$ gives
\begin{equation}
\mathbb{E}[\tilde{\boldsymbol{g}}_i^{t}\mid\mathcal{F}_{t-1}]
=\nabla_{\boldsymbol{\theta}}\hat{J}_i(\boldsymbol{\theta}_{t-1})+\boldsymbol{o}_g^{t-1},
\label{eq:appC_g_cond_mean}
\end{equation}
where $\boldsymbol{o}_g^{t-1}$ is the bias term. Splitting it into online and offline components, and adding and subtracting the ideal target $\nabla_{\boldsymbol{\theta}}\hat{J}_i(\boldsymbol{\theta}_{t-1})$, gives
\begin{equation}
\begin{aligned}
\boldsymbol{o}_g^{t-1}
=
&(1-\xi_t)\Big(
\mathbb{E}[\tilde{\boldsymbol{g}}_{i,\mathrm{on}}^t\mid\mathcal{F}_{t-1}]
-\nabla_{\boldsymbol{\theta}}\hat{J}_i(\boldsymbol{\theta}_{t-1})
\Big) \\
&+\xi_t\Big(
\mathbb{E}[\tilde{\boldsymbol{g}}_{i,\mathrm{off}}^t\mid\mathcal{F}_{t-1}]
-\nabla_{\boldsymbol{\theta}}\hat{J}_i(\boldsymbol{\theta}_{t-1})
\Big).
\end{aligned}
\end{equation}
The online part is split into the critic approximation error and the distribution-drift error; the total-variation and critic-tracking bounds used above control these two terms. The offline part is bounded as in \eqref{eq:appC_bJoff_bound}, using the boundedness of the Q-function and the score function. Consequently,
\begin{equation}
\|\boldsymbol{o}_{g}^{t-1}\|_2
\le c_{g,1}\Big(t^{-2}+T_t\,\beta_{t-T_t+1}\Big)
+c_{g,2}\,\xi_t
+ c_{g,3}\,\bar{\epsilon}_{\mathrm{cri}}(t),
\label{eq:appC_og_bound}
\end{equation}
where $\bar{\epsilon}_{\mathrm{cri}}(t)\triangleq\max_{0\le j\le I}\epsilon_{t,j}^{\mathrm{cri}}$ is the aggregated critic tracking error and $c_{g,1},c_{g,2},c_{g,3}>0$ are constants independent of $t$ and $m_Q$.
Combining \eqref{eq:appC_og_bound} with Assumption~\ref{assumption:step_size} gives $\sum_t \alpha_t\|\boldsymbol{o}_g^{t-1}\|_2<\infty$, which verifies condition~(3) of Lemma~\ref{lemma:appC_SA_tracking}. With $\mathcal{W}=\mathbb{R}^{n_{\theta}}$, condition~(1) is immediate. Condition~(2) follows from the boundedness of the score function and the critic outputs, and condition~(4) follows from Assumption~\ref{assumption:step_size}. For condition~(5), $\nabla_{\boldsymbol{\theta}}\hat{J}_i(\boldsymbol{\theta})$ is Lipschitz over the compact set $\boldsymbol{\Theta}$ and \eqref{eq:appC_theta_step} gives $\|\boldsymbol{\theta}_{t}-\boldsymbol{\theta}_{t-1}\|_2=\mathcal{O}(\beta_{t-1})$. Hence $\|\nabla_{\boldsymbol{\theta}}\hat{J}_i(\boldsymbol{\theta}_{t})-\nabla_{\boldsymbol{\theta}}\hat{J}_i(\boldsymbol{\theta}_{t-1})\|_2/\alpha_{t-1}\to 0$ whenever $\beta_t/\alpha_t\to 0$. All conditions of Lemma~\ref{lemma:appC_SA_tracking} are therefore satisfied for the target sequence, and
\begin{equation}
\big\|\hat{\boldsymbol{g}}_i^{t}-\nabla_{\boldsymbol{\theta}}\hat{J}_i(\boldsymbol{\theta}_{t})\big\|_2\to 0,
\qquad \text{a.s.}
\label{eq:appC_g_track_aux}
\end{equation}

The discrepancy between the auxiliary policy-gradient target and the true policy gradient satisfies
\begin{align}
\big\|\nabla_{\boldsymbol{\theta}}\hat{J}_i(\boldsymbol{\theta}_t)-\nabla_{\boldsymbol{\theta}}J_i(\boldsymbol{\theta}_t)\big\|_2
&=\Big\|\mathbb{E}_{\sigma_{\pi_{\boldsymbol{\theta}_t}}}\Big[
\big(\hat{Q}_i^{\pi_{\boldsymbol{\theta}_t}}
-Q_i^{\pi_{\boldsymbol{\theta}_t}}\big)\,
\boldsymbol{\phi}_{\boldsymbol{\theta}_t}\Big]\Big\|_2 \nonumber\\
&\le B_{\pi}\,\mathbb{E}_{\sigma_{\pi_{\boldsymbol{\theta}_t}}}\Big[
\big|\hat{Q}_i^{\pi_{\boldsymbol{\theta}_t}}(\boldsymbol{s},\boldsymbol{a}) \notag\\
&\qquad\qquad
-Q_i^{\pi_{\boldsymbol{\theta}_t}}(\boldsymbol{s},\boldsymbol{a})\big|
\Big].
\label{eq:appC_grad_diff_reduce}
\end{align}
Under the local-linearization analysis of the critic in Section~\ref{sec:proof-of-critic} and Assumption~\ref{assumption:target_Q}, $\mathbb{E}_{\sigma_{\pi_{\boldsymbol{\theta}_t}}}\Big[\big|\hat{Q}_i^{\pi_{\boldsymbol{\theta}_t}}-Q_i^{\pi_{\boldsymbol{\theta}_t}}\big|\Big]$ admits the uniform bound:
\begin{equation}
\mathbb{E}_{\sigma_{\pi_{\boldsymbol{\theta}_t}}}\Big[\big|\hat{Q}_i^{\pi_{\boldsymbol{\theta}_t}}-Q_i^{\pi_{\boldsymbol{\theta}_t}}\big|\Big]
\le c_Q\Big(|\hat{J}_i^t-J_i(\boldsymbol{\theta}_t)|+\epsilon_{m_Q}\Big),
\label{eq:appC_Q_mismatch_bound}
\end{equation}
for some constant $c_Q>0$ independent of $t$ and $m_Q$.
Combining \eqref{eq:appC_grad_diff_reduce}--\eqref{eq:appC_Q_mismatch_bound} and using \eqref{eq:appC_J_consistency} yields
\begin{equation}
\limsup_{t\to\infty}\big\|\nabla_{\boldsymbol{\theta}}\hat{J}_i(\boldsymbol{\theta}_t)-\nabla_{\boldsymbol{\theta}}J_i(\boldsymbol{\theta}_t)\big\|_2
\le c_{\pi}\,\epsilon_{m_Q},
\label{eq:appC_aux_to_true_grad}
\end{equation}
where $c_{\pi}\triangleq B_{\pi}c_Q$.

Using the triangle inequality,
\begin{equation}
\begin{aligned}
\big\|\hat{\boldsymbol{g}}_i^{t}-\nabla_{\boldsymbol{\theta}}J_i(\boldsymbol{\theta}_t)\big\|_2
&\le \big\|\hat{\boldsymbol{g}}_i^{t}-\nabla_{\boldsymbol{\theta}}\hat{J}_i(\boldsymbol{\theta}_t)\big\|_2 \\
&\quad + \big\|\nabla_{\boldsymbol{\theta}}\hat{J}_i(\boldsymbol{\theta}_t)-\nabla_{\boldsymbol{\theta}}J_i(\boldsymbol{\theta}_t)\big\|_2.
\end{aligned}
\label{eq:asm2}
\end{equation}
When $t\to\infty$, the first term vanishes and the second term is bounded by $c_{\pi}\epsilon_{m_Q}$ by \eqref{eq:appC_aux_to_true_grad}.
Combining \eqref{eq:asm2} with \eqref{eq:appC_J_consistency}, we complete the proof of Lemma~\ref{lemma-asymptotic-consistency}.

\section{\texorpdfstring{Proof of Lemma~\ref{lemma-convergence-rate-surrogate}}{Proof of the surrogate convergence-rate lemma}}
\label{sec:proof-of-surrogate}
Since the analysis for $\epsilon_{J}(t)$ and $\epsilon_{g}(t)$ are similar, we only provide the more complicated derivation of the latter due to space limit.

Define the accumulated gradient-tracking error
\begin{equation}
\begin{aligned}
\boldsymbol{y}_k
&\triangleq \hat{\boldsymbol{g}}_i^k
-\nabla_{\boldsymbol{\theta}}J_i(\boldsymbol{\theta}_k), \\
\epsilon_{g,i}(t)
&\triangleq \frac{1}{t+1}\sum_{k=0}^{t}
\mathbb{E}\big[\|\boldsymbol{y}_k\|_2^2\big].
\end{aligned}
\label{eq:appD_epsgi_def}
\end{equation}
Let
\begin{equation}
 \boldsymbol{\Delta}_k\triangleq \nabla_{\boldsymbol{\theta}}J_i(\boldsymbol{\theta}_k)-\nabla_{\boldsymbol{\theta}}J_i(\boldsymbol{\theta}_{k+1})
\label{eq:appD_Deltag_def}
\end{equation}
and define the cross term
\begin{equation}
 \Xi_{k+1}\triangleq
 \big\langle
 \boldsymbol{y}_k,
 \tilde{\boldsymbol{g}}_i^{k+1}-\nabla_{\boldsymbol{\theta}}J_i(\boldsymbol{\theta}_k)
 \big\rangle.
\label{eq:appD_Xi_def}
\end{equation}
Equation~\eqref{eq:g-average} gives the recursion
\begin{align}
 \boldsymbol{y}_{k+1}
 &=\hat{\boldsymbol{g}}_i^{k+1}-\nabla_{\boldsymbol{\theta}}J_i(\boldsymbol{\theta}_{k+1}) \notag\\
 &=\boldsymbol{y}_k+\boldsymbol{\Delta}_k+\alpha_{k+1}\big(\tilde{\boldsymbol{g}}_i^{k+1}-\hat{\boldsymbol{g}}_i^k\big).
\label{eq:appD_y_recursion}
\end{align}
Using $2\langle a,b\rangle\le \|a\|_2^2+\|b\|_2^2$ and noting that
\begin{equation}
 \big\langle \boldsymbol{y}_k,\tilde{\boldsymbol{g}}_i^{k+1}-\hat{\boldsymbol{g}}_i^k\big\rangle
 =\Xi_{k+1}-\|\boldsymbol{y}_k\|_2^2,
\end{equation}
this gives
\begin{equation}
\begin{aligned}
 \|\boldsymbol{y}_{k+1}\|_2^2
 &\le (1-2\alpha_{k+1})\|\boldsymbol{y}_k\|_2^2
 +2\alpha_{k+1}\Xi_{k+1} \\
 &\quad +2\langle \boldsymbol{y}_k,\boldsymbol{\Delta}_k\rangle
 +2\|\boldsymbol{\Delta}_k\|_2^2 \\
 &\quad +2\alpha_{k+1}^2
 \big\|\tilde{\boldsymbol{g}}_i^{k+1}-\hat{\boldsymbol{g}}_i^k\big\|_2^2.
\end{aligned}
\label{eq:appD_y_square}
\end{equation}
It can be decomposed into five parts:
\begin{align}
 \epsilon_{g,i}(t)
 &\le
 \underbrace{\frac{1}{t+1}\sum_{k=0}^{t}\frac{\mathbb{E}[\|\boldsymbol{y}_k\|_2^2-\|\boldsymbol{y}_{k+1}\|_2^2]}{2\alpha_{k+1}}}_{M_1(t)} \notag\\
 &\quad +\underbrace{\frac{1}{t+1}\sum_{k=0}^{t}\frac{1}{\alpha_{k+1}}\mathbb{E}[\|\boldsymbol{\Delta}_k\|_2^2]}_{M_2(t)} \notag\\
 &\quad +\underbrace{\frac{1}{t+1}\sum_{k=0}^{t}\alpha_{k+1}\mathbb{E}\Big[\big\|\tilde{\boldsymbol{g}}_i^{k+1}-\hat{\boldsymbol{g}}_i^k\big\|_2^2\Big]}_{M_3(t)} \notag\\
 &\quad +\underbrace{\frac{1}{t+1}\sum_{k=0}^{t}\mathbb{E}[\Xi_{k+1}]}_{M_4(t)} \notag\\
 &\quad +\underbrace{\frac{1}{t+1}\sum_{k=0}^{t}\frac{1}{\alpha_{k+1}}\mathbb{E}[\langle \boldsymbol{y}_k,\boldsymbol{\Delta}_k\rangle]}_{M_5(t)}.
\label{eq:appD_gradient_decomp}
\end{align}
In the following, we bound five terms separately.

For $M_1(t)$, let $Y_k\triangleq\mathbb{E}[\|\boldsymbol{y}_k\|_2^2]$. Then
\begin{equation}
\begin{aligned}
\sum_{k=0}^{t}
\frac{Y_k-Y_{k+1}}{2\alpha_{k+1}}
={}& \frac{Y_0}{2\alpha_1}
-\frac{Y_{t+1}}{2\alpha_{t+1}} \\
&\quad +\sum_{k=1}^{t}\left(
\frac{1}{2\alpha_{k+1}}-\frac{1}{2\alpha_k}\right)
Y_k.
\end{aligned}
\end{equation}
since $\alpha_k$ is non-increasing and $\mathbb{E}[\|\boldsymbol{y}_k\|_2^2]$ is uniformly bounded, we have
\begin{equation}
 M_1(t)=\mathcal{O}\!\Big(\frac{1}{(t+1)\alpha_{t+1}}\Big),
\label{eq:appD_M1_bound}
\end{equation}

For $M_2(t)$, $\nabla_{\boldsymbol{\theta}}J_i(\cdot)$ is Lipschitz on the compact set $\boldsymbol{\Theta}$, there exists $L_i>0$ such that

\begin{equation}
\|\boldsymbol{\Delta}_k\|_2
=\big\|\nabla_{\boldsymbol{\theta}}J_i(\boldsymbol{\theta}_k)-\nabla_{\boldsymbol{\theta}}J_i(\boldsymbol{\theta}_{k+1})\big\|_2
\le L_i\|\boldsymbol{\theta}_{k+1}-\boldsymbol{\theta}_k\|_2.
\end{equation}
The actor update \eqref{eq:theta update} gives
$
\|\boldsymbol{\theta}_{k+1}-\boldsymbol{\theta}_k\|_2
=\beta_k\|\bar{\boldsymbol{\theta}}_k-\boldsymbol{\theta}_k\|_2
\le c\beta_k,
$
with the compact set $\boldsymbol{\Theta}$. Hence
\begin{equation}
\mathbb{E}[\|\boldsymbol{\Delta}_k\|_2^2]\le c\beta_k^2,
\end{equation}
which implies
\begin{equation}
M_2(t)=\mathcal{O}\!\Bigg(\frac{1}{t+1}\sum_{k=0}^{t}\frac{\beta_k^2}{\alpha_{k+1}}\Bigg).
\label{eq:appD_M2_bound}
\end{equation}

For $M_3(t)$, $\tilde{\boldsymbol{g}}_i^{k+1}$ is formed from the bounded score function and the bounded critic output. There exists a constant $G_2<\infty$ such that
\begin{equation}
\sup_k\mathbb{E}\big[\|\tilde{\boldsymbol{g}}_i^{k+1}\|_2^2\big]\le G_2.
\end{equation}
The inequality $\|a-b\|_2^2\le 2\|a\|_2^2+2\|b\|_2^2$ gives
\begin{equation}
\mathbb{E}\Big[\big\|\tilde{\boldsymbol{g}}_i^{k+1}-\hat{\boldsymbol{g}}_i^k\big\|_2^2\Big]\le c,
\end{equation}
and therefore
\begin{equation}
M_3(t)=\mathcal{O}\!\Bigg(\frac{1}{t+1}\sum_{k=0}^{t}\alpha_{k+1}\Bigg).
\label{eq:appD_M3_bound}
\end{equation}

For $M_4(t)$, according to Section~\ref{sec:proof-of-asymptotic-consistency}, the conditional mean of the instantaneous gradient estimate satisfies
\begin{equation}
\mathbb{E}[\tilde{\boldsymbol{g}}_i^{k+1}\mid \mathcal{F}_k]
=\nabla_{\boldsymbol{\theta}}\hat{J}_i(\boldsymbol{\theta}_k)+\boldsymbol{o}_{g,i}^k,
\label{eq:appD_conditional_mean}
\end{equation}
where
\begin{equation}
\|\boldsymbol{o}_{g,i}^k\|_2
\le c\Big(\bar{\epsilon}_{\mathrm{cri}}(k)+(k+1)^{-2}+T_k\beta_{k-T_k+1}+\xi_{k+1}\Big).
\label{eq:appD_og_bound}
\end{equation}
Here $\bar{\epsilon}_{\mathrm{cri}}(k)\triangleq \max_{0\le j\le I}\epsilon_{k,j}^{\mathrm{cri}}$.

Then $\mathbb{E}[\Xi_{k+1}]$ can be expanded as:
\begin{align}
\mathbb{E}[\Xi_{k+1}]
&=\mathbb{E}\Big[\big\langle \boldsymbol{y}_k,\mathbb{E}[\tilde{\boldsymbol{g}}_i^{k+1}\mid \mathcal{F}_k]-\nabla_{\boldsymbol{\theta}}J_i(\boldsymbol{\theta}_k)\big\rangle\Big] \notag\\
&=\mathbb{E}\Big[\big\langle \boldsymbol{y}_k,\boldsymbol{o}_{g,i}^k\big\rangle\Big]
+\mathbb{E}\Big[\big\langle \boldsymbol{y}_k,\nabla_{\boldsymbol{\theta}}\hat{J}_i(\boldsymbol{\theta}_k)-\nabla_{\boldsymbol{\theta}}J_i(\boldsymbol{\theta}_k)\big\rangle\Big].
\label{eq:appD_Xi_expand}
\end{align}
We bound the two terms on the right-hand side separately.

For the first term, $\|\boldsymbol{y}_k\|_2$ is uniformly bounded in expectation because both $\hat{\boldsymbol{g}}_i^k$ and $\nabla_{\boldsymbol{\theta}}J_i(\boldsymbol{\theta}_k)$ are uniformly bounded. Hence, by \eqref{eq:appD_og_bound},
\begin{equation}
\begin{aligned}
\left|\mathbb{E}\Big[\big\langle \boldsymbol{y}_k,\boldsymbol{o}_{g,i}^k\big\rangle\Big]\right|
&\le c\,\mathbb{E}[\|\boldsymbol{o}_{g,i}^k\|_2]\\
&\le c\Big(\bar{\epsilon}_{\mathrm{cri}}(k)+(k+1)^{-2}\\
&\quad+T_k\beta_{k-T_k+1}+\xi_{k+1}\Big).
\end{aligned}
\label{eq:appD_Xi_first}
\end{equation}

Second, by \eqref{eq:appC_grad_diff_reduce}--\eqref{eq:appC_Q_mismatch_bound},
\begin{equation}
\big\|\nabla_{\boldsymbol{\theta}}\hat{J}_i(\boldsymbol{\theta}_k)-\nabla_{\boldsymbol{\theta}}J_i(\boldsymbol{\theta}_k)\big\|_2
\le c\Big(\big|\hat{J}_i^k-J_i(\boldsymbol{\theta}_k)\big|+\epsilon_{m_Q}\Big).
\label{eq:appD_gradgap}
\end{equation}
Again using the boundedness of $\|\boldsymbol{y}_k\|_2$,
\begin{align}
&\left|\mathbb{E}\Big[\big\langle \boldsymbol{y}_k,
\nabla_{\boldsymbol{\theta}}\hat{J}_i(\boldsymbol{\theta}_k)
-\nabla_{\boldsymbol{\theta}}J_i(\boldsymbol{\theta}_k)
\big\rangle\Big]\right| \notag\\
&\qquad\le c\,\mathbb{E}\Big[\big|\hat{J}_i^k-J_i(\boldsymbol{\theta}_k)\big|\Big]
+c\epsilon_{m_Q}.
\label{eq:appD_Xi_second_pre}
\end{align}
Cauchy--Schwarz gives
\begin{align}
\frac{1}{t+1}\sum_{k=0}^{t}\mathbb{E}\Big[\big|\hat{J}_i^k-J_i(\boldsymbol{\theta}_k)\big|\Big]
&\le \sqrt{\frac{1}{t+1}\sum_{k=0}^{t}\mathbb{E}\Big[\big|\hat{J}_i^k-J_i(\boldsymbol{\theta}_k)\big|^2\Big]} \notag\\
&\le \sqrt{\epsilon_J(t)}.
\label{eq:appD_abs_to_epsJ}
\end{align}
Combining \eqref{eq:appD_Xi_expand}--\eqref{eq:appD_abs_to_epsJ}, we conclude that
\begin{align}
M_4(t)
&\le \mathcal{O}\!\Bigg(\frac{1}{t+1}\sum_{k=0}^{t}\Big(
\bar{\epsilon}_{\mathrm{cri}}(k)+(k+1)^{-2} \notag\\
&\qquad\qquad\qquad
+T_k\beta_{k-T_k+1}+\xi_{k+1}\Big)\Bigg) \notag\\
&\quad +\mathcal{O}\!\Big(\sqrt{\epsilon_J(t)}+\epsilon_{m_Q}\Big).
\label{eq:appD_M4_bound}
\end{align}

For $M_5(t)$, utilizing Cauchy-Schwarz inequality, we have:
\begin{align}
|M_5(t)|
&\le \frac{1}{t+1}\sum_{k=0}^{t}\frac{1}{\alpha_{k+1}}\mathbb{E}\big[\|\boldsymbol{y}_k\|_2\,\|\boldsymbol{\Delta}_k\|_2\big] \notag\\
&\le \sqrt{\frac{1}{t+1}\sum_{k=0}^{t}\mathbb{E}[\|\boldsymbol{y}_k\|_2^2]}
\sqrt{\frac{1}{t+1}\sum_{k=0}^{t}\frac{\mathbb{E}[\|\boldsymbol{\Delta}_k\|_2^2]}{\alpha_{k+1}^2}}.
\label{eq:appD_M5_pre}
\end{align}
With the bound on $\mathbb{E}[\|\boldsymbol{\Delta}_k\|_2^2]$ derived in $M_2(t)$,
\begin{equation}
\frac{1}{t+1}\sum_{k=0}^{t}\frac{\mathbb{E}[\|\boldsymbol{\Delta}_k\|_2^2]}{\alpha_{k+1}^2}
\le c\,\frac{1}{t+1}\sum_{k=0}^{t}\frac{\beta_k^2}{\alpha_{k+1}^2}.
\end{equation}
This gives
\begin{equation}
\begin{aligned}
M_5(t)
&=\mathcal{O}\!\Big(\sqrt{\epsilon_{g,i}(t)}\,\sqrt{N(t)}\Big),\\
N(t)
&\triangleq \frac{1}{t+1}\sum_{k=0}^{t}\frac{\beta_k^2}{\alpha_{k+1}^2}.
\end{aligned}
\label{eq:appD_M5_bound}
\end{equation}

Substituting \eqref{eq:appD_M1_bound}, \eqref{eq:appD_M2_bound}, \eqref{eq:appD_M3_bound}, \eqref{eq:appD_M4_bound}, and \eqref{eq:appD_M5_bound} into \eqref{eq:appD_gradient_decomp} gives
\begin{equation}
\begin{aligned}
\epsilon_{g,i}(t)
&\le \mathcal{O}\!\Big(\sqrt{\epsilon_{g,i}(t)}\,\sqrt{N(t)}\Big)
\\
&\quad+\mathcal{O}\!\Big(\frac{1}{(t+1)\alpha_{t+1}}\Big)\\
&\quad +\mathcal{O}\!\Bigg(\frac{1}{t+1}\sum_{k=0}^{t}\Big(
\bar{\epsilon}_{\mathrm{cri}}(k)+(k+1)^{-2}\\
&\qquad\qquad
+T_k\beta_{k-T_k+1}+\alpha_{k+1}
\Big)\Bigg)\\
&\quad +\mathcal{O}\!\Bigg(\frac{1}{t+1}\sum_{k=0}^{t}
\frac{\beta_k^2}{\alpha_{k+1}}\Bigg)
\\
&\quad+\mathcal{O}\!\Bigg(\frac{1}{t+1}\sum_{k=0}^{t}\xi_{k+1}\Bigg)\\
&\quad+\mathcal{O}\!\Big(\sqrt{\epsilon_J(t)}+\epsilon_{m_Q}\Big).
\end{aligned}
\label{eq:appD_epsgi_pre_solve}
\end{equation}
Let $F_i(t)\triangleq \epsilon_{g,i}(t)$ and collect all terms other than $\sqrt{F_i(t)}\sqrt{N(t)}$ into
\begin{equation}
\begin{aligned}
R(t)
&\triangleq \mathcal{O}\!\Big(\frac{1}{(t+1)\alpha_{t+1}}\Big)
+\mathcal{O}\!\Bigg(\frac{1}{t+1}\sum_{k=0}^{t}\Big(
\bar{\epsilon}_{\mathrm{cri}}(k)+(k+1)^{-2}\\&+T_k\beta_{k-T_k+1}
+\alpha_{k+1}\Big)\Bigg)+\mathcal{O}\!\Bigg(\frac{1}{t+1}\sum_{k=0}^{t}
\frac{\beta_k^2}{\alpha_{k+1}}\Bigg)
\\&\quad+\mathcal{O}\!\Bigg(\frac{1}{t+1}\sum_{k=0}^{t}\xi_{k+1}\Bigg)
+\mathcal{O}\!\Big(\sqrt{\epsilon_J(t)}+\epsilon_{m_Q}\Big).
\end{aligned}
\label{eq:appD_R_def}
\end{equation}
Thus \eqref{eq:appD_epsgi_pre_solve} takes the form
\begin{equation}
F_i(t)\le c_1\sqrt{F_i(t)}\sqrt{N(t)}+c_2R(t)
\label{eq:appD_quadratic_ineq}
\end{equation}
for some constants $c_1,c_2>0$. Set $u\triangleq \sqrt{F_i(t)}\ge 0$. Solving this quadratic inequality yields
\begin{equation}
\begin{aligned}
u
&\le c_1\sqrt{N(t)}+\sqrt{c_2R(t)},\\
\epsilon_{g,i}(t)
&=\mathcal{O}\!\big(N(t)+R(t)\big).
\end{aligned}
\label{eq:appD_epsgi_intermediate}
\end{equation}
Substituting the definitions of $N(t)$ and $R(t)$ into \eqref{eq:appD_epsgi_intermediate} yields

\begin{equation}
\begin{aligned}
\epsilon_{g,i}(t)
&\le \mathcal{O}\!\Big(\frac{1}{(t+1)\alpha_{t+1}}\Big)\\
&\quad +\mathcal{O}\!\Bigg(\frac{1}{t+1}\sum_{k=0}^{t}\Big(
\bar{\epsilon}_{\mathrm{cri}}(k)+(k+1)^{-1}+T_k\beta_{k-T_k+1}\Big)\Bigg)\\
&\quad +\mathcal{O}\!\Bigg(\frac{1}{t+1}\sum_{k=0}^{t}\Big(
\beta_k^2(\alpha_{k+1}^{-2}+\alpha_{k+1}^{-1})+\alpha_{k+1}\Big)\Bigg)\\
&\quad +\mathcal{O}\!\Bigg(\frac{1}{t+1}\sum_{k=0}^{t}\xi_{k+1}\Bigg)
{}+\mathcal{O}\!\Big(\sqrt{\epsilon_J(t)}+\epsilon_{m_Q}\Big).
\end{aligned}
\label{eq:appD_epsgi_final_fixed_i}
\end{equation}

The accumulated gradient-tracking error $\epsilon_g(t)$ is bounded by the sum over the componentwise errors:
\begin{equation}
\epsilon_g(t)=\frac{1}{t+1}\sum_{k=0}^{t}\max_{0\le j\le I}\mathbb{E}\Big[\big\|\hat{\boldsymbol{g}}_j^k-\nabla_{\boldsymbol{\theta}}J_j(\boldsymbol{\theta}_k)\big\|_2^2\Big]
\le \sum_{j=0}^{I}\epsilon_{g,j}(t).
\end{equation}
The proof of Lemma~\ref{lemma-convergence-rate-surrogate} follows.

\bibliographystyle{IEEEtran}
\bibliography{RLreferences}

@article{wang2024sldac,
  title={Single-loop deep actor-critic for constrained reinforcement learning with provable convergence},
  author={Wang, Kexuan and Liu, An and Lin, Baishuo},
  journal={IEEE Transactions on Signal Processing},
  volume={72},
  pages={4871--4887},
  year={2024},
  publisher={IEEE}
}

@inproceedings{zhang2025prcrl,
  title={A Policy Reuse Reinforcement Learning Framework for Hard Latency Constrained Resource Scheduling},
  author={Zhang, Luyuan and Liu, An},
  booktitle={2025 IEEE Wireless Communications and Networking Conference (WCNC)},
  pages={1--6},
  year={2025},
  organization={IEEE}
}

@article{zhang2025hrl,
  title={A hybrid reinforcement learning framework for hard latency constrained resource scheduling},
  author={Zhang, Luyuan and Liu, An and Wang, Kexuan},
  journal={IEEE Internet of Things Journal},
  year={2025},
  publisher={IEEE}
}

@article{stolyar2001largestLWDF,
  title={Largest weighted delay first scheduling: Large deviations and optimality},
  author={Stolyar, Alexander L and Ramanan, Kavita},
  journal={Annals of Applied Probability},
  pages={1--48},
  year={2001},
  publisher={JSTOR}
}

@article{tassiulas1992stability,
  title={Stability properties of constrained queueing systems and scheduling policies for maximum throughput in multihop radio networks},
  author={Tassiulas, Leandros and Ephremides, Anthony},
  journal={IEEE Transactions on Automatic Control},
  volume={37},
  number={12},
  pages={1936--1948},
  year={1992},
  doi={10.1109/9.182479},
  publisher={IEEE}
}

@book{neely2010stochastic,
  title={Stochastic Network Optimization with Application to Communication and Queueing Systems},
  author={Neely, Michael J.},
  series={Synthesis Lectures on Learning, Networks, and Algorithms},
  publisher={Springer, Cham},
  year={2010},
  doi={10.1007/978-3-031-79995-2}
}

@article{shi2011wmmse,
  title={An Iteratively Weighted {MMSE} Approach to Distributed Sum-Utility Maximization for a {MIMO} Interfering Broadcast Channel},
  author={Shi, Qingjiang and Razaviyayn, Meisam and Luo, Zhi-Quan and He, Chen},
  journal={IEEE Transactions on Signal Processing},
  volume={59},
  number={9},
  pages={4331--4340},
  year={2011},
  doi={10.1109/TSP.2011.2147784},
  publisher={IEEE}
}

@article{chen2020RL4RRM,
  title={Age of information aware radio resource management in vehicular networks: A proactive deep reinforcement learning perspective},
  author={Chen, Xianfu and Wu, Celimuge and Chen, Tao and Zhang, Honggang and Liu, Zhi and Zhang, Yan and Bennis, Mehdi},
  journal={IEEE Transactions on wireless communications},
  volume={19},
  number={4},
  pages={2268--2281},
  year={2020},
  publisher={IEEE}
}

@article{zangooei2023RL4RRMsurvey,
  title={Reinforcement learning for radio resource management in RAN slicing: A survey},
  author={Zangooei, Mohammad and Saha, Niloy and Golkarifard, Morteza and Boutaba, Raouf},
  journal={IEEE Communications Magazine},
  volume={61},
  number={2},
  pages={118--124},
  year={2023},
  publisher={IEEE}
}

@inproceedings{fernandez2006policyreuse,
  title={Probabilistic policy reuse in a reinforcement learning agent},
  author={Fern{\'a}ndez, Fernando and Veloso, Manuela},
  booktitle={Proceedings of the fifth international joint conference on Autonomous agents and multiagent systems},
  pages={720--727},
  year={2006}
}

@article{fernandez2010policyreuse,
  title={Probabilistic policy reuse for inter-task transfer learning},
  author={Fern{\'a}ndez, Fernando and Garc{\'\i}a, Javier and Veloso, Manuela},
  journal={Robotics and Autonomous Systems},
  volume={58},
  number={7},
  pages={866--871},
  year={2010},
  publisher={Elsevier}
}

@article{taylor2009transferlearning,
  title={Transfer learning for reinforcement learning domains: A survey.},
  author={Taylor, Matthew E and Stone, Peter},
  journal={Journal of Machine Learning Research},
  volume={10},
  number={7},
  year={2009}
}

@book{CMDP,
  title={Constrained Markov decision processes},
  author={Altman, Eitan},
  volume={7},
  year={1999},
  publisher={CRC press}
}

@article{SCAOPO,
  title={Successive convex approximation based off-policy optimization for constrained reinforcement learning},
  author={Tian, Chang and Liu, An and Huang, Guan and Luo, Wu},
  journal={IEEE Transactions on Signal Processing},
  volume={70},
  pages={1609--1624},
  year={2022},
  publisher={IEEE}
}

@inproceedings{CPO,
  title={Constrained policy optimization},
  author={Achiam, Joshua and Held, David and Tamar, Aviv and Abbeel, Pieter},
  booktitle={ICML},
  pages={22--31},
  year={2017},
  organization={PMLR}
}

@ARTICLE{CSSCA,
  author={Liu, An and Lau, Vincent K. N. and Kananian, Borna},
  journal={IEEE Trans. Signal Process.},
  title={Stochastic Successive Convex Approximation for Non-Convex Constrained Stochastic Optimization},
  year={2019},
  volume={67},
  number={16},
  pages={4189-4203},
  doi={10.1109/TSP.2019.2925601}}

@article{Lemma3,
  title={Feasible direction methods for stochastic programming problems},
  author={Ruszczy{\'n}ski, Andrzej},
  journal={Mathematical Programming},
  number = {1},
  volume={19},
  pages={220--229},
  year={1980},
  publisher={Springer}}

@inproceedings{DQlearning,
  title={A finite-time analysis of {Q}-learning with neural network function approximation},
  author={Xu, Pan and Gu, Quanquan},
  booktitle={ICML},
  pages={10555--10565},
  year={2020},
  organization={PMLR}
}

@article{Cao2019b,
  title={Generalization bounds of stochastic gradient descent for wide and deep neural networks},
  author={Cao, Yuan and Gu, Quanquan},
  journal={Proc. Adv. Neural Inf. Process. Syst.},
  volume={32},
  pages = {10835--10845},
  year={2019}
}

@inproceedings{Allen2019b,
  title={A convergence theory for deep learning via over-parameterization},
  author={Allen-Zhu, Zeyuan and Li, Yuanzhi and Song, Zhao},
  booktitle={ICML},
  pages={242--252},
  year={2019},
  organization={PMLR}
}

@book{Gaussianpolicy,
  title={Reinforcement learning: An introduction},
  author={Sutton, Richard S and Barto, Andrew G},
  year={2018},
  publisher={MIT press}
}

@phdthesis{Slater,
  title={Successive convex approximation: Analysis and applications},
  author={Razaviyayn, Meisam},
  year={2014},
  school={University of Minnesota}
}

@article{linnerAC,
  title={On finite-time convergence of actor-critic algorithm},
  author={Qiu, Shuang and Yang, Zhuoran and Ye, Jieping and Wang, Zhaoran},
  journal={IEEE J. Sel. Areas Inf. Theory},
  volume={2},
  number={2},
  pages={652--664},
  year={2021},
  publisher={IEEE}
}

@article{tsitsiklis1999avgtd,
  title={Average cost temporal-difference learning},
  author={Tsitsiklis, John N and Van Roy, Benjamin},
  journal={Automatica},
  volume={35},
  number={11},
  pages={1799--1808},
  year={1999},
  publisher={Elsevier}
}

@inproceedings{Assumption31,
  title={A theoretical analysis of deep {Q}-learning},
  author={Fan, Jianqing and Wang, Zhaoran and Xie, Yuchen and Yang, Zhuoran},
  booktitle={Proc. Learn. Dyn. Control},
  pages={486--489},
  year={2020},
  organization={PMLR}
}

@inproceedings{Assumption32,
  title={Boosted fitted {Q}-iteration},
  author={Tosatto, Samuele and Pirotta, Matteo and d'Eramo, Carlo and Restelli, Marcello},
  booktitle={ICML},
  pages={3434--3443},
  year={2017},
  organization={PMLR}
}

@article{PPOLagTRPOLag,
  title={Benchmarking safe exploration in deep reinforcement learning},
  author={Ray, Alex and Achiam, Joshua and Amodei, Dario},
  journal={arXiv preprint arXiv:1910.01708},
  volume={7},
  number={1},
  pages={2},
  year={2019}
}

@article{RZF,
  title={A vector-perturbation technique for near-capacity multiantenna multiuser communication-part {I}: channel inversion and regularization},
  author={Peel, Christian B and Hochwald, Bertrand M and Swindlehurst, A Lee},
  journal={IEEE Trans. Commun.},
  volume={53},
  number={1},
  pages={195--202},
  year={2005},
  publisher={IEEE}
}

@article{TwoTimescaleAC,
  title={A finite-time analysis of two time-scale actor-critic methods},
  author={Wu, Yue Frank and Zhang, Weitong and Xu, Pan and Gu, Quanquan},
  journal={Advances in Neural Information Processing Systems},
  volume={33},
  pages={17617--17628},
  year={2020}
}

@inproceedings{average8,
  title={On-policy deep reinforcement learning for the average-reward criterion},
  author={Zhang, Yiming and Ross, Keith W},
  booktitle={International Conference on Machine Learning},
  pages={12535--12545},
  year={2021},
  organization={PMLR}
}

@inproceedings{average1,
  title={Model-free reinforcement learning in infinite-horizon average-reward markov decision processes},
  author={Wei, Chen-Yu and Jahromi, Mehdi Jafarnia and Luo, Haipeng and Sharma, Hiteshi and Jain, Rahul},
  booktitle={International conference on machine learning},
  pages={10170--10180},
  year={2020},
  organization={PMLR}
}

@ARTICLE{huang2024beamformingisac,
  author={Huang, Zhe and Liu, An},
  journal={Mobile Communications},
  title={Beamforming Optimization for Integrated Sensing and Communication Systems: A Deep Reinforcement Learning Approach},
  year={2024},
  volume={48},
  number={10},
  pages={41-48},
  doi={10.3969/j.issn.1006-1010.20241013-0003}}

\end{document}